\documentclass{article}
\usepackage[dvipsnames]{xcolor}
\usepackage[utf8]{inputenc}
\usepackage{tikz}
\usetikzlibrary{shapes.geometric}
\usepackage{pgfplots}
\usepackage{graphicx}
\usepackage{imakeidx}
\usepackage{comment}
\usepackage{framed}
\usepackage{subcaption}
\usepackage{amsmath,amsfonts,amssymb,amsthm}
\usepackage{mathtools}
\usepackage{commath}
\usepackage[sc,osf]{mathpazo}
\usepackage[margin=1.5in]{geometry}
\usepackage[math]{cellspace}
\usepackage{tocloft}
\usepackage{tabularx,tabulary,booktabs}
\usepackage{lipsum,multicol}
\usepackage[colorlinks,citecolor=red,urlcolor=blue,bookmarks=false,hypertexnames=true]{hyperref} 
\usepackage[noabbrev,capitalize]{cleveref}
\usepackage{comment}
\usepackage{hypernat}
\usepackage{appendix}
\usepackage{authblk}
\usepackage{csquotes}
\usepackage{nicefrac}

\usepackage[style=authoryear,sorting=nyt]{biblatex}
\newtheorem*{lemma}{Lemma}
\newtheorem{assumption}{Assumption}
\newtheorem{proposition}{Proposition}

\title{The Impact of Geopolitical Conflicts on Trade, Growth, and Innovation}
\author[1]{Carlos G\'oes}
\author[2]{Eddy Bekkers\thanks{The opinions expressed in this article should be attributed to its authors. They are not meant to represent the positions or opinions of the WTO and its Members and are without prejudice to Members’ rights and obligations under the WTO. We are grateful to Bob Koopman for encouraging us to undertake a research project on global economic decoupling. We benefited from helpful discussions with Ralph Ossa, Marc Bachetta, Ezra Oberfield, Ana Maria Santacreu, Marc Muendler, Fabian Trottner, Valerie Ramey, Johannes Wieland, Kyle Handley,  as well as with seminar participants at the European Trade Study Group (ETSG), Washington University at Saint Louis, UC San Diego, IBRE-FGV, Global Trade Analysis Project (GTAP), and the joint World Bank-IMF-WTO Trade Workshop. All errors and omissions are solely our responsibility. Góes: cgoes@ucsd.edu. Bekkers: eddy.bekkers@wto.org.}}
\affil[1]{UC San Diego}
\affil[2]{World Trade Organization}

\date{\today}

\begin{document}

\maketitle
\begin{abstract}
Geopolitical conflicts have increasingly been a driver of trade policy. We study the potential effects of global and persistent geopolitical conflicts on trade, technological innovation, and economic growth.  In conventional trade models the welfare costs of such conflicts are modest. We build a multi-sector multi-region general equilibrium model with dynamic sector-specific knowledge diffusion, which magnifies welfare losses of trade conflicts. Idea diffusion is mediated by the input-output structure of production, such that both sector cost shares and import trade shares characterize the source distribution of ideas. Using this framework, we explore the potential impact of a ``decoupling of the global economy,'' a hypothetical scenario under which technology systems would diverge in the global economy. We divide the global economy into two geopolitical blocs \textemdash East and West \textemdash based on foreign policy similarity and model decoupling through an increase in iceberg trade costs (full decoupling) or tariffs (tariff decoupling). Results yield three main insights. First, the projected welfare losses for the global economy of a decoupling scenario can be drastic, as large as 12\% in some regions. They are the largest in the lower-income regions as they would benefit less from technology spillovers from richer areas. Second, the described size and pattern of welfare effects are specific to the model with the diffusion of ideas. Without diffusion of ideas, the size and variation across regions of the welfare losses would be substantially smaller. Third, a multi-sector framework exacerbates diffusion inefficiencies induced by trade costs relative to a single-sector one.
\end{abstract}



\section{Introduction}

Open markets and free trade have been a basic tenet of the international order
emerging out of World War II. Over that period, a large consensus regarding
the need to reduce trade costs and prioritize gains from trade led to a
continuous deepening of the international trade regime.

The European Union enlarged eastwards and many countries joined the World
Trade Organization (WTO), including Russia, the center-of-gravity of the
former Soviet bloc. At the same time, China displayed spectacular growth
rates, integrated to world markets, and was recognized as a market economy by
also joining the international trade system.

However, the last decade has witnessed the beginning of a backlash against
global trade integration. Political scientists conjecture that the emergence
of China as a new superpower against the incumbent U.S. might lead to strategic
competition between these countries \textemdash one in which geopolitical
forces and the desire to limit interdependence take primacy over win-win
international cooperation\footnote{See \textcite{wei_towards_2019} and
\textcite{wyne_how_2020} for a review of the debate among respectively Chinese
and American scholars about the shift in foreign policies toward each
other.}. Rising support for populist and isolationist parties in many Western
countries points towards the same direction\footnote{For evidence of the
impact of trade shocks on the rise of populist parties to power, see
\textcite{colantone_trade_2018}.}. Additionally, the 2022 War in Ukraine and
the subsequent strong retaliation of the European Union, the United States,
and their allies against Russia suggest that the international economic order
based on open markets and expanded globalization could be replaced by a more
fragmented international economic system.

Using these facts as motivation, this paper aims to determine the potential
effects of increased and persistent large-scale geopolitical conflicts on
trade, economic growth, and technological innovation. Some of the adverse
effects are well-known. Increased trade barriers decrease domestic welfare and
gains from trade by shifting production away from the most cost-efficient
producers and leaving households with a lower level of total consumption.
Canonical trade models capture this \textit{static result} through the fact
that welfare is proportional to the degree of trade openness
\parencite{arkolakis_new_2012}\footnote{This is true of all canonical trade
models. As \textcite{arkolakis_new_2012} show,
\textcite{armington_theory_1969}, \textcite{krugman_scale_1980},
\textcite{eaton_technology_2002}, and \textcite{melitz_impact_2003}, albeit
different in motivation, are isomorphic and summarize gains from trade by some
variation of the expression $\mathbb{G}\propto(\pi_{ii})^{\frac{1}%
{\varepsilon}}$, where $\pi_{ii}$ is domestic trade share and $\varepsilon$ is
the elasticity of trade flows with respect to trade costs.}.

However, some of the main concerns of policymakers and practitioners regarding
potentially detrimental effects of limiting trade are abstracted away in
standard models. For instance, these models typically assume a fixed
technology distribution for domestic firms, thereby limiting gains from trade
to static gains. This assumption renders it impossible to address one of the
most important long-term consequences of continued geopolitical conflicts or
receding globalization \textemdash namely, reduced technology and know-how
spillovers that happen through trade.

In order to realistically assess the impact of trade conflicts on global
innovation, we build a multi-sector multi-region general equilibrium model
with dynamic sector-specific knowledge diffusion. As in
\textcite{buera_global_2020}, we model the arrival of new ideas as a learning
process of producers in a certain country-sector. By engaging in international
trade, i.e. importing intermediate inputs, domestic innovators have access to
new sources of ideas, whose quality depends on the productivity of the
countries and sectors from which they source intermediates.

In our model, idea diffusion is mediated by the input-output structure of
production, such that both sectoral intermediate input cost shares and import
trade shares characterize the source distribution of ideas. Innovation is
summarized by describing productivity in different sectors as evolving
according to a trade-share weighted average of trade partners' sectoral
productivities. This process is controlled by a parameter that determines the
speed of diffusion of ideas, which we calibrate using historical data on
output growth. Our approach implies that the strength of ideas diffusion is a
function of the strength of input-output linkages in production. This is
consistent with the intuitive idea proposed in \textcite{buera_global_2020} suggesting that ideas can be diffused through the purchases of intermediate inputs.

Productivity thus evolves endogenously as a by-product of micro-founded market
decisions \textemdash i.e., an externality that market agents affect with
their behavior but do not take it into account when making decisions. In this
framework, the outbreak of large-scale trade conflicts will have spillover
effects on the future path of sectoral productivities of all countries.
Changes in trade costs induce trade diversion and creation, which, in turn,
impact productivity dynamics in a way that is not internalized by agents.

After characterizing the model, we use it to perform policy experiments in the
context of heightened geopolitical conflicts. We explore the potential impact
of a \textquotedblleft decoupling of the global economy,\textquotedblright\ a
hypothetical scenario under which technology systems would diverge in the
global economy. We divide the global economy into a Western bloc and an Eastern
bloc based on differential scores in foreign policy similarity.

First, we simulate increased trade costs arising from geopolitical
circumstances, which increase frictions prohibitively if one country wants to
trade with another one outside its bloc. Second, we simulate a scenario of a
global increase in tariffs, in which all countries move from a cooperative
tariff setting in the context of the WTO to a non-cooperative tariff setting.
\textcite{nicita_cooperation_2018} estimate to increase global tariffs, on
average, by 32 percentage points. For simplicity, we use this average number
as a reference and we assume that countries in different blocs raise tariffs
against countries in the other bloc by this average amount. Third, we explore
the potential effect of moving one of the regions from the Western Bloc to the
Eastern Bloc. Fourth, we limit decoupling to electronic equipment, the sector
displaying so far the most decoupling policies. These four policy experiments are
essential to analyze the impact of decoupling, the difference between
different ways to decouple (with resource-dissipating iceberg trade costs or
rent-generating tariffs), the role of technology spillovers in the model by
analyzing bloc membership, and decoupling in the sector most scrutinized. To
limit the already large number of policy experiments, we focus on the
hypothetical scenario of a complete decoupling into a Western and Eastern
Bloc. Hence, we do not explore a scenario with a ``neutral'' bloc.

Our analysis leads to five main findings. First, a multi-sector framework
exacerbates diffusion inefficiencies induced by trade costs relative to a
single-sector one and can be the result of differences in trade costs, unit
costs and productivities between sectors in a country's trading partner; we
explore this issue both through theory and simulations\footnote{Before
conducting simulations with the multi-sector, multi-region model calibrated to
real-world data, we explore the discrepancy between actual and optimal levels
of idea diffusion. This comparison shows that to maximize the total diffusion
of ideas, trade shares must be at their optimal point \textit{in every
sector}.}. Second, we show that the projected welfare losses for the global
economy of a decoupling scenario can be drastic, as large as 12\% in some
regions; and are largest in the lower-income regions as they would suffer the most from reduced technology spillovers from richer areas. Third, the described size and pattern of welfare effects are specific to the model with diffusion of ideas. In a dynamic setting with diffusion of ideas welfare losses are
larger and display more variation. Fourth, if one of the middle income
regions, Latin America and the Caribbean (LAC), would switch from the
higher-income Western bloc to the lower-income Eastern bloc, its welfare costs of decoupling would be significantly higher. This experiment illustrates that policymakers in low- and middle-income countries would face difficult decisions if decoupling would aggravate. Fifth and finally, the welfare costs of decoupling only in electronic equipment, the sector where decoupling is already taking place, would be much smaller than under full decoupling, albeit sizeable, ranging from $0.4-1.9\%$.

We make five main contributions to the literature. First, we build a
multi-sector model of the global economy with Bertrand competition, profits,
and technology spillovers which can be solved recursively that permits
assessing realistic trade policy experiments. Second, we analyze diffusion
inefficiencies in a multi-sector framework both analytically and numerically
showing analytically such inefficiencies can be the result of differences in
trade costs and unit costs between sectors and showing numerically that such
inefficiencies tend to be larger in a multi-sector framework. Third, we
calibrate the strength of the diffusion of technologies through trade with a
tight fit between simulated and historical GDP growth rates. Fourth, we
examine the long-run effects of real-world policy experiments related to the
decoupling of the global economy. Last, we draw insights from the Political
Science literature in order to incorporate geopolitical conflicts into a
workhorse trade model.

Our model builds on the work that evaluates the impact of trade on innovation
and shows that trade openness can increase the level of domestic innovation,
particularly on the single-sector model of \textcite{buera_global_2020}.
Compared to previous work, we present a recursive model with intermediate
linkages; calibrate the strength of the diffusion of ideas to target
historical GDP\ growth rates across all regions; and explore diffusion
inefficiencies in a multi-sector setting. 

\paragraph{Related Literature.}

Our paper is closely related to the literature that adds dynamics to trade
models by incorporating knowledge diffusion channels. The earliest
explorations of this topic go back to \textcite{eaton_international_1999}, who
developed a multi-country dynamic model in which firms innovate by investing
in research and development (R\&D) and knowledge diffuses, after some lag, to
other markets. In this model, diffusion happened somewhat mechanically,  was unrelated to trade, and eventually reached all countries\footnote{In
differentiated varieties of trade models, knowledge diffusion shows up in papers
like \textcite{romer_endogenous_1990},
\textcite{rivera-batiz_international_1991}, and
\textcite{grossman_product_1989}. In the text, we focus on papers that
incorporate knowledge diffusion to Ricardian models, which is the class of
models that this paper falls in.}.

More recently, \textcite{alvarez_idea_2013} combined the
\textcite{eaton_technology_2002} Ricardian model of trade with an idea
diffusion process first presented by \textcite{kortum_research_1997}.
Importantly, the authors conjectured that the diffusion process is
proportional to the quality of managers of firms whose products reach a given destination market. Ideas flow from one market to another in proportion to the trade linkages between them. Therefore, impediments to trade have not only static but also dynamic costs ---as they decrease knowledge diffusion. 

Our model is also closely related to the one described by
\textcite{santacreu_knowledge_2017}, who extended the original model by
\textcite{eaton_international_1999} incorporating lag-diffusion dynamics into
a multi-sector framework. \textcite{santacreu_knowledge_2017} build a
multi-sector model of trade, innovation, and knowledge diffusion, exploring
how the welfare gains from trade are affected by knowledge diffusion through
their impact on changes in comparative advantage. They show that the welfare
gains are larger with endogenous knowledge diffusion because existing
specialization patterns tend to get reinforced by knowledge diffusion.

There are three main differences between their work and ours. First, our model
emphasizes the nexus between trade and idea diffusion, whereas
\textcite{santacreu_knowledge_2017} model technology spillovers as being
independent of the amount of trade (or other endogenous variables like FDI or migration). Additionally, while they calibrate knowledge spillovers with data on patent citation, we calibrate the strength of the diffusion of ideas based on the fit between actual and simulated historical GDP growth rates. Third, the papers have a different focus: we focus on concrete policy questions and explore how the effect of potential trade policy changes are affected by the inclusion of ideas diffusion in the model; while they highlight how patterns of comparative advantage change with technology spillovers.

Finally, \textcite{deng_specialization_2016} integrates \textcite{buera_global_2020} in a \textcite{levchenko_evolution_2016} multi-sector framework and is thus similar to our work. \textcite{deng_specialization_2016} argues that
the framework displays strong convergence in comparative advantage, generates
dynamic gains from trade about 1/3 larger than static gains from trade, and
identifying central players in technology diffusion. There are two main
differences with our work. First, \textcite{deng_specialization_2016} employs a different way to
calibrate the model, following \textcite{levchenko_evolution_2016} to estimate trade costs and
productivity parameters and an approach to fit observed trade and production
flows to estimate the parameter governing the diffusion of ideas parameters.
We instead infer trade costs and productivity based on observed data and
target GDP growth rates with the diffusion of ideas parameter.\footnote{Trade
costs are calibrated targeting observed trade shares as in new quantitative
trade models applying exact hat algebra, whereas productivity is calibrated
based on actual productivity data. With this approach and the chosen
calibration of the ideas diffusion parameter, we stay closer to observed data. As such, baseline values to which counterfactual
experiments are applied are identical to actual values, ensuring that the impact of counterfactual experiments is not distorted.} Second, except for the comparison of dynamic
and static gains from trade, we explore very different questions with the
model constructed. \textcite{deng_specialization_2016} explores issues like convergence in comparative
advantage and central players in technology diffusion, whereas we explore the
costs of geopolitical decoupling and the repercussions of bloc membership besides
studying the inefficiencies of ideas diffusion in a multi-sector framework. 

\paragraph{Organization.}

This paper is organized as follows. In Section \ref{env} we present the model, detailing production, demand, and consumption of the global economy. We also describe the dynamic evolution of productivities in different regions and sectors. In Section \ref{intuition} we describe the discrepancy between the actual and optimal diffusion of ideas in a multi-sector framework. In Section \ref{calibration} we discuss the calibration of the model and underpin the examined policy experiments. In Section \ref{results} we present the results of our main policy experiments and some alternative simulations. Finally, we conclude in Section \ref{conclusion} summarizing the key takeaways.

\section{Environment\label{env}}

Time is discrete and indexed by $t\in\mathcal{T}$. There are $d\in\mathcal{D}$
regions in the global economy, which cover every part of the world economy,
either as a stand-alone country, or a regional aggregate of countries. In each
region, there are multiple industries $i\in\mathcal{I}$.

\subsection{Demand}

In each region $d$ and each period $t$ a representative agent maximizes
Cobb-Douglas preferences over consumption of goods in different sectors
$i\in\mathcal{I}$, $q_{d,t}^{c,i}$:%
\begin{align}
\max_{\{q_{d,t}^{i}\}_{i\in\mathcal{I}}}\sum_{i\in\mathcal{I}}(q_{d,t}%
^{c,i})^{\kappa_{d}^{i}}\quad s.t.\quad\sum_{i\in\mathcal{I}}\kappa_{d}^{i}  &
= 1\nonumber\\
\sum_{i\in\mathcal{I}}p_{d,t}^{i}q_{d,t}^{c,i}  &  \leq (1-s_{d,t}
)Y_{d,t}\nonumber\\
Y_{d,t}  & = w_{d,t}\ell_{d,t}+r_{d,t}k_{d,t}+T_{d,t}+\sum_{i\in\mathcal{I}}
\Pi_{d,t}^{i}%
\end{align}

$Y_{d,t}$ is gross income determined by $w_{d,t},\ell_{d,t}$, the wage and
measure of workers; $r_{d,t}k_{d,t}$, capital income; $T_{d,t}$, transfers;
and $\Pi_{d,t}^{i}$, profits. We set $s_{d,t}$ to be the exogenous savings
rate. Hence, we abstract from intertemporal optimization as discussed further
below in the subsection on market clearing. Savings are used to finance
investment as discussed below. This preference structure implies the following
demand function for goods in sector $i\in\mathcal{I}$:%

\begin{equation}
q_{d,t}^{c,i}=\frac{(1-s_{d,t})\kappa_{d}^{i}Y_{d,t}}{p_{d,t}^{i}}%
\end{equation}

The aggregate price index satisfies:%

\begin{equation}
P_{d,t}^{c}=K\cdot\Pi_{i\in\mathcal{I}}(p_{d,t}^{i})^{\kappa_{d}^{i}}%
\end{equation}

With $K=\Pi_{i\in\mathcal{I}}(\kappa_{d}^{i})^{-\kappa_{d}^{i}}$, a constant.

\subsection{Production}

There are many producers of different varieties $\omega$ of each commodity
$i$. Firms are endowed with identical technology and combine factors of
production $f_{d,t}^{i}$ and intermediate inputs $m_{d,t}^{i}$ to produce
variety $q_{d,t}^{i}\left(  \omega\right)  $:%

\begin{equation}
q_{d,t}^{i}(\omega) = z_{d,t}^{i}(\omega) \left[  (\Psi_{d,t}^{i,f})^{\frac
{1}{\rho_{i}}}(f_{d,t}^{i})^{\frac{\rho_{i}-1}{\rho_{i}}} + (\Psi_{d,t}%
^{i,m})^{\frac{1}{\rho_{i}}}(m_{d,t}^{i})^{\frac{\rho_{i}-1}{\rho_{i}}}
\right] ^{\frac{\rho_{i}}{\rho_{i}-1}}%
\end{equation}

The cost of the unit input bundle, $c_{d,t}^{i}$, is a function of the prices
of factors of production, $pf_{d,t}^{i}$, and prices of commodities used as
intermediates, $p_{d,t}^{i}$:
\begin{equation}
c_{d,t}^{i}= \left[  \Psi_{d,t}^{i,f}(pf_{d,t}^{i})^{1-\rho_{i}} + \Psi
_{d,t}^{i,m}(p_{d,t}^{i})^{1-\rho_{i}} \right] ^{\frac{1}{1-\rho_{i}}}%
\end{equation}

Firms combine factors of production ($f_{d,t}^{i}$) and intermediate
commodities ($m_{d,t}^{i}$) according to the following sub-production
functions:
\begin{align}
f_{d,t}^{i}  & =\left[  (\Psi_{d,t}^{i,k})^{\frac{1}{\nu_{i}}}k_{d,t}%
^{\frac{\nu_{i}-1}{\nu_{i}}}+(\Psi_{d,t}^{i,l})^{\frac{1}{\nu_{i}}}\ell
_{d,t}^{\frac{\nu_{i}-1}{\nu_{i}}}\right]  ^{\frac{\nu_{i}}{\nu_{i}-1}}\\
m_{d,t}^{i}  & =\left[  \sum_{j\in\mathcal{I}}(\Psi_{d,t}^{i,j})^{\frac{1}%
{\mu_{i}}}(q_{d,t}^{m,i,j})^{\frac{\mu_{i}-1}{\mu_{i}}}\right]  ^{\frac
{\mu_{i}}{\mu_{i}-1}}%
\end{align}

The first aggregator combines capital, $k_{d,t}$, and labor, $\ell_{d,t}$, as
factors of production, while the second one uses sectoral commodities
$q_{d,t}^{m,j}$ as intermediate inputs.

\subsection{Supply of Factors of Production}

The supply of the three factors of production changes over time. They are
perfectly mobile and thus have a uniform price across sectors. For each
country, an exogenous path of endowments of labor is imposed based on external
projections from the United Nations and the International Monetary Fund as
described in the data section below.

Aggregate capital, $k_{d,t}$, is a function of capital in the previous period,
$k_{d,t-1}$, depreciation, $\delta$, and investment, $in_{d,t}$, evolving
according to the following law of motion:%

\begin{equation}
k_{d,t}=(1-\delta_{d})k_{d,t-1}+in_{d,t}\label{capacc}%
\end{equation}

Investment in region $d$ is a Leontief function of sectoral investment,
$q_{d,t}^{in,i}$ implying the following expression for sectoral investment
demand and the corresponding price index of investment, $p_{d,t}^{in}$:
\begin{align}
q_{d,t}^{in,i} &  =\overline{\chi}_{d,t}^{i}{in}_{d,t}\\
p_{d,t}^{in} &  =\sum_{i\in\mathcal{I}}\overline{\chi}_{d,t}^{i}p_{d,t}^{in,i}%
\end{align}

We assume that the ratio of a region's trade balance to its total income is
fixed. Abstracting from other components of the current account, the capital
account is equal to the trade balance. Assuming a fixed trade balance ratio
(relative to income) thus implies that the investment rate is equal to the
savings rate minus the trade balance rate, $tb_{d,t}$. Hence, in equilibrium we have:%

\begin{equation}
\sum_{i\in\mathcal{I}}p_{d,t}^{in}in_{d,t}=\left(  s_{d,t}-tb_{d,t}\right)
Y_{d,t}%
\end{equation}

\subsection{International trade}

Consumers, investors and firms demand sectoral commodities $q_{d,t}^{j}$ by
amounts $q_{d,t}^{c,j}$, $q_{d,t}^{in,j}$ and $q_{d,t}^{m,j}$, respectively. A
local producer source the cheapest landed varieties $\{q_{d,t}^{j}%
(\omega):\omega\in\lbrack0,1]\}$ from all countries $s\in\mathcal{D}$ and
produces the sectoral commodity according to the following technology:
\begin{equation}
q_{d,t}^{j}= \left[ \int_{[0,1]}q_{d,t}^{j}(\omega)^{\frac{\sigma_{j}%
-1}{\sigma_{j}}}d\omega\right] ^{\frac{\sigma_{j}}{\sigma_{j}-1}}\label{qd}%
\end{equation}

The price of commodity $j\in\mathcal{I}$ thus satisfies:%
\begin{equation}
p_{d,t}^{j}= \left[ \int_{[0,1]}p_{d,t}^{j}(\omega)^{1-\sigma_{j}}%
d\omega\right] ^{\frac{1}{1-\sigma_{j}}}\label{pd}%
\end{equation}

Trade happens through demand for varieties used as inputs in the production of
sectoral goods $q_{s,t}^{j}$. These goods, in turn, are used in two different
ways: as intermediate inputs in the production of varieties and investment
goods; and in final consumption. Let $x_{sd,t}^{i}(\omega)$ be the landed unit
cost of supplying variety $\omega$ of commodity $i \in\mathcal{I}$ produced in
source region $s \in\mathcal{D}$ and delivered to region $d \in\mathcal{D}$:%

\begin{equation}
x_{sd,t}^{i}(\omega)\equiv\frac{tm_{sd,t}^{i}\cdot\tau_{s,t}^{i}\cdot
c_{sd,t}^{i}}{z_{s,t}^{i}(\omega)}=\frac{\tilde{x}_{sd,t}^{i}}{z_{s,t}%
^{i}(\omega)}%
\end{equation}

where $tm_{sd,t}^{i}$ are gross import taxes, which can be source and
destination specific; $\tau_{sd,t}^{i} \ge1$ are bilateral iceberg trade
costs; and $z_{s,t}^{i}(\omega)$ is the firm's productivity. The last equality
follows from defining $\tilde{x}_{sd,t}^{i} \equiv tm_{sd,t}^{i}\cdot
\tau_{sd,t}^{i} \cdot c_{s,t}^{i}$ as the landed input bundle costs.

Since varieties can be sourced from every region $s\in\mathcal{D}$ consumers
in destination region $d\in\mathcal{D}$ will only buy variety $\omega$ from
the source with the lowest landed price. Following
\textcite{bernard_plants_2003}, producers engage in Bertrand competition.
Hence, the firm with the lowest price of a variety captures the entire market
for that variety. The firm will set the price either equal to the marginal
cost of the second most efficient firm (domestically or from other regions) or
equal to its monopoly price if the marginal cost of the nearest competitor is
higher than the monopoly price. More formally, for each country, order firms
$k=[1,2,\cdots]$ such that $z_{1s,t}^{i}(\omega)>z_{2s,t}^{i}(\omega),\cdots$.
If the lowest-cost provider of the variety $\omega$ to country $d\in
\mathcal{D}$ is a producer from country $s\in\mathcal{D}$, the price in $d$ satisfies:%

\begin{equation}
p_{d,t}^{i}(\omega) = \min\left\{  \underbrace{ \frac{\sigma_{i}}{\sigma
_{i}-1} \frac{\tilde{x}_{sd,t}^{i}}{z^{i}_{1s,t}(\omega)} }_{\text{optimal
monopolist price}} , \underbrace{ \frac{\tilde{x}_{sd,t}^{i}}{z^{i}%
_{2s,t}(\omega)} }_{\substack{\text{MC of 2nd most} \\\text{ productive firm
from } s}} , \min_{n \neq s} \underbrace{ \frac{\tilde{x}_{nd,t}^{i}}%
{z^{i}_{1n,t}(\omega)}}_{\substack{\text{MC of most productive} \\\text{ firm
from other countries} }} \right\}
\end{equation}

\begin{assumption}
[Productivity draws]We follow the canonical \textcite{eaton_technology_2002}
assumption that and take $z_{s,t}^{i}(\omega) : \mathcal{A} \times\mathcal{D}
\times\mathcal{I} \times\Omega\to\mathbb{R}_{+}$ to be the realization of an
i.i.d. random variable, where $\mathcal{A}$ is the set of states of the world.
Productivity is distributed according to a Type II Extreme Value Distribution (Fr\'echet).%

\begin{equation}
F_{s,t}^{i}(z) = \exp\{ - \lambda_{s,t}^{i} z^{-\theta_{i}} \}
\end{equation}

\end{assumption}

The country-specific Fr\'{e}chet distribution has a region-commodity-specific
location parameter $\lambda_{s,t}^{i}$, which denotes absolute advantage
(better draws for all varieties), and a sector-specific scale parameter
$\theta_{i}$, which governs comparative advantage. The location parameter,
$\lambda_{s,t}^{i}$, describes the productivity of region $s$ in sector $i$
and thus determines its absolute advantages, whereas the dispersion parameter
$\theta_{i}$ governs the variation of productivity within sector $i$ between
countries and thus determines the strength of comparative advantage. A higher
$\theta_{i}$ implies less variability in productivity and thus lower potential
for diversification according to comparative advantage.

We show in the Appendix that prices in the destination region $d\in
\mathcal{D}$ will be, respectively:%

\begin{equation}
p_{d,t}^{i}=\boldsymbol{\Gamma}_{1}(\Phi_{d,t}^{i})^{-\frac{1}{\theta_{i}}%
}\label{eq: price_com_solve}%
\end{equation}

\noindent where $\boldsymbol{\Gamma}_{1}$ is a constant
\footnote{Specifically, $\boldsymbol{\Gamma}_{1} \equiv\Big[ 1 - \frac
{\sigma_{i}-1}{\theta_{i}} + \frac{\sigma_{i}-1}{\theta_{i}} \Big(\frac
{\sigma_{i}}{\sigma_{i}-1}\Big)^{-\theta_{i}} \Big] \Gamma\Big( \frac
{1-\sigma_{i}+\theta_{i}}{\theta_{i}} \Big)$, where $\Gamma(\cdot)$ is the
Gamma function}; $\Phi_{d,t}^{i} \equiv\sum_{s \in\mathcal{D}} \lambda
_{s,t}^{i} ( \tilde{x}_{sd,t}^{i}) ^{-\theta_{i}}$.

As there are infinitely many varieties in the unit interval, by the law of
large numbers, the expenditure share of destination region $d\in\mathcal{D}$
on goods coming from source country $s\in\mathcal{D}$ converges to its
expected value. $\pi_{sd,t}^{i}$ denotes the expenditure share of demand in
region $d\in\mathcal{D}$ on goods coming from region $s\in\mathcal{D}$ as a
share of total expenditure on commodity $i\in\mathcal{I}$:%

\begin{equation}
\pi_{sd,t}^{i}\equiv\frac{\lambda_{s,t}^{i}(\tilde{x}_{sd,t}^{i})^{-\theta
_{i}}}{\Phi_{d,t}^{i}} \label{eq: trade_share_com}%
\end{equation}

Sectoral commodities can be used both for final consumption and as
intermediate inputs. Given the assumptions above, $\pi_{t}^{i,j}$, which is
the expenditure on goods coming from $s$ to be used as intermediate inputs in sector $i$ of country $d$ as a share of their total expenditures on
goods from sector $j$ is equal to the trade shares for final consumption, i.e: $(\forall i,j) \pi_{sd,t}^{i,j} = \pi_{sd,t}^{j}$.

In the presence of Bertrand Competition, we show in the Appendix that source
firms realize a profit that is proportional to the total expenditure of
destination countries. In particular, profits are:%

\begin{equation}
\Pi_{s,t}^{i}=\frac{1}{1+\theta_{i}}\sum_{d\in\mathcal{D}}\pi_{sd,t}%
^{i}e_{d,t}^{i};\quad e_{d,t}^{i}=\sum_{ag\in\left\{  c,in,m\right\}  }%
e_{d,t}^{ag,i}%
\end{equation}

\noindent with ith $e_{d,t}^{ag,i}=p_{d,t}^{ag,i}q_{d,t}^{ag,i}$.

\subsection{Equilibrium}

Our model is characterized by a sequence of static equilibria satisfying
equilibrium equations in each of the periods $t$, consisting of equilibrium in the product market and the factor market, detailed in the Appendix. Solutions in the different periods are related in two ways. First, the capital stock in period $t$ is determined by investment in period $t$ and the capital stock in period $t-1$ (plus depreciation) as specified in equation (\ref{capacc}). Second, the country-sector-specific technology parameters $\lambda_{d,t}^{i}$ change over time as specified in equation (\ref{eq: final_law_of_motion}) below.

We abstract from intertemporal optimization of consumption, imposing instead a
fixed savings rate. This makes the model computationally more tractable and
leads to a more straightforward interpretation of the simulation results. We
focus on the long-run effects of decoupling of the global economy from 2020 to
2040 and not on the effects of decoupling on the trajectory followed by the
economy to 2040. However, abstracting from intertemporal optimization implies
that potential effects through changes in savings rates on capital
accumulation are also abstracted from\footnote{The assumption of a fixed trade
balance implies that the capital stock is also not affected by potential
changes in the capital balance in response to shocks. However, the
international finance literature suggests that standard open economy models
with intertemporal optimization have generated counterfactual predictions on
the direction of capital flows between developed and emerging countries in the
1990s and 2000s. Capital was flowing on net from emerging to developed
economies instead of capital flowing to emerging economies with higher
growth rates as predicted by the standard models.}.

\subsection{Dynamic innovation}

Unlike in the standard \textcite{eaton_technology_2002} model or in the
Bertrand-competition version developed in \textcite{bernard_plants_2003}, we
assume that each region's location parameter evolves over time. Each commodity
$i \in\mathcal{I}$ and each country $d \in\mathcal{D}$ has a different
period-specific productivity distribution $F_{d,t}^{i}(z)$.

Our model follows a strand of the literature that models ideas diffusion
through random matches between domestic and foreign managers\footnote{For a
detailed review of this literature, see the comprehensive review chapter
published by \textcite{buera_idea_2018}.}. Seminal examples of this work
include \textcite{jovanovic_growth_1989} and \textcite{kortum_research_1997}.
More recently, \textcite{alvarez_idea_2013} explored how idea diffusion is
intertwined with trade linkages. Like \textcite{buera_global_2020}, we assume
that a manager draws new insights as a by-product of sourcing a basket of inputs.

We extend this framework to a model with diffusion of ideas in a multi-sector
context and solve it in a recursive fashion that permits the assessment of the
long-run effects of policy experiments. The idea diffusion mechanism is
mediated by the input-output structure of production, such that both sector
cost shares and import trade shares characterize the source distribution of
ideas\footnote{As mentioned earlier, our work is closely related to
\textcite{santacreu_knowledge_2017}, who extend
\textcite{kortum_research_1997} to a multi-sector framework. We differ in that
they model diffusion as happening separately from trade, rather than as a
trade-externality.}.

\begin{assumption}
[Idea formation]New ideas are the transformation of two random variables,
namely: (i) original insights $o$, which arrive according to a power law:
$O_{t}(o) = Pr(O < o) = 1 - \alpha_{t} o^{-\theta}$; (ii) derived insights
$z^{\prime}$, drawn from a source distribution $G_{d,t}^{i}(z)$. After the
realization of those two random variables, the new idea has productivity $z =
o (z^{\prime})^{\beta}$, where $o$ is the original component of the new idea,
$z^{\prime}$ is the derived insight, and $\beta\in[0,1)$ captures the
contribution of the derived insights to new ideas. Local producers only adopt
new ideas if their quality dominates the quality of local varieties.
Therefore, for any period, the domestic technological frontier evolves according to \footnote{Here we simply use the fact that $o = z(z^{\prime-\beta}$ and note that, given an insight $z^{\prime}$, at any moment $t$ the arrival rate of ideas of quality better than $z$ is $Pr(O > o) = Pr(O > z(z^{\prime-\beta} ) = \alpha_{t} z^{-\theta}(z^{\prime})^{\beta\theta}$. We then integrate over all possible values of $z^{\prime}$.}:%

\[
F_{d,t+\Delta}^{i}(z) = \underbrace{F_{d,t}^{i}(z)}_{Pr\{\text{productivity}%
<z\text{ at } t\}} \times\underbrace{\Big( 1 - \int_{t}^{t+\Delta} \int%
\alpha_{\tau}z^{-\theta} (z^{\prime})^{\beta\theta} dG_{d,\tau}^{i}(z^{\prime
}) d\tau\Big)}_{Pr\{\text{no better draws in } (t,t+\Delta)\} }%
\]

\end{assumption}

\begin{lemma}
[Generic Law of Motion, \cite{buera_global_2020}]Given Assumption 2, if, for
any $t$, $F_{d,t}^{i}(z)$ is Fr\'echet with location parameter $\lambda
_{d,t}^{i} = \int_{-\infty}^{t} \alpha_{\tau}\int(z^{\prime})^{\beta\theta
_{i}} dG_{d,\tau}^{i}(z^{\prime}) d\tau$ and scale parameter $\theta_{i}$, the
former evolving according to the following law of motion:%

\begin{equation}
\label{eq: lambda_law_of_motion}\Delta\lambda_{d,t}^{i} = \alpha_{t}
\int(z^{\prime})^{\beta\theta_{i}} dG_{d,t}^{i}(z^{\prime})
\end{equation}

\noindent where $\alpha_{t}$ is a parameter that controls the arrival rate of
ideas and $\beta$ is the sensitivity of current productivity to derived
insights. The integral on the right-hand side of the equation denotes the
average productivity of ideas drawn from source distribution $G_{d,t}%
^{i}(z^{\prime})$\footnote{Equation (\ref{eq: lambda_law_of_motion}) is a
discrete-time approximation of the continuous-time law of motion derived in
the Appendix.}.
\end{lemma}

\begin{proof}
Appendix.
\end{proof}

To fully characterize (\ref{eq: lambda_law_of_motion}), we need to define the
source distribution. We assume that managers learn from their suppliers, such
that $G_{d,t}^{i}(z^{\prime})$ is proportional to the sourcing decisions in
the production of commodity $i$ in country $d$. Productivity thus evolves
endogenously as a by-product of sourcing decisions. Additionally, we assume
that insights take time to come to fruition. Rather than drawing insights from
interactions with suppliers in the current period, we assume that insights
take one period to materialize. Intuitively, we are assuming that
entrepreneurs have to study their purchases for one period and only then draw
insights. This assumption will be convenient because it will allow us to
compute the law of motion for technology without relying on present-period
trade shares. Therefore, we will be able to solve the model recursively and
use it for counterfactual analysis of the long-run impact of policy experiments.

\begin{assumption}
[Source Distribution from Intermediates]The source distribution $G_{d,t}
^{i}(z^{\prime}) \equiv \sum_{j \in\mathcal{I}} \eta_{d,t-1}^{i,j} \sum_{s
\in\mathcal{D}} H_{sd,t-1}^{i,j}(z^{\prime})$, where $\eta_{d,t}^{i,j}$ is the
intermediate cost share of sector $j$ when producing good $i$ in region $d$;
and $H_{sd,t-1}^{i,j}(z^{\prime})$ is the fraction of commodities for which
the lowest cost supplier in period $t-1$ is a firm located in $s
\in\mathcal{D}$ with productivity weakly less than $z^{\prime}$.
\end{assumption}

\begin{proposition}
[Law of Motion in a Multi-Sector Framework]Given Assumptions 1-3, in the
multi-sector multi-region economy described in the previous section, the
country-sector-specific technology parameter evolves according to the
following process:%

\begin{equation}
\label{eq: final_law_of_motion}\Delta\lambda_{d,t}^{i} = \alpha_{t} \sum_{j
\in\mathcal{I}} \Gamma(1-\beta) \eta_{d,t-1}^{i,j} \sum_{s \in\mathcal{D}}
(\pi_{sd,t-1}^{i,j})^{1-\beta}(\lambda_{s,t-1}^{j})^{\beta}%
\end{equation}

\noindent where $\Gamma(\cdot)$ is the gamma function, $\eta_{d,t-1}^{i,j}$
are cost shares, and $\pi_{sd,t-1}^{i,j}$ are intermediate input trade shares.
\end{proposition}

\begin{proof}
Combining the result of the Lemma stated above and Assumption 2, we can express the law of motion as:
\begin{eqnarray*}
\Delta \lambda_{d,t}^{i} &=& \alpha_{t} \int z^{\beta \theta_i} dG_{d,t}^{i}(z) \\
&=& \alpha_{t} \sum_{j \in \mathcal{I}} \eta_{d,t-1}^{i,j} \sum_{s \in \mathcal{D}} \int z^{\beta \theta_i} dH_{sd,t-1}^{i,j}(z)
\end{eqnarray*}
The strategy of the proof relies on Bertrand competition and Assumption 1, regarding productivity draws. The joint distribution of the two most productive firms in a given industry $i$ in country $s$ is given by $F_{s,t}^{i}(z_1,z_2) =(1+ \lambda_{s,t}^{i}[z_2^{-\theta_i}-z_1^{-\theta_i}] )  \exp \{ - \lambda_{s,t}^{i} z_2^{-\theta_i}  \}$.
Incorporating landed input bundle costs $\tilde{x}_{sd,t}^{i}$, we calculate the probability that the lowest cost producer at destination $d$ is from $s$ and has productivity lower than $z_2$ or in the range $[z_2,z_1)$. In the Appendix, we work out each integral $\int z^{\beta \theta_i} dH_{sd,t-1}^{i,j}(z)$ and derive the result stated in the Proposition.
\end{proof}

This result extends the main proposition of \textcite{buera_global_2020} to a multi-sector framework. We will use equation (\ref{eq: final_law_of_motion}) and a calibrated path of $\alpha_{t}$ to solve for an endogenous path for $\lambda_{d,t}^{i}$.

\section{Discussion and Intuition of Ideas Diffusion in a Multi-sector
Framework\label{intuition}}

\label{section: intuition}

In this section, we provide some intuition regarding how the idea diffusion
mechanism operates in the multi-sector framework. We will use a simplified
two-country, two-region economy to show how the actual market equilibrium
outcome exhibits both within- and between-sector deviations from an outcome
maximizing ideas diffusion. More specifically, we will show with an example
that in a multi-sector framework differences in trade costs, unit costs and
productivities between sectors can already lead to deviations between the
actual and optimal import shares.

Below we denote industries as $i,-i$ and label regions as home ($h$) and
foreign ($f$). $\eta^{i}$ denotes own-cost share of industry $i$, assumed to
be identical in both countries; $\lambda_{h}^{i}$ is the productivity in
sector $i$ at home; and $\pi^{i,-i}_{h}$ stands for the domestic trade share
of $h$ in the total intermediate cost of inputs from industry $-i$ in the
production of industry $i$. We drop time subscripts for simplicity. In a
two-by-two symmetric economy, equation (\ref{eq: final_law_of_motion}) for
industry $i$ at home is proportional to a weighted average of the two sector
input shares:
\begin{align*}
\Delta\lambda^{i}  &  \propto\eta^{i} [(\pi_{h}^{i,i})^{1-\beta}(\lambda
_{h}^{i})^{\beta} + (1-\pi_{h}^{i,i})^{1-\beta}(\lambda_{f}^{i})^{\beta}]\\
&  + (1-\eta_{d}^{i}) [(\pi_{h}^{i,-i})^{1-\beta}(\lambda_{h}^{-i})^{\beta} +
(1-\pi_{h}^{i,-i})^{1-\beta}(\lambda_{f}^{-i})^{\beta}]
\end{align*}

What would the optimal total domestic trade shares in sectors $i,-i$ be? If a
planner were to choose $\pi_{h}^{i,i},\pi_{h}^{i,-i}$ to maximize idea
diffusion\footnote{Specifically, a planner is maximizing $\max_{\{\pi
_{h}^{i,i},\pi_{h}^{i,-i}\}}\eta^{i}[(\pi_{h}^{i,i})^{1-\beta}(\lambda_{h}%
^{i})^{\beta}+(1-\pi_{h}^{i,i})^{1-\beta}(\lambda_{f}^{i})^{\beta}%
]+(1-\eta_{d}^{i})[(\pi_{h}^{i,-i})^{1-\beta}(\lambda_{h}^{-i})^{\beta}%
+(1-\pi_{h}^{i,-i})^{1-\beta}(\lambda_{f}^{-i})^{\beta}]$, which is a
separable and strictly concave programming problem in $\pi_{h}^{i,i},\pi
_{h}^{i,-i}$.}, the ratio of optimal total expenditure shares within a given
sector would satisfy:%

\[
\Bigg( \frac{\eta^{i} \pi_{h}^{i,i}}{\eta^{i} (1-\pi_{h}^{i,i})}
\Bigg)^{\text{Diffusion Optimum}} = \frac{\lambda_{h}^{i}}{\lambda_{f}^{i}}%
\]

How does this compare with the total domestic trade shares that result from
market optimization by private agents? Market allocations incorporate unit
costs $x_{h}^{i}$ and trade costs $\tau\geq1$ (assumed to be symmetric):%

\[
\Bigg(\frac{\eta^{i}\pi_{h}^{i,i}}{\eta^{i}(1-\pi_{h}^{i,i})}%
\Bigg)^{\text{Actual Trade}}=\frac{\lambda_{h}^{i}(x_{h}^{i})^{-\theta}%
}{\lambda_{f}^{i}(\tau^{i}\cdot x_{f}^{i})^{-\theta}}%
\]

In general, the market allocation will be different from the diffusion optimal
one, except if differences in trade and unit costs exactly cancel out, i.e.:
$\tau^{i}=x_{h}^{i}/x_{f}^{i}$. This within-sector distortion mimics the
single-sector results of \textcite{buera_global_2020}. Below, we show that in
a multi-sector framework not only there are deviations within each sector, but
also that in general, they are not proportional across sectors: i.e., there are
cross-sector distortions. Consider first the ratio of domestic shares in total
trade expenditures in sectors $i,-i$ that induce optimal idea diffusion:%

\begin{equation}
\label{eq: planner}\Bigg( \frac{\eta^{i} \pi_{h}^{i,i}}{(1-\eta^{i}) \pi
_{h}^{i,-i}} \Bigg)^{\text{Diffusion Optimum}} = \underbrace{\frac{\eta^{i}%
}{1-\eta^{i}}}_{\text{cost share}} \times\underbrace{\frac{\lambda_{h}^{i}%
}{\lambda_{h}^{-i}}}_{\text{own-productivity}} \times\underbrace{
\Bigg( \frac{\lambda_{h}^{i} + \lambda_{f}^{i}}{\lambda_{h}^{-i} + \lambda
_{f}^{-i}} \Bigg)^{-1} }_{\text{industry-wise productivity}}%
\end{equation}

The ratio can be decomposed into a cost-share component, a relative
own-productivity component; and an industry-wise relative productivity
component. Intuitively, optimal domestic trade allocation in industry $i$ will
increase relative to industry $-i$ if intermediate cost share of industry $i$
increases and if the relative domestic productivity of industry $i$ goes up.
The ratio is decreasing in industry-wise relative productivity: if the productivity gap between foreign and home is larger in industry $i$ relative
to industry $-i$, optimal domestic trade share of industry $i$ will decrease
relative to industry $-i$.

How does this compare with the total domestic trade shares that result from
market optimization? The industry-wise productivity ratio is adjusted by unit
and trade costs:%

\begin{equation}
{\small {\label{eq: ft}\Bigg(\frac{\eta^{i}\pi_{h}^{i,i}}{(1-\eta^{i})\pi
_{h}^{i,-i}}\Bigg)^{\text{Actual Trade}}=\underbrace{\frac{\eta^{i}}%
{1-\eta^{i}}}_{\text{cost share}}\times\underbrace{\frac{\lambda_{h}^{i}%
(x_{h}^{i})^{-\theta}}{\lambda_{h}^{-i}(x_{h}^{-i})^{-\theta}}}_{\text{own
cost-adj. productivity}}\times\underbrace{\Bigg(\frac{\lambda_{h}^{i}%
(x_{h}^{i})^{-\theta}+\lambda_{f}^{i}(\tau^{i}\cdot x_{f}^{i})^{-\theta}%
}{\lambda_{h}^{-i}(x_{h}^{-i})^{-\theta}+\lambda_{f}^{-i}(\tau^{-i}\cdot
x_{f}^{-i})^{-\theta}}\Bigg)^{-1}}_{\text{industry-wise cost-adj.
productivity}}}}%
\end{equation}

These differences induce a gap between the diffusion optimal allocation and
the market allocation. Define $\aleph$ as the ratio of equations
(\ref{eq: ft}) for (\ref{eq: planner}):%

\begin{equation}
\aleph=\underbrace{\Bigg(\frac{x_{h}^{i}}{x_{h}^{-i}}\Bigg)^{-\theta}%
}_{\text{domestic cost gap}}\times\underbrace{\Bigg(\frac{\lambda_{h}%
^{i}(x_{h}^{i})^{-\theta}+\lambda_{f}^{i}(\tau^{i}\cdot x_{f}^{i})^{-\theta}%
}{\lambda_{h}^{i}+\lambda_{f}^{i}}\Bigg)^{-1}}_{\text{industry-wise
cost-induced deviation in }i}\times\underbrace{\Bigg(\frac{\lambda_{h}%
^{-i}(x_{h}^{-i})^{-\theta}+\lambda_{f}^{-i}(\tau^{-i}\cdot x_{f}%
^{-i})^{-\theta}}{\lambda_{h}^{-i}+\lambda_{f}^{-i}}\Bigg)}%
_{\text{industry-wise cost-induced deviation in }-i}\label{eq: aleph}%
\end{equation}

Whenever $\aleph\neq1$, there is a sectoral distortion in domestic trade
expenditure shares. If $\aleph> 1$ ($<1$), domestic trade expenditure share on
sector $i$ relative to sector $-i$ is above (below) the diffusion-optimal
ratio. Even if $\aleph=1$, that does not guarantee the absence of deviations from
the optimal diffusion point. Rather, it means that deviations (or absence
thereof) are proportional in both sectors, such that domestic trade share in
one sector is not disproportionately higher (lower) in sector $i$ relative to
sector $-i$.

In general, deviations need not be proportional. In fact, only in knife edge
cases $\aleph=1$\footnote{For instance, if countries have identical input
costs across industries $x_{h}^{i}=x_{h}^{-i},x_{f}^{i}=x_{f}^{-i}$; and
either industries in each country have identical productivity ($\lambda
_{h}^{i}=\lambda_{h}^{-i}$, $\lambda_{f}^{i}=\lambda_{f}^{-i}$); or
sector-specific productivities at home are a linear transformation of the
sector-specific productivities at foreign ($\lambda_{h}^{i}=\kappa\cdot
\lambda_{f}^{i}$, $\lambda_{f}^{-i}=\kappa\cdot\lambda_{f}^{-i}$, $\kappa
\in\mathbb{R}_{++}$)}. This underscores that, in a multi-sector framework,
there will not only be \textit{within sector distortions}, but also
\textit{between sector distortions}. Domestic sourcing will be biased towards
the industry with the lowest relative cost, even if that industry is not very
productive. For instance, if costs are disproportionately low in one industry
$i$ relative to industry $-i$, either domestically or industry-wise, domestic
trade share will be disproportionately high in industry $i$ under market
relative to the optimal allocation and $\aleph>1$.

More generally, differences between the optimal and market import shares
emerge from the interaction of differences in trade costs between sectors,
differences in productivities between countries and sectors, and differences
in unit costs between countries and sectors. A range of examples can be
explored of simple settings leading to an intuitive deviation between optimal and actual market shares. First, suppose that productivities and unit costs are identical across sectors and countries. Equation (\ref{eq: aleph}) shows that differences in trade costs between the two sectors then lead to $\aleph\neq1$.
More specifically, $\aleph>1$ if $\tau^{i}>\tau^{-i}$, implying that the
relative domestic spending share in sector $i$ relative to sector $-i$ is
above the optimum. Intuitively, region $h$ is importing too much in sector
$-i$ in this case since trade costs are smaller than in sector $i$. Second,
suppose trade costs are identical between sectors, productivities identical
between sectors and regions, and unit costs identical across sectors in region
$i$. In this case higher unit costs in region $f$ in sector $i$ compared to
sector $-i$, $x_{f}^{i}>x_{f}^{-i}$, implies again that $\aleph>1$.
Intuitively, the domestic spending share in sector $i$ relative to sector $-i$
is above the optimum in this case because region $h$ is importing too much in
sector $-i$, because of lower unit costs. Third, suppose now that trade costs
and unit costs are identical between countries and sectors and domestic
productivities are identical between the two sectors. $\lambda_{f}^{i}%
>\lambda_{f}^{-i}$ implies then $\aleph>1$. Intuitively, the domestic spending
share in region $h$ is too high in sector $i$ compared to sector $-i$ in this
case, because the region should import more from the trading partner $f$ in
the more productivity sector $i$. More complicated examples of the interaction
of differences in the two or three variables can be explored as well. The main
point to take away from these examples is that in a multi-sector framework
differences only in sectoral trade costs, sectoral productivities and unit
costs in a country's trading partner already lead to a discrepancy between the
actual and optimal domestic spending shares.

We next illustrate the additional complexities arising in a multi-sector
framework geometrically. First, consider what happens \textit{within} sector
$i$ in a fully symmetric two-by-two economy. Even with identical countries,
the strict concavity of the diffusion equation implies that idea diffusion is not
uniform as $\pi_{h}^{i,i}$ varies. The optimal diffusion point is $(\pi
_{h}^{i,i})^{\text{Diffusion Optimum}}=\lambda_{h}^{i}/(\lambda_{h}%
^{i}+\lambda_{f}^{i})=1/2$. Under the market allocation, trade costs induce
home bias such that domestic share is $(\pi_{h}^{i,i})^{\text{Actual Trade}%
}=1/(1+\tau^{-\theta})>1/2$ and ideas diffusion is below the optimal point. If
trade costs increase and $\tau\rightarrow\infty$, the home country moves to
autarky, and deviations from the optimal idea diffusion reach a maximum. We
plot the optimal, actual trade, and autarky points along the ideas diffusion
function for sector $i$ on the left-hand side panel of Figure
\ref{fig: two-by-two}.

\begin{figure}[th]
\centering
\includegraphics[width=\textwidth]{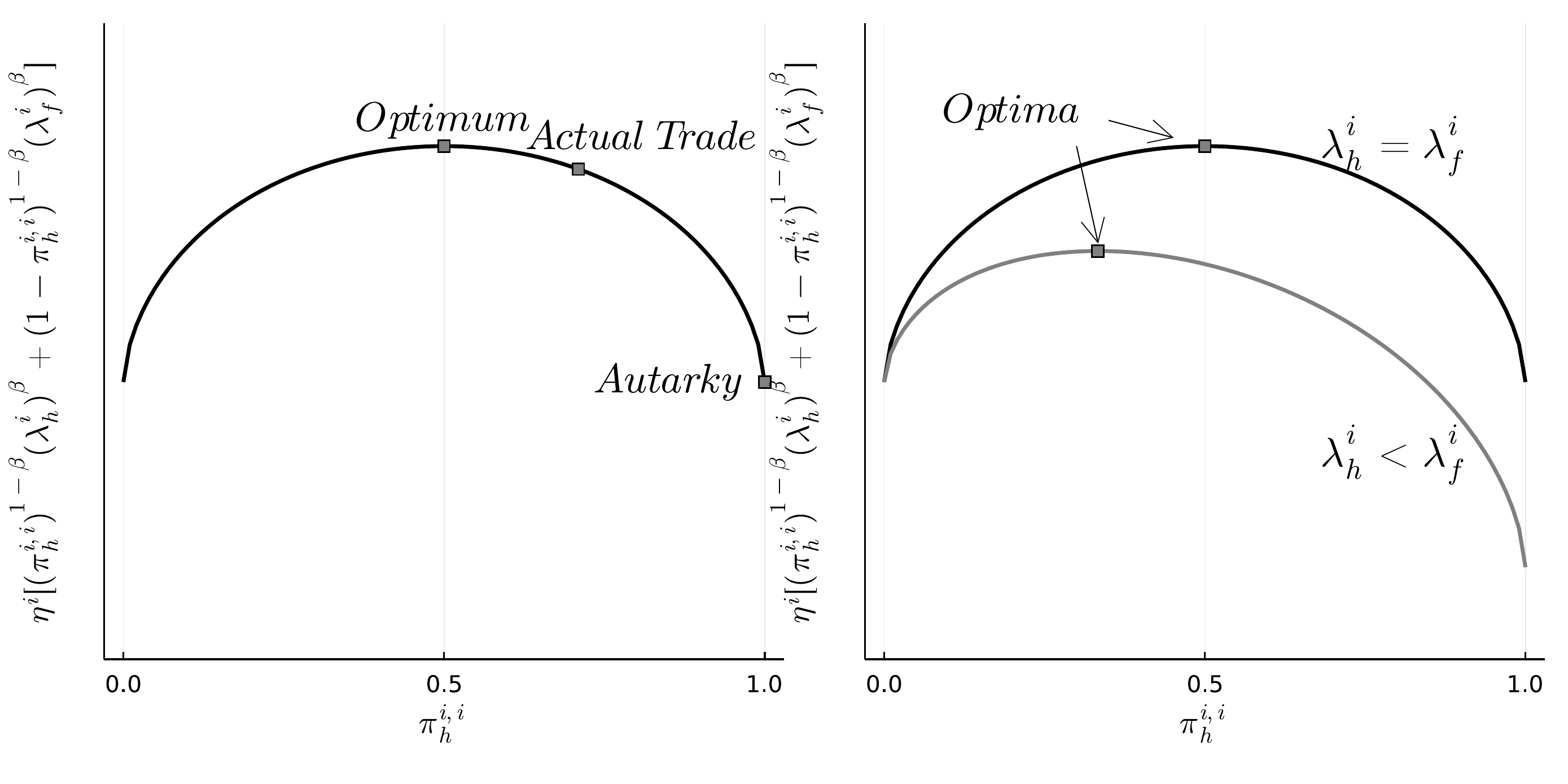}\caption{\textbf{Within
sector idea diffusion functions in a two-by-two economy.}
{\protect\footnotesize {\ Both panels plot the idea diffusion functions for
the home country in a two-by-two model within sector $i$: $\eta^{i} [(\pi
_{h}^{i,i})^{1-\beta}(\lambda_{h}^{i})^{\beta} + (1-\pi_{h}^{i,i})^{1-\beta
}(\lambda_{f}^{i})^{\beta}]$. The left-hand side panel shows the optimal,
actual trade, and autarky points along the ideas diffusion function when
countries are fully symmetric ($\lambda_{h}^{i} = \lambda_{f}^{i}$). The right
hand side panel plots the functions and diffusion optimal points for the cases
when countries have identical productivities $\lambda_{h}^{i} = \lambda
_{f}^{i}$ and the home country is less productive $\lambda_{h}^{i} <
\lambda_{f}^{i}$.} }}%
\label{fig: two-by-two}%
\end{figure}

The right-hand side panel illustrates what happens when the home country has a
lower productivity in sector $i$. The curve shifts down at the autarky point
and the optimal solution moves to the left (smaller domestic trade
share)\footnote{Formally, once countries are no longer symmetric, we need to
make the following regularity condition to guarantee convergence to the
autarky equilibrium: $\lim_{\tau\to\infty}\ (\tau x_{f}^{i})/x_{h}^{i} =
+\infty$ . Most models make this assumption either explicitly or implicitly.}.
When $\lambda_{h}^{i} < \lambda_{f}^{i}$, diffusion losses from high trade
costs are higher. This highlights a key characteristic of this class of
models: countries that are \textit{less productive} in a given sector have
\textit{higher dynamic gains from trade}\footnote{In fact, for any $\pi
^{i}_{h} \in(0,1]$, the marginal change in diffusion as $\pi^{i}_{h}$
increases will be increasing in a country's productivity. To see that, take
$\frac{\partial\Delta\lambda_{h}^{i}}{\partial\pi_{h}^{i}} = \alpha\cdot
\Gamma(1-\beta) \cdot\eta^{i} (1-\beta) [(\pi_{h}^{i,i})^{-\beta}(\lambda
_{h}^{i})^{\beta} - (1-\pi_{h}^{i,i})^{-\beta}(\lambda_{f}^{i})^{\beta}]$,
which is increasing in $\lambda_{h}^{i}$.}.

When considering a multi-sector framework, within-sector inefficiencies
accumulate. For instance, suppose that in sector $i$ domestic trade share
$\pi_{h}^{i}=0$ while, in sector $-i$, $\pi_{h}^{i}=1$. If $\eta^{i}=1/2$,
deviations from optimal idea diffusion will be at a maximum even though total
trade share will be at the optimal point ($1/2$). In a multi-sector framework,
the fact that \textit{total domestic trade share} is at its optimal point is
necessary but \textit{no longer sufficient for optimal diffusion}. To maximize
total diffusion, trade shares must be at their optimal point \textit{in every
sector}.

Figure \ref{fig: two-by-two3d} underlines this fact. It shows that there is a
unique point in the $[0,1]^{2} \times[0,\infty)$ space that maximizes idea
diffusion in a two-by-two symmetric model as $\pi_{h}^{i,i}, \pi_{h}^{i,-i}$
vary. With $\eta^{i}=1/2$, every point in the diagonal $\pi_{h}^{i,i}%
=1-\pi_{h}^{i,-i}$ will have total trade share at its optimal point $1/2$.
Additionally, every point in the counterdiagonal $\pi_{h}^{i,i}=\pi_{h}%
^{i,-i}$ has absence of between sector distortion ($\aleph=1$). However,
neither fact is sufficient to guarantee optimal diffusion. Only if trade
shares are optimal in both sectors (i.e., $\pi_{h}^{i,i}=\pi_{h}^{i,-i}=1/2$)
diffusion is maximized in this simplified economy. Any other point will have
some degree of inefficiency.

\begin{figure}[th]
\centering
\includegraphics[scale=0.4]{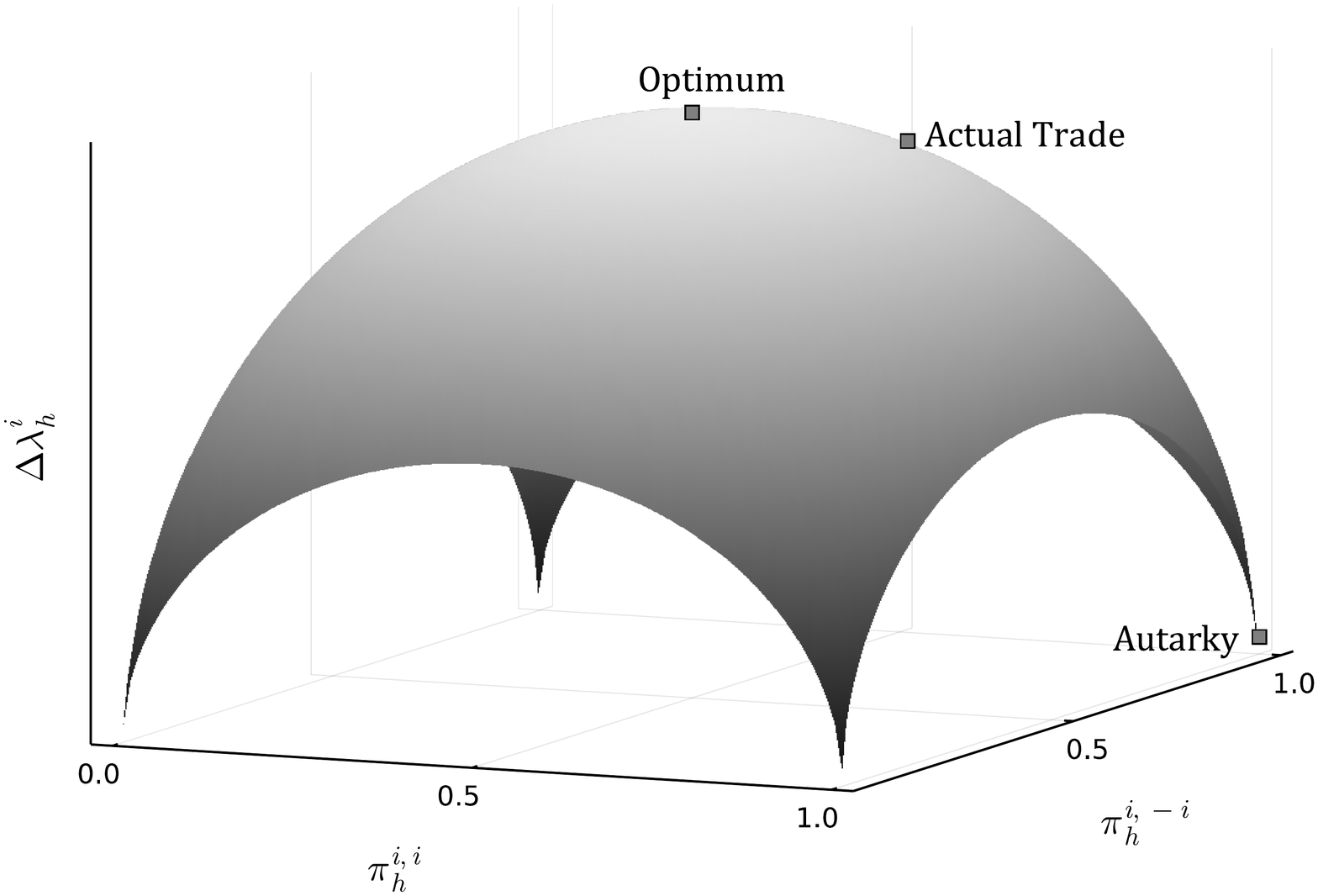}\caption{\textbf{Idea
diffusion function in a two-by-two economy.} {\protect\footnotesize {\ The
graph shows $\Delta\lambda^{i} \propto\eta^{i} [(\pi_{h}^{i,i})^{1-\beta
}(\lambda_{h}^{i})^{\beta} + (1-\pi_{h}^{i,i})^{1-\beta}(\lambda_{f}%
^{i})^{\beta}] + (1-\eta_{d}^{i}) [(\pi_{h}^{i,-i})^{1-\beta}(\lambda_{h}%
^{-i})^{\beta} + (1-\pi_{h}^{i,-i})^{1-\beta}(\lambda_{f}^{-i})^{\beta}]$ as a
function of domestic trade share in sectors $i,-i$: $\pi_{h}^{i}, \pi_{h}%
^{-i}$. If countries and sectors are identical and $\eta^{i}=1/2$, Optimal,
Actual Trade, and Autarky allocations are represented in this figure. The marginal contributions of each sector to total diffusion are as shown in the
left panel of Figure \ref{fig: two-by-two}} }}%
\label{fig: two-by-two3d}%
\end{figure}

It is easy to see how the problem increases in complexity and inefficiencies
accumulate as the number of countries and sectors increase. With $N$ countries
and $J$ sectors, there are $(N-1)J$ free parameters that would need to be at
their optimal points if diffusion were to be maximized. In a market
equilibrium, those parameters will not be at their optimal points and each one
of them will contribute to some deviation from optimal diffusion. Therefore, a
multi-sector framework is important to have a more realistic assessment of
diffusion in policy experiments. We derive multi-sector, multi-region versions
of equations \ref{eq: planner}, \ref{eq: ft}, and \ref{eq: aleph} in Appendix
\ref{appendix: optimaldiffusion}, but most of the intuition can be represented
with the simplified version presented in this section.

\section{Calibration and Setup of Policy Experiments\label{calibration}}

In this section, we first outline the employed baseline data and behavioral
parameters. The parameter determining the strength of diffusion of ideas is
calibrated by minimizing the difference between the historical and simulated
GDP\ growth rate across many countries in the world economy. We then motivate
and describe the detailed setup of the experiments.

\subsection{Data and Behavioral Parameters}

\subsubsection{Baseline Data}

The model is calibrated to trade and production data from the 2014 version of
the GTAP Data Base, Version GTAP10A. This means that all spending and cost
shares are set equal to the shares in the 2014 database, following the same
calibration procedure as in models employing exact hat algebra
\parencite{dekle_unbalanced_2007}. The data are aggregated into 10 regions and 6
sectors as specified in Table \ref{tab: aggregation}. The model is solved
until 2040 in a sequence of recursive dynamic simulations, thus solving the
model period per period, using the model solution in the previous period as
the starting point for the next period. Population grows based on
UN\ population projections and labor supply grows based on International
Monetary Fund projections for employment (until 2025) and United Nations
projections regarding working age population (from 2026 until 2040).

\begin{table}[ptbh]
\caption{Overview of regions and sectors}%
\label{tab: aggregation}
\centering
\begin{tabular}
[c]{llr|rrr}\hline
\multicolumn{2}{c}{Region} &  & \multicolumn{2}{c}{Sector} & \\
Code & Description &  & \multicolumn{1}{l}{Code} &
\multicolumn{1}{l}{Description} & \\\hline
\texttt{chn} & China &  & \multicolumn{1}{l}{\texttt{pri}} &
\multicolumn{1}{l}{Primary (agri \& natres)} & \\
\texttt{e27} & European Union 27 &  & \multicolumn{1}{l}{\texttt{lmn}} &
\multicolumn{1}{l}{Light manufacturing} & \\
\texttt{jpn} & Japan &  & \multicolumn{1}{l}{\texttt{hmn}} &
\multicolumn{1}{l}{Heavy manufacturing} & \\
\texttt{ind} & India &  & \multicolumn{1}{l}{\texttt{elm}} &
\multicolumn{1}{l}{Electronic Equipment} & \\
\texttt{lac} & Latin America &  & \multicolumn{1}{l}{\texttt{tas}} &
\multicolumn{1}{l}{Business services} & \\
\texttt{ode} & Other developed &  & \multicolumn{1}{l}{\texttt{ots}} &
\multicolumn{1}{l}{Other Services} & \\
\texttt{rwc} & Rest of the World - Eastern bloc &  &  &  & \\
\texttt{rwu} & Rest of the World - Western bloc &  &  &  & \\
\texttt{rus} & Russia &  &  &  & \\
\texttt{usa} & United States &  &  &  & \\\hline
\end{tabular}
\end{table}

The data in the GTAP\ Data Base do not include profit income as in our model
with Bertrand competition. Therefore, we have to modify the baseline data
employed, considering that profit income $\Pi_{s}^{i}$ is a share $\frac
{1}{1+\theta_{i}}$ of the total value of sales in sector $i$ in region $s$. We
have done this as follows, proceeding in two steps. First, we reduced the
value of payments to the production factor capital (capital income) by $50\%$
and reallocated it to profit income. With this step the share of profit income
in the value of sales is not yet equal to $\frac{1}{1+\theta_{i}}$. Therefore, in a second step we employ our model to modify the base data to target the
share of profit income in the value of sales for each country and sector. The
reason to proceed in two steps is that capital income in some cases is smaller
than profit income required by the model. This is especially the case in
sectors with large intermediate linkages and a small trade elasticity, because
profit income is a share of gross output in the Bertrand model, whereas
capital income is part of net output. The fact that capital income is for some
countries smaller than profit income also implies that numerical problems in
finding balanced data with profits can appear when raising the number of
countries and sectors. Therefore, the number of countries and sectors is
limited in the simulations presented here.\footnote{As an alternative
approach, we can reduce the value of payments to factors of production by an
identical share for all production factors and reallocate this to profit
income. The reallocation is set such that profit income $\Pi_{s,t}^{i}$
becomes a share $\frac{1}{1+\theta_{i}}$ of the value of sales. However, this
approach is not followed because there is a risk that factor income in the
data is smaller than profit income implied by the model. As discussed, this is
especially a risk in sectors with large intermediate linkages and a small
trade elasticity.}

\subsubsection{Behavioral Parameters}

The dispersion parameter of the Fr\'{e}chet distribution, $\theta_{i}$, equal
to the trade elasticity, is based on the estimates of trade elasticities in
\textcite{hertel_how_2007}. The substitution elasticity between value-added
and intermediates, $\rho$, between intermediates from different sectors, $\mu
$, are equal to zero, implying a Leontief structure. As such we follow the
approach employed in most CGE models, which finds empirical support in recent
estimates with US data (\cite{atalay_how_2017}). The substitution elasticities
between production factors, $\nu_{i}$, are based on the values in the GTAP
Data Base.

Table \ref{tab: behparameters} displays the values of the dispersion parameter
of the Fr\'echet distribution, $\theta_{i}$ and the substitution elasticity
between production factors $\nu_{i}$.

\begin{table}[ptbh]
\caption{Behavioral parameters}%
\label{tab: behparameters}
\centering
\begin{tabular}
[c]{lrrr}\hline
& $\theta_{i}$ & $\nu_{i}$ & \\\hline
Primary (agriculture \& natres) & 10.09 & 0.27 & \\
Light manufacturing & 4.60 & 1.20 & \\
Heavy manufacturing & 5.99 & 1.26 & \\
Electronic Equipment & 7.80 & 1.26 & \\
Business services & 2.80 & 1.26 & \\
Other Services & 2.90 & 1.42 & \\\hline
Source & \cite{hertel_how_2007} & \cite{hertel_how_2007} & \\\hline
\end{tabular}
\end{table}

Even though the location parameters of the sector-country specific Fr\'{e}chet
distribution $\lambda_{s,t}^{i}$ evolve endogenously in this model, their
starting values need to be calibrated. We calibrate the starting values
$\{\lambda_{s,0}^{i}\}_{s\in\mathcal{D},i\in\mathcal{I}}$ using the assumption
that this parameter is proportional to PPP-adjusted labor productivity in each
sector-country in our baseline year, 2014. This approach is similar to
\cite{buera_global_2020} who infer the location parameters based on total
factor productivity calculated as Solow-residuals. An alternative approach
would be to estimate both trade costs and the location parameters based on the
gravity equation as in \cite{levchenko_evolution_2016}. However, estimating these parameters
structurally for the small set of regions and sectors employed in the model,
respectively 10 and 6, will be complicated. Furthermore, with our approach and the chosen calibration of the ideas diffusion parameter we stay closer to observed data. As such, baseline values to which counterfactual experiments are applied are identical to actual values, ensuring that the impact of counterfactual experiments is not distorted. Therefore, we have chosen to identify the
location parameters of the Fr\'{e}chet distribution based on sectoral labor
productivity data.

We constructed a database of sectoral productivity by combining two sources: the
World Input-Output Database and the World Bank's Global Productivity Database.
We provide details of how we aggregated sectors and country groups in the
Appendix \ref{appendix: lambda}. Figure \ref{fig: labor_prod} shows the
distribution of the calibrated $\lambda_{s,0}^{i}$ parameters across
industries $i$ of each region $s$ in our model.

\begin{figure}[ptbh]
\centering
\includegraphics[scale=0.6]{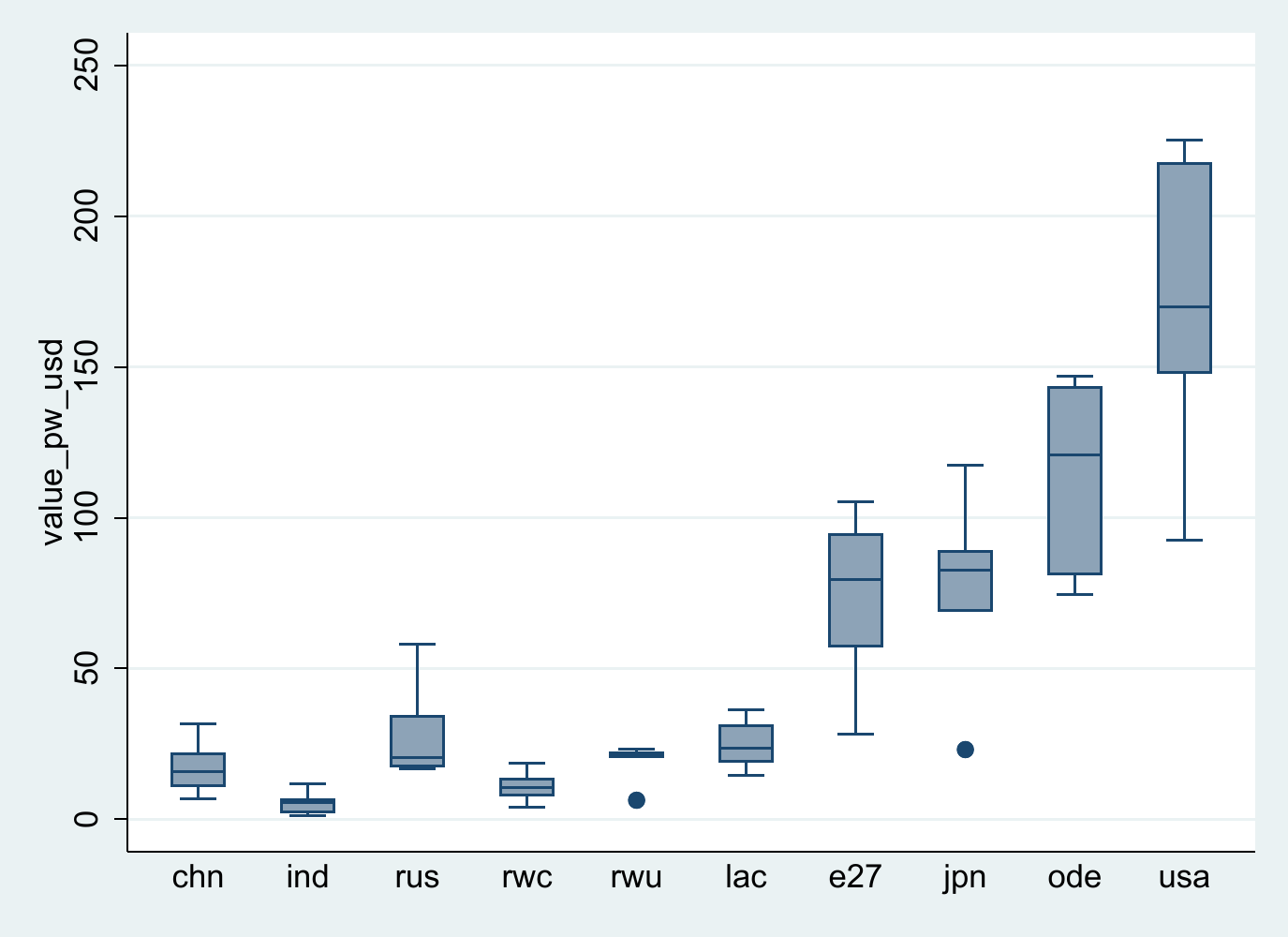}\caption{Distribution
of the calibrated $\lambda_{s,0}^{i}$ parameters across industries $i$ of each
region $s$ in our model. We assume that this parameter is proportional to
PPP-adjusted labor productivity in each sector-country. After the initial
period, the location parameter of the sector-country-specific Fr\'echet
distribution $\lambda_{s,t}^{i}$ evolves endogenously according to the model.}%
\label{fig: labor_prod}%
\end{figure}

\cite{buera_global_2020} set the growth rate of $\alpha_{t}$ equal to the
population growth rate of the US. We adopt the same heuristics and set the
growth rate of the autonomous arrival rate of ideas $\alpha_{t}$ at $1.18\%$
per year, equal to the projected global population growth rate from 2021 to 2040.

The ideas diffusion parameter $\beta$ is uniform across sectors and determined
based on model validation, using simulated methods of moments, as described
below. The variance of the growth rates of GDP rise as $\beta$ increases,
because at higher levels of $\beta$ there is more diffusion of ideas leading
to convergence of income levels thus implying larger differences between
growth rates of poor and rich countries. As discussed in Section
\ref{section: intuition}, countries starting with lower productivity
parameters have larger dynamic gains from trade. That effect increases in the
value of $\beta$\footnote{In the limiting point $\beta\rightarrow1$,
ideas diffuse instantly and every country that is not in autarky experiences
equal productivity gains in absolute terms. With equal gains in absolute
terms, those countries with lower productivity experience larger growth in
relative terms as $\beta$ increases.}.

We simulate the model from 2004 to 2019 imposing historical growth rates of
population and the labor force (based on IMF\ data) without further policy
changes varying the level of $\beta$ and evaluate for which values of $\beta$
the mean and variance of the growth rates of GDP are closest to the mean and
variance in historical data.\footnote{In this exercise the growth rate of
$\alpha_{t}$ is set at $1.18\%$ per year, the global population growth rate
between 2004 and 2019.} Formally, we are minimizing the following loss
function with $m$ the historical moment and $m\left(  \beta\right)  $ the simulated moment for either real GDP per capita growth or real GDP growth in
the 2004-19 period:%

\begin{align}
\label{eq: minbeta}\min_{\beta}  & \sum_{m\in\{\mu,\sigma\}} w^{GDPpc}%
(m(\beta)^{GDPpc,model}-m^{GDPpc,hist})^{2} +\\
&   (1-w^{GDPpc})(m(\beta)^{GDP,model}-m^{GDP,hist})^{2}\nonumber
\end{align}
where $w^{GDPpc}$ is the exogenous weight put on real GDP per capita growth,
rather than aggregate real GDP growth.

As a first step, we raise $\beta$ in steps of $0.05$ from $0$ to
$0.6$.\footnote{For values larger than $0.6$ the variance in the simulations
becomes unrealistically high, so these are disregarded.} This exercise
indicates that the simulated growth rates are closest to historical growth
rates for $\beta$ between $0.4$ and $0.5$ (See Appendix Table \ref{beta_0_06}%
). Therefore, as a second step, we simulate the model raising $\beta$ in steps
of $0.01$ from $0.4$ to $0.5$.

Figure (\ref{fig: loss_function}) plots the loss function (\ref{eq: minbeta})
for values of $\beta\in\lbrack0.4,0.5]$. Table \ref{tab: beta_version} in the Appendix displays additional summary statistics (mean, standard deviation, maximum and minimum) of average growth rates of GDP and GDP per capita between
2004 and 2019. Regardless of the weight, the loss function takes a minimum
when $\beta\in[0.44,0.45]$, a rather short interval. Taking an agnostic stance
and setting the weight $w^{GDPpc}=1/2$, we find that the loss function
(\ref{eq: minbeta}) is minimized for $\beta=0.44$, but results are virtually
unchanged by setting it to any value in the aforementioned interval.

\begin{figure}[pth]
\centering
\includegraphics[scale=0.7]{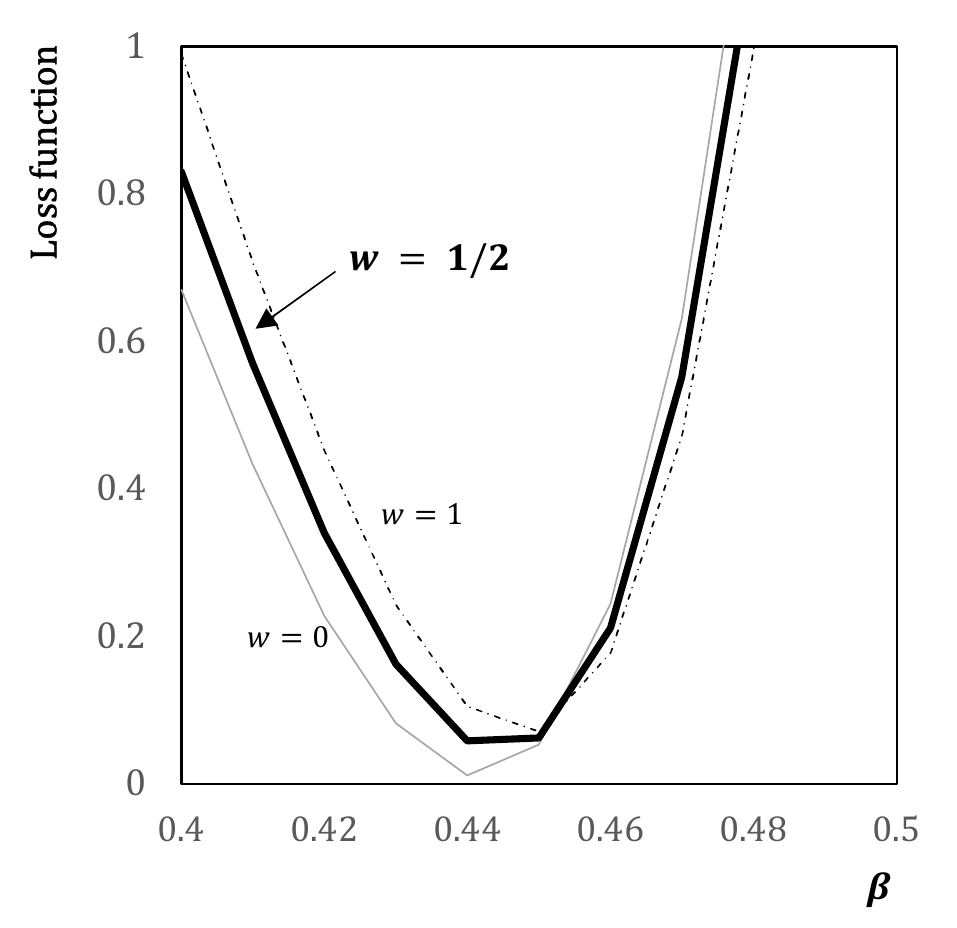}\caption{Plots of loss
function (\ref{eq: minbeta}) for values of $\beta\in[0.4,0.5]$. The solid grey
line shows the loss function with the parametrization of $w^{GDPpc}=0$, which
is minimized at $\beta=0.44$. The dotted grey line shows the loss function
with the parametrization of $w^{GDPpc}=1$, which is minimized at $\beta=0.45$
. The thicker black line shows the loss function with the parametrization of
$w^{GDPpc}=1/2$, which is minimized at $\beta=0.44$.}%
\label{fig: loss_function}%
\end{figure}

Figure \ref{scatter_hist_sim} compares projected GDP growth rates for
$\beta=0.44$ in individual regions with historical GDP growth rates. This figure shows that the simulated GDP growth rates are relatively close to the
historical growth rates, suggesting that also for individual regions the model does a good job at replicating historical growth rates. Furthermore, these
results can also be interpreted as an analysis of the under- and
overperformance of different regions compared to the projections of the model
with diffusion of ideas through trade.

\begin{figure}[pth]
\centering
\includegraphics[scale=0.4]{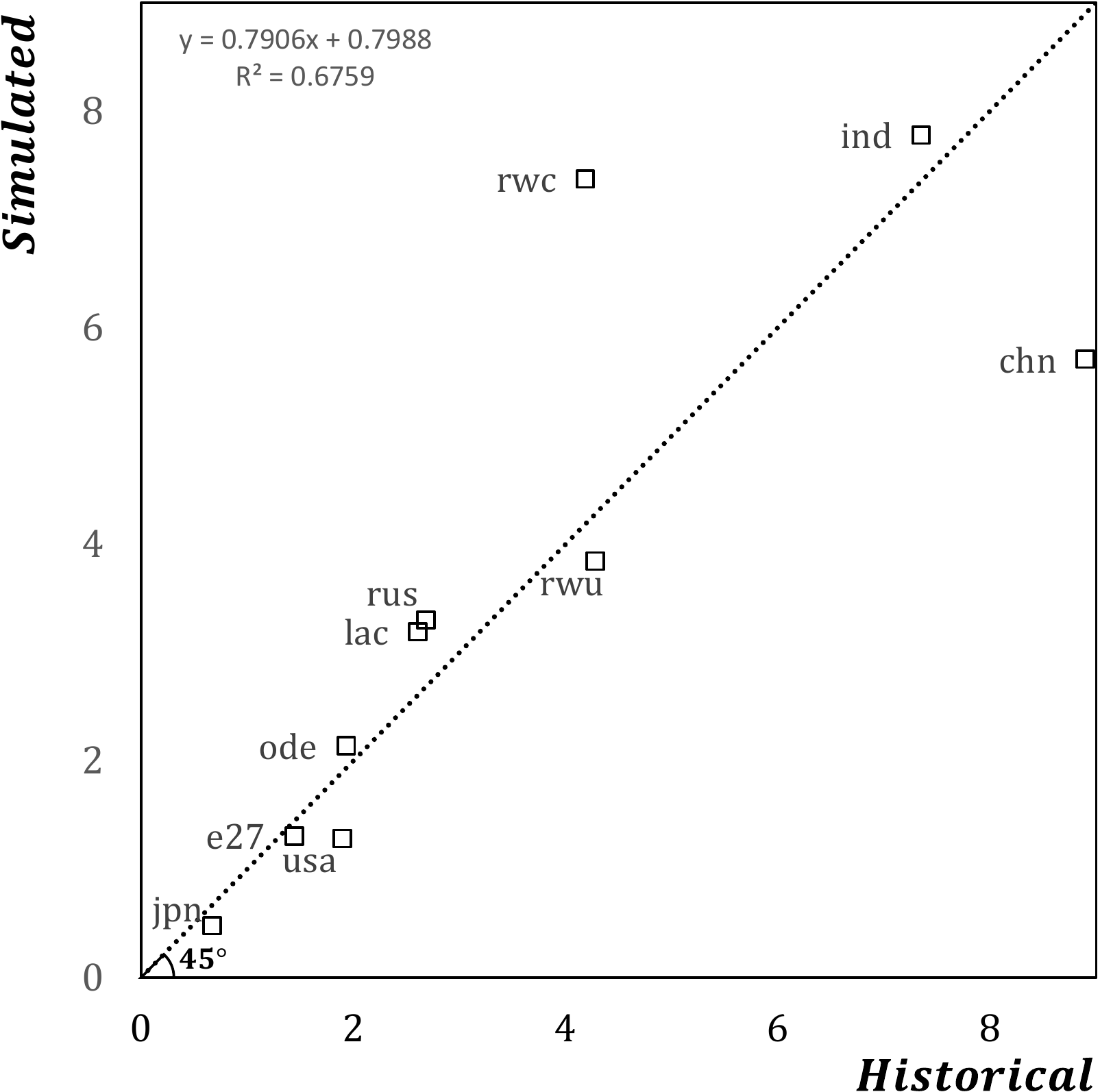}\caption{Historical and
simulated (for $\beta=0.44$) GDP growth rates (average 2004-2019)}%
\label{scatter_hist_sim}%
\end{figure}

In Appendix Table \ref{beta_0_06} we display the same summary statistics for
$\beta$ ranging from $0$ to $0.6$, in steps of $0.05$. This table makes clear
that calibration to historical data is important. An excessively high value of
$\beta$ will lead to higher than historically observed growth rates in
low-income regions. Conversely, a too low value $\beta$ will fail to show the
convergence patterns exhibited in the historical data.

In principle, the described approach could be employed to identify variation in
the ideas diffusion parameter between sectors employing data on growth rates
of sectoral output. However, consistent data on sectoral growth rates of real
output are not available for a broad set of countries employed in this study.

A possible concern of the employed method to calibrate $\beta$ is that
variation in growth between countries is not only driven by ideas diffusion
through trade implying technological catch-up. However, the calibration does
take variation in labor force growth into account and considers capital
accumulation as a source of growth. Furthermore, we follow an approach similar
to \cite{buera_global_2020} who calibrate $\beta$ based on targeting
US\ growth rates and find a value of $\beta$ substantially larger than in our
calibration, $0.6$.

\subsection{Motivation and Set-up of Policy Experiments\label{setup}}

\subsubsection{Motivation of Policy Experiments}

Our main motivation for simulating large-scale trade conflicts is the
possibility of receding globalization due to a political backlash. Challenges
to the international trade regime (and to globalization at large) might seem
like some circumstantial discontinuity in a long-run trend toward increasing
openness. However, as we show below, political scientists argue that there is
reason to believe that strategic geopolitical rivalries could trump economic
gains ---at least partially--- in their relationships between the U.S. and its
allies with China and Russia their allies.

There is ample evidence of substantial gains from trade openness, which can be
as large as 50\% of national income \parencite{ossa_why_2015} even in a static
setting. At the same time, recent empirical evidence about frictions in local
labor markets (\cite{autor_china_2013}; \cite{dix-carneiro_trade_2017})
highlights the distributional aspects of trade liberalization.

These concerns can translate into political grievances and may have led to an
increase in the number of populist and isolationist parties in Western
countries calling for less open trade policies
\parencite{colantone_trade_2018}. An increasing number of political parties
use anti-globalization rhetoric to rally the support of constituents that have
grievances against the distributional consequences of automation, structural
change, offshoring, and trade opening, as shown in the review of the political
science literature by \textcite{walter_backlash_2021}.

A clear example of the shift in the consensus during the last decade is the trade conflict between the U.S. and China, which started under the Trump
Administration. The economic discourse shifted away from emphasizing the gains from trade to a framing of trade as a zero-sum game and to the use of
national-security provisions of the international trade regime to engage in
protectionist policy-making\footnote{For a contemporaneous review of the
policies implemented, see \textcite{bown_trumps_2019}.}.

These geopolitical disputes are exemplified not only in the trade conflict
between the U.S. and China but also in industry-specific policy changes, such as the U.S. government pressuring allies against allowing the participation of Chinese telecommunications companies in new infrastructure developments or limited cooperation in science and education between the two countries \parencite{tang_decoupling_2021}.

\textcite{wei_towards_2019} provides a review of debates among Chinese
scholars. Some Chinese analysts see an escalating and continuous conflict
between China and the U.S. as a natural and ``structural'' development of a
shifting international system that is moving from a unipolar (the U.S. being
the only superpower) to a bipolar (China becoming a superpower on an equal
footing to the U.S.) balance of power\footnote{In the context above, the
balance of power between functionally equivalent states (the ``international
structure'') provides incentives for strategic behavior by governments that
try to maximize their power. We can interpret the unipolar or bipolar
configurations as equilibria and disruptions between such equilibria as
transition paths. This is known as the ``structural realism'' theory of
international politics, developed by Kenneth \textcite{waltz_theory_2010}.}.
They tie a scenario of a continuous confrontation between the U.S. and China
either to strategic geopolitical forces or to domestic political forces in
America \parencite{zhao_is_2019}.

In the West, political scientist Joseph S. Nye Jr.
\parencite*{nye_jr_power_2020} highlights that, while an abrupt decoupling
between the U.S. and China is unlikely, both parties will try to decrease
their (inter-)dependence with respect to each other's actions, except where
the costs of disengagement are too high to bear\footnote{Nye Jr. is mostly
known for his joint work with Robert Koehane on \textquotedblleft complex
interdependence\textquotedblright\ during the post-World War II era
\parencite{keohane_power_2011}. The authors focus their analysis on the
creation of international rules and practices in a world in which the use of
military force is very costly due to interdependence between multiple agents
that engage both internationally and domestically. For instance, a great
degree of trade integration increases the costs (and decreases the
probability) of outright military conflict.}. American policymakers and
academics also motivate the conflict between China and the U.S. on
geopolitical grounds. Although the tone of confrontation is stronger when
coming from right-of-center policymakers and scholars\footnote{See, for
instance, the remarks of former White House Trade Council Director to Congress
\parencite{navarro_peter_white_2018} or a policy blueprint for decoupling by
\textcite{scissors_partial_2020}, who is a scholar at the America Enterprise
Institute, a right-of-center think tank.}, both sides of the political
spectrum in the U.S. discuss the readjustment of the economic relationship with China due to geopolitical concerns \parencite{wyne_how_2020}.

The 2022 War in Ukraine and the global-scale retaliation against the Russian
Federation that followed is another example of how geopolitical interests can take precedence over gains from economic integration. The escalation began in 2014, after Russia's annexation of Crimea. The U.S. and its allies approved several rounds of sanctions against Russia, culminating in its expulsion from the G8. The confrontation reached another level in the aftermath of Russia's invasion of Ukraine. Western countries imposed sanctions on banking transactions, froze foreign reserves, and closed the airspace for Russian planes. In March 2022, the G7 and the European Union revoked their recognition of Russia's most-favored-nation status, opening the door for large tariff increases, and limited the operation of multinational corporations Russian subsidiaries\footnote{For a timeline of sanctions against Russia in the
context of the War in Ukraine, see this website maintained by the Peterson
Institute for International Economics: \url{https://www.piie.com/blogs/realtime-economic-issues-watch/russias-war-ukraine-sanctions-timeline}}.

Like in the case of China, the relationship between Russia and the West can
also be interpreted through the lens of a strategic game among great powers.
Scholars argue about the geopolitical nature of the conflict
\parencite{mearsheimer_why_2014} and the ``reawakening'' of geopolitics
\parencite{weber_russias_2022}. Therefore, in either case, sanctions and trade
conflicts fall within the larger backdrop of a strategic confrontation. These
simultaneous conflicts can potentially be interpreted as a larger clash
between the U.S. and its allies --- a Western bloc; and Russia, China, and
their allies --- an Eastern bloc. As scholars have argued, we could be
observing the emergence of a ``China-Russia entente''
\parencite{lukin_russiachina_2021}, which could lead to a ``New Cold War'' \parencite{abrams_new_2022}.

We use these facts as motivation to apply our model to conduct hypothetical
policy experiments of trade decoupling between East and West: namely, to
simulate the effects of large-scale geopolitical conflicts between these
blocs, in which players try to limit the level of interdependence between each
bloc due to political drivers, even if that leads to economic costs.

\subsubsection{Set-up of Policy Experiments}

We are agnostic about the degree of future decoupling between East and West.
Nonetheless, the fact that international relations scholars envisage
disengagement as a real possibility underscores that estimating the potential
economic consequences of decoupling is an important exercise. As our model
highlights, changes in trade patterns and sourcing decisions have not only
static effects, but also dynamic effects with respect to potential growth and
innovation. Our policy experiments try to disentangle the static and dynamic
costs of decoupling.

In order to develop the decoupling scenarios, we classify different regions as
belonging to either a Western or an Eastern bloc. We do so by resorting to the Foreign Policy Similarity Database, which uses the UN General Assembly voting records for a large set of countries to calculate foreign policy similarity
indices for each country pair \parencite{hage_choice_2011}. Intuitively, the
index takes countries who vote similarly in the United Nations (compared to
the expected level of similarity of a random voting pattern) as being similar
in their foreign policy.

We ordered country groups in terms of their foreign policy similarity with
China and the United States in order to place the ten regions of our model
either in a Western (U.S.) or Eastern (China) bloc. Figure \ref{fig:map} shows that Europe, Canada, Australia, Japan, South Korea
fall in the Western bloc. Latin America and Sub-Saharan Africa fall somewhere
in between, with the former being closer to the U.S. than the latter. India,
Russia, and most of North Africa and Southeast Asia fall closer to China.

\begin{figure}[th]
\centering
\includegraphics[width=\textwidth]{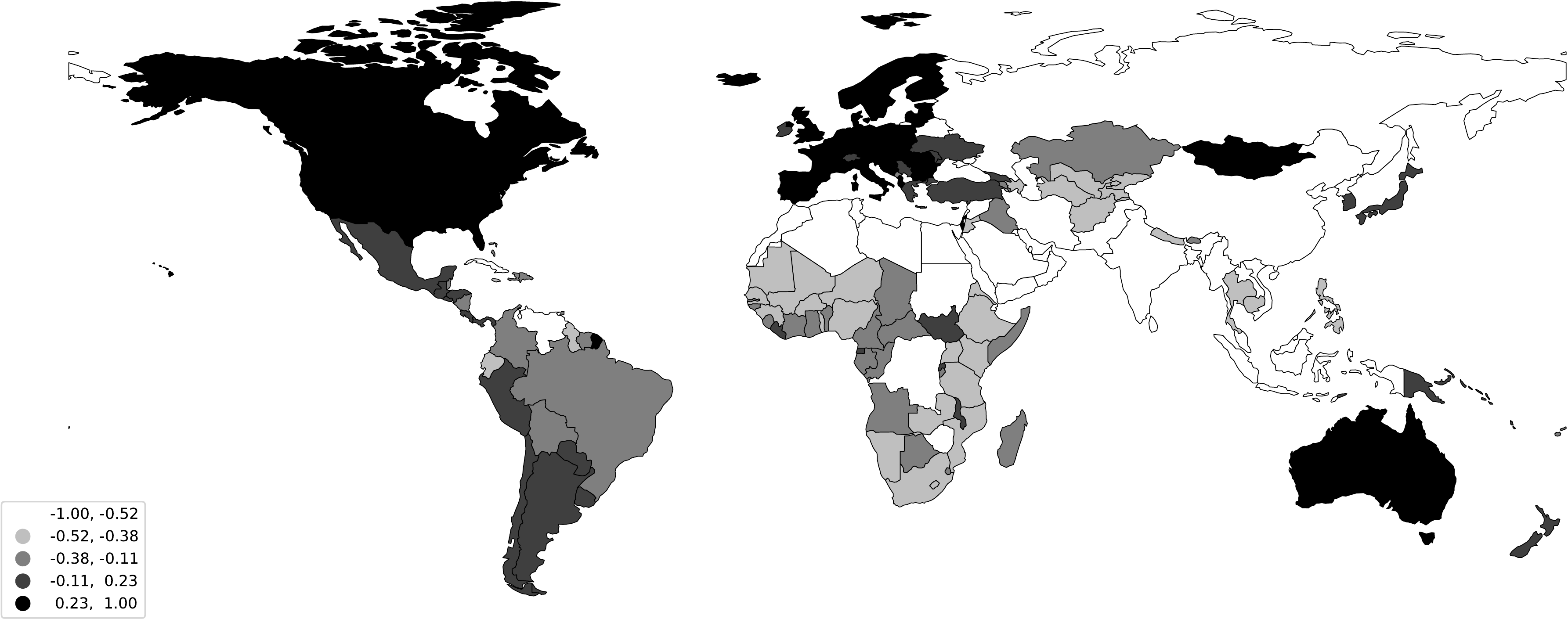}\caption{Differential Foreign
Policy Similarity Index. Values are normalized such that $1$ represents
maximum relative similarity with the U.S. and $-1$ represents maximum relative
similarity with China. {\protect\footnotesize {\ The map shows the difference
between pairwise similarity indices $\kappa_{i,US} - \kappa_{i,China}$. The
parameter $\kappa_{i,j}$ represents the foreign policy similarity of countries
$i,j$, based on vote similarity in the United Nations General Assembly. Given
vote possibilities $n,m \in\{ 1, \cdots, k\}$, one can calculate a matrix $P =
[p_{nm}]$, where entry $p_{nm}$ represents the share of votes in which country
$i$ took position $n$ and country $j$ took position $m$. Given matrix $P$,
$\kappa_{i,j} = 1 - \sum_{m\neq n} p_{mn} / \sum_{m\neq n} p_{m} p_{n}$, where
$p_{m},p_{n}$ are expected marginal propensity of any country to take position
$m,n$ at a random vote. For more details, see \textcite{hage_choice_2011}.}}}%
\label{fig:map}%
\end{figure}

As a robustness exercise, we repeated the same procedure but replaced China
for Russia as the center of gravity of the Eastern bloc. We plot the two
Differential Foreign Policy Similarity Indices for each country in our sample
in Figure~\eqref{fig:similarity_comp}. Most countries fall very close to the
45-degree line (i.e., the regression coefficient is close to 1) and the
correlation between the two series is very high. This suggests that,
qualitatively, results will be very similar using either Russia or China as
the main country in the Eastern bloc.

\begin{figure}[th]
\centering
\includegraphics[scale=0.7]{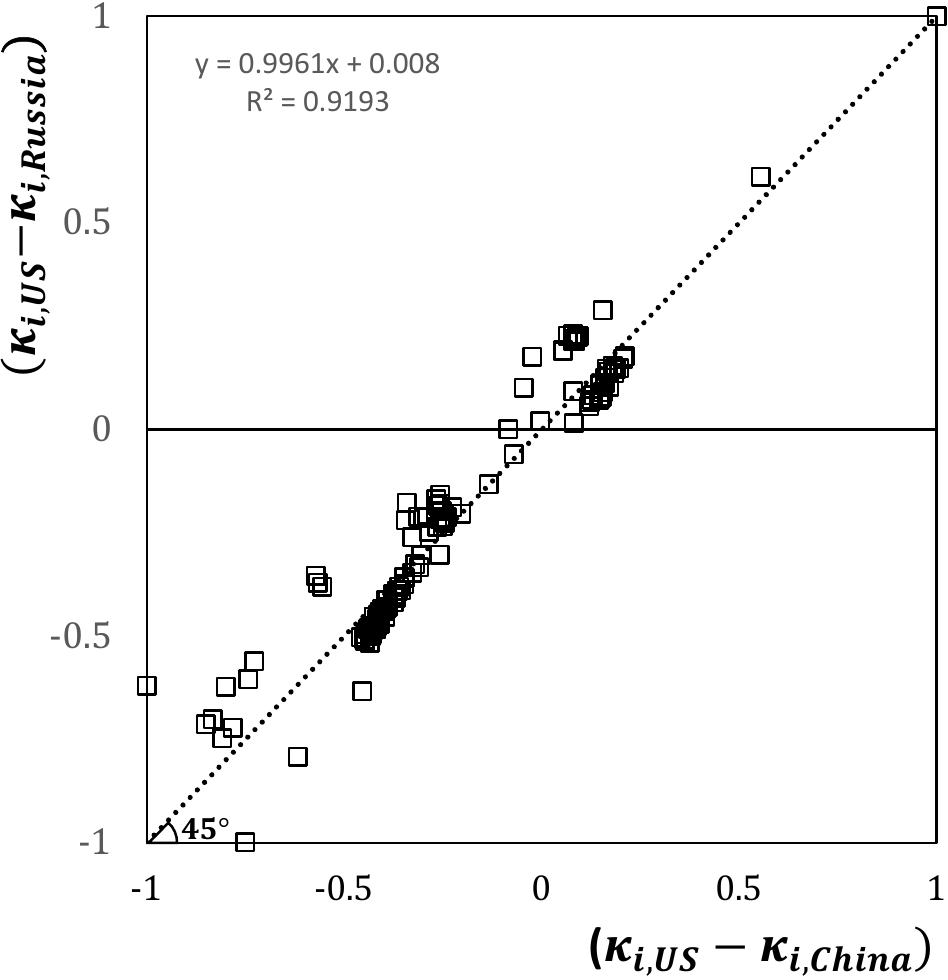}\caption{Differential
Foreign Policy Similarity Index. Values are normalized such that $1$
represents maximum relative similarity with the U.S. and $-1$ represents
maximum relative similarity with China or Russia. See the caption of Figure
\ref{fig:map} for further details on the indices.}%
\label{fig:similarity_comp}%
\end{figure}

After classifying the countries into Eastern and Western geopolitical blocs,
we design two different policy experiments. We first increase iceberg trade
costs $\tau_{sd,t}^{i}$ to a point where virtually all of the trade happens
exclusively within each bloc. In total, we increase bilateral trade costs by
$\sim160$ percentage points. We assume that the increases in trade costs is
permanent. We label this scenario \textbf{full decouple}. This provides an
important limiting case that can be useful for putting bounds on potential effects.

The second scenario relies on work by \textcite{nicita_cooperation_2018}, who
estimate that a move from cooperative to non-cooperative tariff setting would
increase average tariffs by 32 percentage points globally. We simulate what
would happen if countries kept cooperative tariff setting within their trade
blocs but moved to non-cooperative tariff setting across trade blocs. For
simplicity, we assume that regions in different blocs increase bilateral
tariffs $tm_{sd,t}^{i}$ by the globally average increase in tariffs when
moving from cooperative to non-cooperative tariffs: 32 percentage point
increases in tariff rates against regions outside the bloc. Again, we assume
that the increases in trade costs are permanent. These tariff increases seem
high. However, the average tariff increase in the China-U.S. trade war was
higher than 21pp. We call this scenario \textbf{tariff decouple}.

Besides the full and tariff decouple scenarios, we explore two additional
policy experiments. First, we evaluate what the impact would be of a switch of
the region Latin America and Caribbean (LAC)\ from the Western bloc to the
Eastern bloc. This sheds some light on the question many countries would face
if the decoupling would aggravate: which bloc to stay closest to?

Second, we explore a less hypothetical scenario: trade decoupling only in the
electronic equipment sector. We perform is a full decoupling between the
original blocs (with LAC in the Western bloc) but restrict the increase in
iceberg trade costs $\tau_{sd,t}^{i}$ only to the electronic equipment
($\mathtt{elm}$) sector. This scenario is motivated by U.S. and Chinese
authorities being increasingly at loggerheads with each other in the
technological arena.

One important example of this process has been the conflict involving Chinese
telecom giant Huawei Technologies. Since 2019, American corporations have been
banned from doing business with Huawei. In a similar move, the New York
Exchange delisted China Unicom, China Mobile, and China Telecom. Despite legal
challenges and a new administration, as of April 2021, neither decision has
been reversed.

Additionally, the U.S. has been using its foreign policy arsenal to pressure
allies to join them in limiting Chinese telecom companies' reach. In
particular, there is a desire to limit Chinese participation in 5G technology
auctions, citing national security and privacy concerns\footnote{North
American Treaty Organization (NATO) researchers \textcite{kaska_huawei_2019}
review the arguments put forth from a Western national security perspective.
This topic is extremely contentious and some Chinese commentators argue that
the U.S. is using national security concerns as an excuse to implement
industrial policy.}. So far, Australia, the United Kingdom and some European
allies have chosen to ban or limit Chinese participation in technological auctions.

This conflict suggests that a large increase in trade costs between the U.S.
allies and Chinese allies regarding technological equipment is a positive
probability scenario in the future. In this case, decoupling would mean a near-total separation of the electronic equipment sectors of the two blocs.

Huawei and Google break of their business connections after the U.S.
government sanctions against the Chinese corporation is a good illustration of what this separation could look like in real life. Huawei used Google's
\textit{Android} ecosystem in their smartphones, which gave their users access
to Google-approved updates and apps. After the ban issued by the Trump
administration, however, Google announced it would comply with the U.S.
government directives and Huawei was forced to shift away from Google software
and design their own operating system \textit{HarmonyOS}.

Since this separation is driven primarily by regulation rather than tariffs,
it is appropriate to think of it as an increase in iceberg trade costs
$\tau_{sd,t}^{i}$ between blocs in the electronic equipment sector and so this
is the scenario we will explore.

\section{Main Results\label{results}}

We have four main scenarios. We simulate full decouple and tariff decouple,
defined as explained above. We simulate either scenario both with and without
diffusion of ideas, in order to assess the additional impact of the diffusion
mechanism. After a discussion of the results of the full and tariff decouple
scenarios, we compare the impact of decoupling on productivity in a
multi-sector and single-sector setting. We finish this section with a
description of the results of the two additional policy experiments that
restrict decoupling to the electronics and equipment sector or change
LAC\ from the Western bloc to the Eastern bloc.

In the results below we report the results of a comparison of the simulation
results with and without policy experiments. We first simulate the dynamic
world economy with no policy change, then do the same with the policy change,
and report the long-run cumulative percentage difference between the two:
$\hat{x}=\sum_{t=p}^{T}(x_{t}^{\prime}-x_{t})/\sum_{t=p}^{T}x_{t}$, where $p$
is the date of the first policy change, $x_{t}^{\prime},$ $x_{t}$ are the
values of variable $x$ with and without the policy change, respectively.

\subsection{Full and Tariff Decouple}

As expected, all scenarios show large negative impacts on cross-bloc trade
after the introduction of the policy intervention. In the \textbf{full
decouple} scenario, trade between the countries in the Western bloc and the
Eastern bloc is virtually shut down, with imports and exports dropping by
$98\%$. Those countries also shift a substantial part of their trade to the
U.S., with trade flows increasing anywhere between $10-42\%$ depending on the
region and scenario. The domestic spending share in the U.S. increases by
about $7\%$. The converse happens in the Eastern bloc but with larger
dispersion across regions. Trade with the U.S. drops by $65-90\%$ while trade
with China increases by $9-60\%$. The domestic trade share in China increases
by $3\%.$ The \textbf{tariff decouple} scenario yields qualitatively similar
results but with smaller magnitudes. We show the results by region and
scenario in Figure \ref{fig:trade}.

\begin{figure}[ptb]
\centering
\includegraphics[scale=0.7]{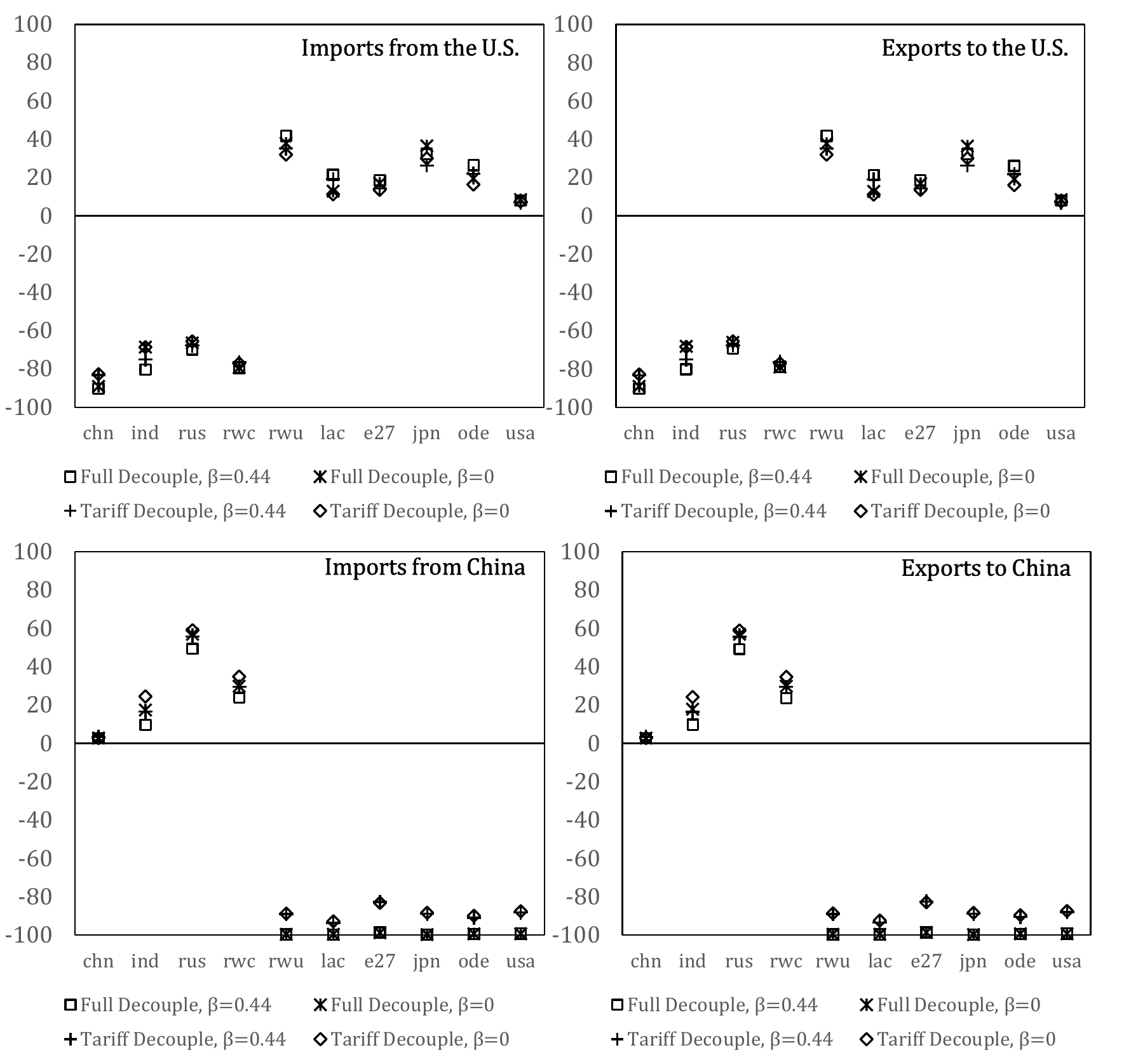}\caption{Cumulative Percentage
Change in Trade Flows with China and the United States, respectively, after
the policy change, by 2040. {\protect\footnotesize {\ \textit{Full Decouple}
increases iceberg trade costs $\tau_{sd,t}^{i}$ by $160$ percentage points.
\textit{Tariff decouple} increases bilateral tariffs $tm_{sd,t}^{i}$, across
groups, by $32$ percentage points. $\beta$ is a parameter that controls the diffusion of ideas according to equation \ref{eq: final_law_of_motion},
assumed to be homogeneous across sectors. Country codes: \texttt{chn}, China;
\texttt{ind}, India; \texttt{rus}, Russia; \texttt{rwc}, Rest of Eastern bloc;
\texttt{rwu}, Rest of Western bloc; \texttt{lac}, Latin America; \texttt{e27},
European Union; \texttt{ode}, Other Developed; \texttt{usa}, United States.
Tables with the values for these charts can be found in the Appendix.}}}%
\label{fig:trade}%
\end{figure}

One of the reasons behind the asymmetry in trade decreases between blocs is
the assumption of a fixed trade-balance-to-income ratio in all regions but
one. This implies that regions with large trade surpluses will shift
proportionally less of their trade flows away from regions in other trade
blocs in order to satisfy the fixed trade balance assumption.

Figure (\ref{fig: welfare}) shows that both the increases in iceberg trade
costs (full decouple) and retaliatory tariff hikes (tariff decouple) induce
substantial welfare decreases for all countries. The effects, however, are
asymmetric. While welfare losses in the Western bloc range anywhere between
$-1\%$ and $-8\%$ (median: $-4\%$) in the Eastern bloc it falls in the $-8\%$
to $-12\%$ range (median: $-10.5\%$).

\begin{figure}[th]
\centering
\includegraphics[scale=0.5]{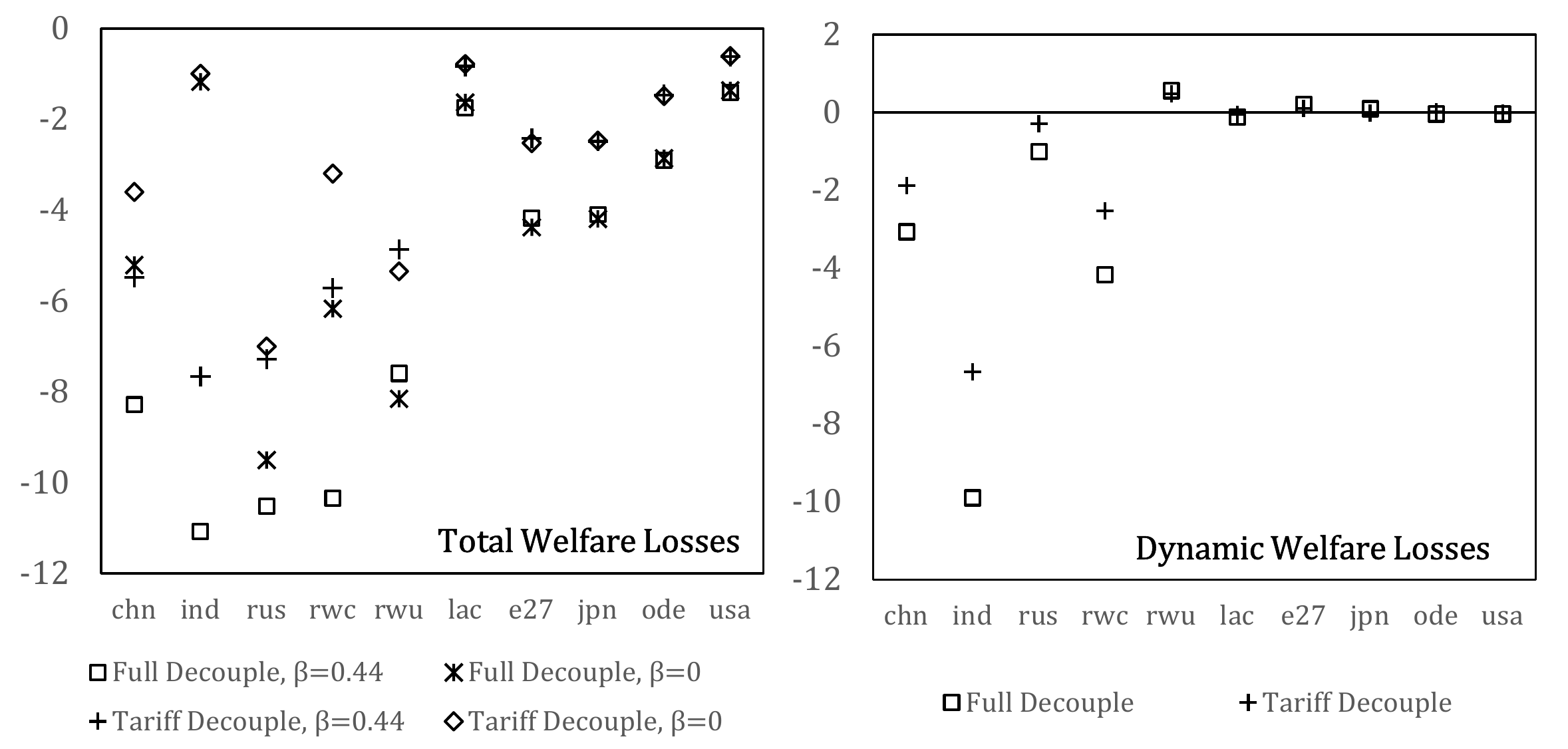}\caption{Cumulative
Percentage Change in Real Income, after the policy change, by 2040.
{\protect\footnotesize {\ \textit{Full Decouple} increases iceberg trade costs
$\tau_{sd,t}^{i}$ by $160$ percentage points. \textit{Tariff decouple}
increases bilateral tariffs $tm_{sd,t}^{i}$, across groups, by $32$ percentage
points. $\beta$ is a parameter that controls the diffusion of ideas
according to equation \ref{eq: final_law_of_motion}, assumed to be homogeneous
across sectors. Country codes: \texttt{chn}, China; \texttt{ind}, India;
\texttt{rus}, Russia; \texttt{rwc}, Rest of Eastern bloc; \texttt{rwu}, Rest
of Western bloc; \texttt{lac}, Latin America; \texttt{e27}, European Union;
\texttt{ode}, Other Developed; \texttt{usa}, United States. Tables with the
values for these charts can be found in the Appendix.}}}%
\label{fig: welfare}%
\end{figure}

The underlying factor driving the divergence in results between the two blocs
is a difference in the evolution of productivity, represented by the scale
parameter of the Fr\'{e}chet distribution of different sectors. Sourcing goods
from high productive countries puts domestic managers in contact with better
quality designs that inspire better ideas through innovation or imitation.

Importantly, the dynamics governed by equation (\ref{eq: final_law_of_motion})
incorporate the input-output structure of production, such that domestic
managers in each sector innovate in proportion to the quality and share of
their inputs. Losing access to high quality designs does not only lead to
static losses, but also to a lower level of future innovation, which implies
larger dynamic losses. Additionally, the input-output structure of the model
implies that cutting ties to innovative regions is particularly costly if the
destination country sources many intermediate inputs from such regions prior
to the policy change.

\begin{figure}[th]
\centering
\includegraphics[scale=0.55]{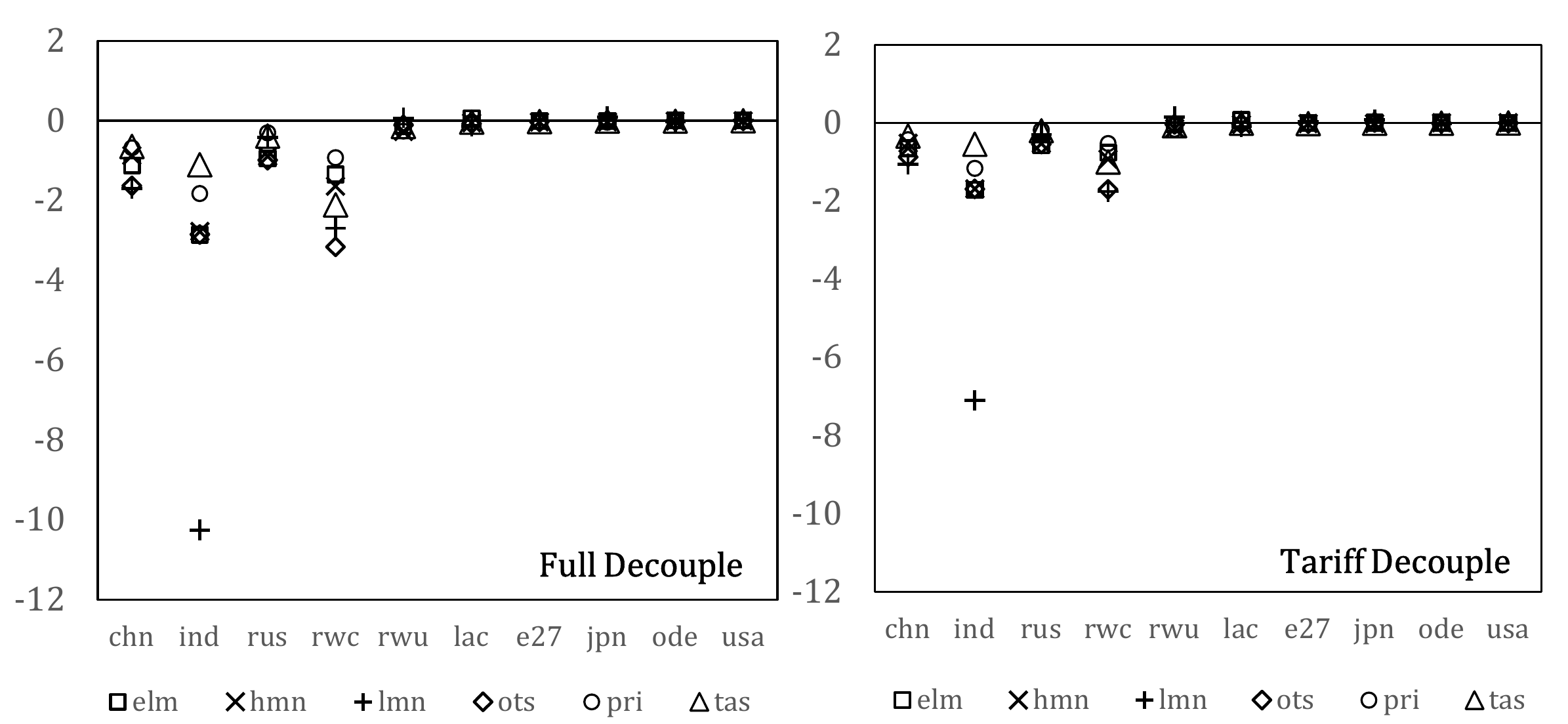}\caption{Cumulative Percentage
Change in the Fr\'echet Distribution location parameter $\lambda_{d,t}^{i}$,
after policy change, by 2040. {\protect\footnotesize {\ \textit{Full Decouple}
increases iceberg trade costs $\tau_{sd,t}^{i}$ by $160$ percentage points.
\textit{Tariff decouple} increases bilateral tariffs $tm_{sd,t}^{i}$, across
groups, by $32$ percentage points. $\beta$ is a parameter that controls the
the diffusion of ideas according to equation \ref{eq: final_law_of_motion},
assumed to be homogeneous across sectors. Country codes: \texttt{chn}, China;
\texttt{ind}, India; \texttt{rus}, Russia; \texttt{rwc}, Rest of Eastern bloc;
\texttt{rwu}, Rest of Western bloc; \texttt{lac}, Latin America; \texttt{e27},
European Union; \texttt{ode}, Other Developed; \texttt{usa}, United States.
Sector codes: \texttt{elm}, Electronic Equipment; \texttt{hmn}, Heavy
manufacturing; \texttt{lmn}, Light manufacturing; \texttt{ots}, Other
Services; \texttt{pri}, Primary Sector; \texttt{tas}, Business services.
Tables with the values for these charts can be found in the Appendix.}}}%
\label{fig:lambda}%
\end{figure}

For those reasons, in our policy experiments, countries in the Eastern bloc
that currently have a lower level of productivity and have larger ties with
innovative countries have larger losses. By looking at results in Figure
\ref{fig:lambda}, one can see the stark contrast between the differential
evolution of $\lambda_{d,t}^{i}$ for those countries in the Western bloc and
those in the Eastern bloc. By cutting ties with richer and innovative markets
such as OECD countries, destination countries such as China, India, and parts
of Asia and Africa shift their supply chains towards lower-quality inputs,
which, in turn, induces less innovation. By contrast, while countries in the
Western bloc do suffer welfare losses, their innovation paths are virtually
unchanged after decoupling, suggesting that nearly all of their losses are
static, rather than dynamic.

There is large dispersion across both sectors and countries in differential
productivity losses. The two most affected regions are India and the ``rest of
the Eastern bloc'' region. Starting from a lower income level than China and
Russia, those regions have a much slower productivity catch-up after severing
trade ties with the West. Sectors with larger supply chain linkages to the
West prior to the policy change, such as manufacturing in India, experience
larger losses.

Among those regions in the Eastern bloc, differential productivity losses are
larger in the manufacturing sectors ($-1.5\%$ and $-3\%$ with full decoupling
and tariff decoupling, respectively; this includes \texttt{elm}, \texttt{lmn},
and \texttt{hmn}) than in the services ($-0.8\%$ and $-1.6\%$, respectively;
\texttt{ots} \texttt{tas}) or primary ($-0.5\%$ and $-1\%$, respectively;
\texttt{pri}) sectors.

Finally, we address the contrast between the static effect (when the diffusion
of ideas mechanism is shut down) and the dynamic effect. For the two poorer
regions of the Eastern bloc, dynamic losses far outsize static losses, which
can be explained by the loss of access to higher-quality inputs. In the
right panel of Figure (\ref{fig: welfare}), we show the dynamic losses for
each region.

In India, static welfare losses amount to $1-2\%$ while dynamic losses range
from $7-10\%$, depending on the decoupling scenario. Static losses to real
income are small because India is a relatively large country and its domestic
trade share in the market equilibrium is large, which limits the range of
goods affected by changes in terms of trade. However, because it is relatively
poor, its losses in the diffusion of ideas version of the model are very
large. By severing ties with the Western bloc, it limits the role of
trade-induced innovation, which is a by-product of having access to high-quality suppliers.

By contrast, in Russia including dynamics leads only to small additional
effects: welfare losses are very similar with or without the ideas diffusion
mechanism. As explained above, this stems both from a higher income starting
point and relatively limited input-output linkages with the West.

\subsection{Diffusion Inefficiencies a Multi-sector vs. a Single-sector
Framework}

In Section \ref{section: intuition}, we stressed that, except in knife-edge
cases, within- and between-sector inefficiencies accumulate as the number of
countries and sectors increase. The concavity of the diffusion process implies
that \textit{total} trade shares being at their optimal points is no longer
sufficient for optimal diffusion. Optimal diffusion requires trade shares to
be at their optimal points \textit{at every sector}. This suggests that, in
most cases, diffusion inefficiencies should increase with the number of sectors.

\begin{figure}[th]
\centering
\includegraphics[scale=0.5]{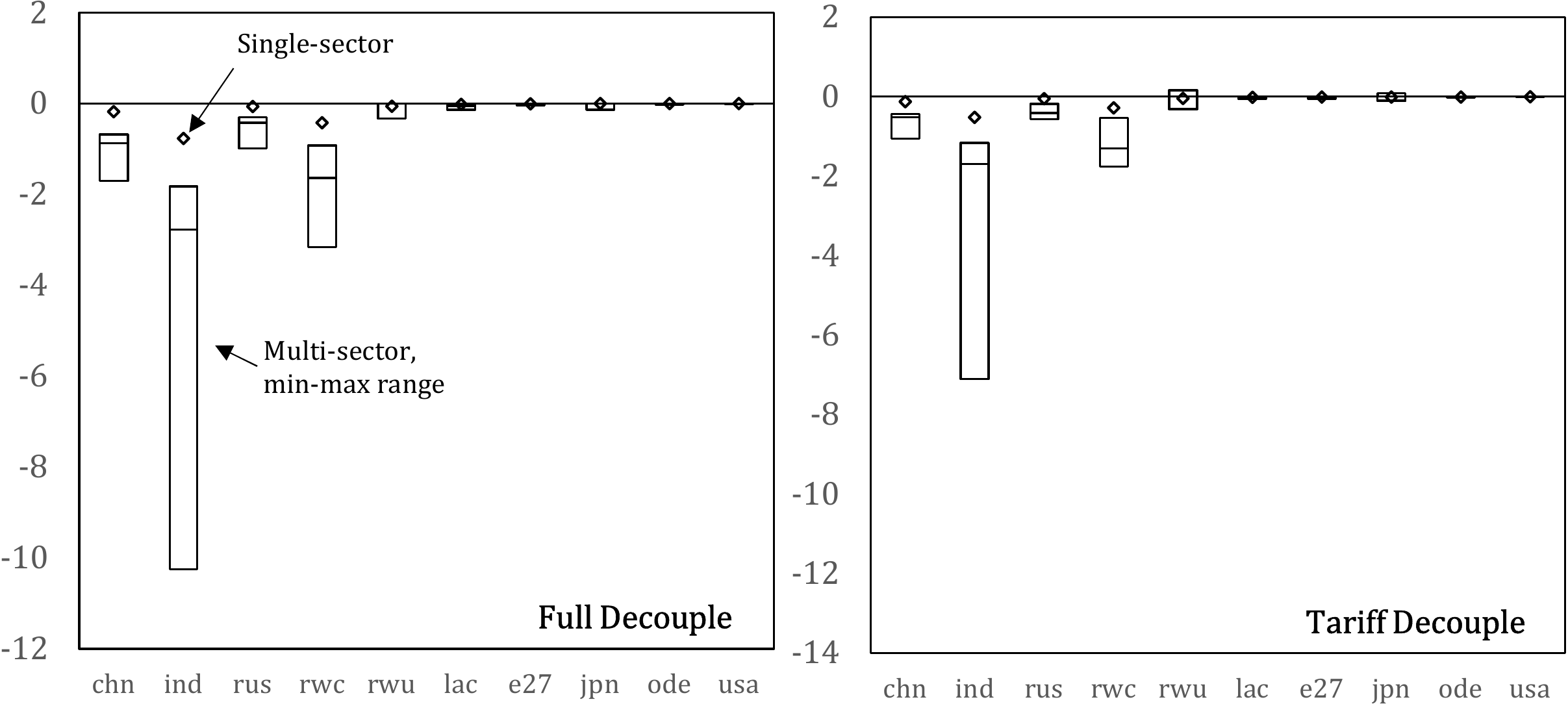}\caption{Multi-sector vs.
Single-sector: Cumulative Percentage Change in the Fr\'echet Distribution
location parameter $\lambda_{d,t}^{i}$, after policy change, by 2040.
{\protect\footnotesize {\ \textit{Full Decouple} increases iceberg trade costs
$\tau_{sd,t}^{i}$ by $160$ percentage points. \textit{Tariff decouple}
increases bilateral tariffs $tm_{sd,t}^{i}$, across groups, by $32$ percentage
points. Country codes: \texttt{chn}, China; \texttt{ind}, India; \texttt{rus},
Russia; \texttt{rwc}, Rest of Eastern bloc; \texttt{rwu}, Rest of Western
bloc; \texttt{lac}, Latin America; \texttt{e27}, European Union; \texttt{ode},
Other Developed; \texttt{usa}, United States. Sector codes: \texttt{elm},
Electronic Equipment; \texttt{hmn}, Heavy manufacturing; \texttt{lmn}, Light
manufacturing; \texttt{ots}, Other Services; \texttt{pri}, Primary Sector;
\texttt{tas}, Business services. Tables with the values for these charts can
be found in the Appendix.}}}%
\label{fig:multi-sec}%
\end{figure}

Our empirical results confirm that theoretical intuition. Figure
\ref{fig:multi-sec} contrasts the results of either the \textbf{full decouple}
or the \textbf{tariff decouple} scenarios under the baseline specification
presented in the previous section and an alternative simulation in which we
collapse the model to a single-sector framework.

In both scenarios, countries in the Easter bloc face higher cumulative diffusion inefficiencies (as
measured by the reduction in the Fr\'echet parameters $\lambda_{d,t}^{i}$) in a multi-sector framework. In fact, the single-sector dynamic productivity losses are outside the min-max range of the sectoral productivity changes for all countries in the Eastern bloc. These results
underscore one important takeaway of this paper: modeling trade diffusion in a simplified single-sector framework can underestimate the level of dynamic losses induced by an increase in trade costs.

\subsection{Consequences of Bloc Membership}

In this section, we consider the consequences of moving one of the regions
---Latin America and the Caribbean (LAC) ---from the Western to the Eastern
bloc. Intuitively, we expect that, by losing access to the highest
productivity suppliers, LAC will experience less productivity growth.
Nonetheless, the quantitative exercise allows us to have a sense of the
magnitude induced by the change in group membership.

\begin{figure}[th]
\centering
\includegraphics[scale=0.7]{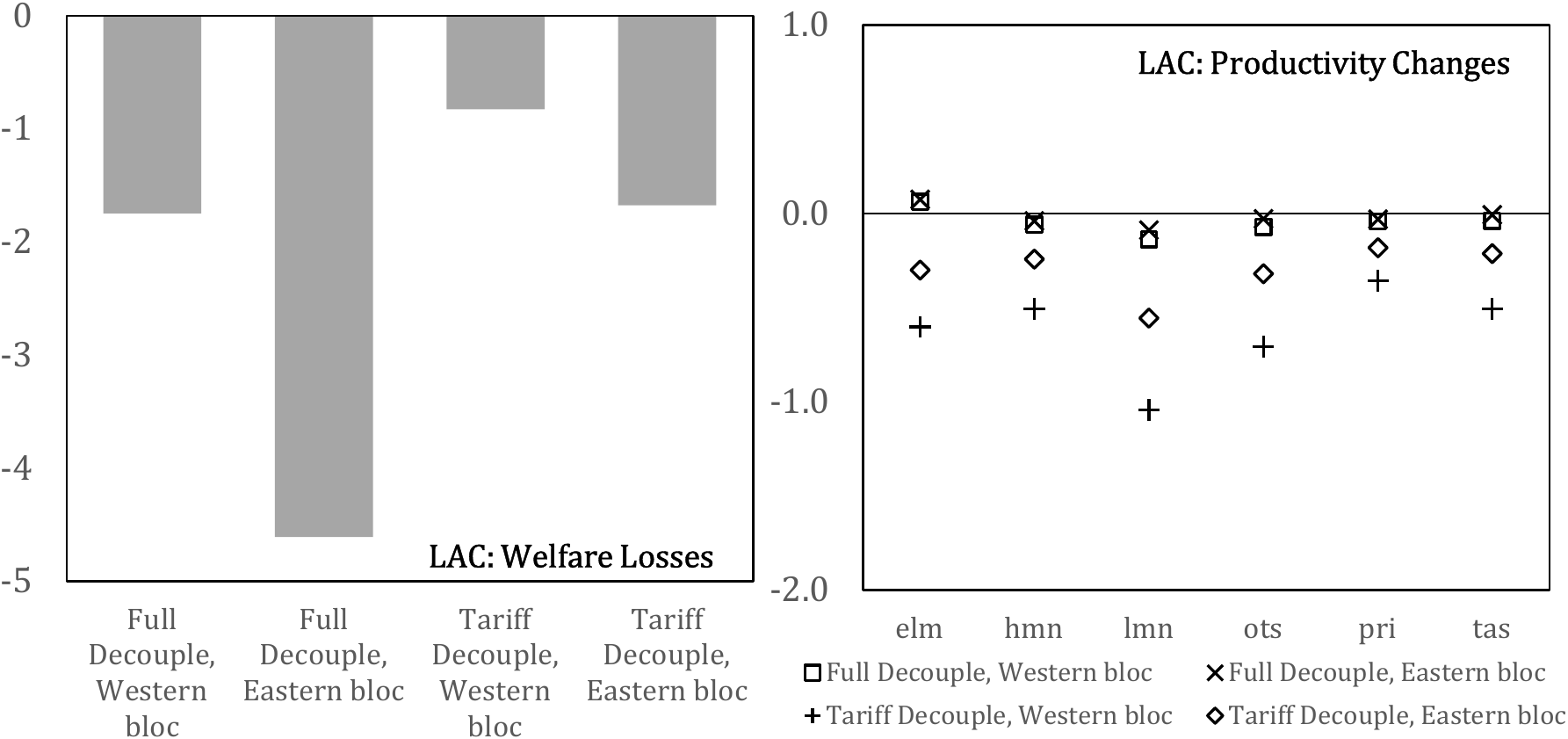}\caption{Left Panel: Cumulative
Percentage Change in Real Income in LAC Region, by scenario. Right Panel:
Cumulative Percentage Change of the Fr\'echet Distribution scale parameter
$\lambda_{d,t}^{i}$ in LAC Region, by scenario.
{\protect\footnotesize {\ \textit{Full Decouple} increases iceberg trade costs
$\tau_{sd,t}^{i}$ by $160$ percentage points. \textit{Tariff decouple}
increases bilateral tariffs $tm_{sd,t}^{i}$, across groups, by $32$ percentage
points. In all cases, we set the parameter that controls the diffusion of
ideas to $\beta=0.44$. Sector codes: \texttt{elm}, Electronic
Equipment; \texttt{hmn}, Heavy manufacturing; \texttt{lmn}, Light
manufacturing; \texttt{ots}, Other Services; \texttt{pri}, Primary Sector;
\texttt{tas}, Business services. Tables with the values for these charts can
be found in the Appendix.}}}%
\label{fig: lac}%
\end{figure}

Figure \ref{fig: lac} compares the results of identical decoupling scenarios,
simulating either \textit{full decouple} or \textit{tariff decouple}. The only
difference is LAC bloc membership. As expected, most of the changes are
concentrated in the LAC region. The left panel of Figure \ref{fig: lac} shows
that welfare losses in LAC are about $100-150\%$ larger when it is included in
the Eastern bloc, for both scenarios. The domestic trade share in LAC\ is
virtually identical under both settings (with LAC\ in the Western or in the
Eastern bloc), implying similar static welfare losses. This suggests that the
increased losses from switching blocs stem almost entirely from dynamic losses.

Moving LAC to the Eastern bloc reduces the welfare losses of decoupling in
India and China by about $2p.p.$ ($16\%$) and $1p.p.$ ($15\%$), respectively
(results not reported). The reason is twofold. First, LAC has a higher income
than India and the Rest of the Eastern bloc. All else equal, on average, its
inclusion in the bloc raises average productivity and decreases dynamic
losses. Second, lower tariff or iceberg trade costs between the Eastern bloc
and LAC induce lower static losses for those countries.

The right-hand side panel of Figure \ref{fig: lac} shows the differential
productivity changes in the LAC region for different sectors. When LAC is
included in the Western bloc, there are essentially no dynamic productivity
losses in any sector: the evolution of the Fr\'echet Distribution scale
parameter $\lambda_{d,t}^{i}$ in the LAC Region behaves very similarly to a
scenario with no policy changes.

In contrast, all sectors have dynamic productivity losses weakly greater than
$1\%$ when we simulate decoupling with LAC as part of the Eastern bloc. There
is large sectoral heterogeneity. Under full decoupling, productivity losses
range from $1\%$ in Electronic Equipment (\texttt{elm}) to $0.4\%$ in Business
Services (\texttt{tas}). These differences are induced by input-output linkages.

This experiment underscores that the costs of decoupling might be unbearably
high for low and middle-income countries that are excluded from the Western
bloc. Many countries in Latin America and Africa benefit from increasingly
large trade ties to China through both having larger market access and access
to lower input costs. However, as the dynamic costs of severing ties with the
West would be very high, and political leaders in those countries might have the incentive to keep an equidistant relationship between East and West, by
preserving both mid-term gains from the relationship with China and longer
term dynamic gains from having access to Western supply chains.

\subsection{Electronic Equipment Decoupling}

Finally, we compare our baseline scenario of \textit{full decouple} in
\textbf{all sectors} with a \textit{full decouple} \textbf{restricted to the
electronics equipment sector}. In both scenarios, we assume that the ideas
diffusion mechanism works as described by equation
(\ref{eq: final_law_of_motion}) and we set $\beta=0.44$, according to the
calibration described before.

\begin{figure}[th]
\centering
\includegraphics[scale=0.6]{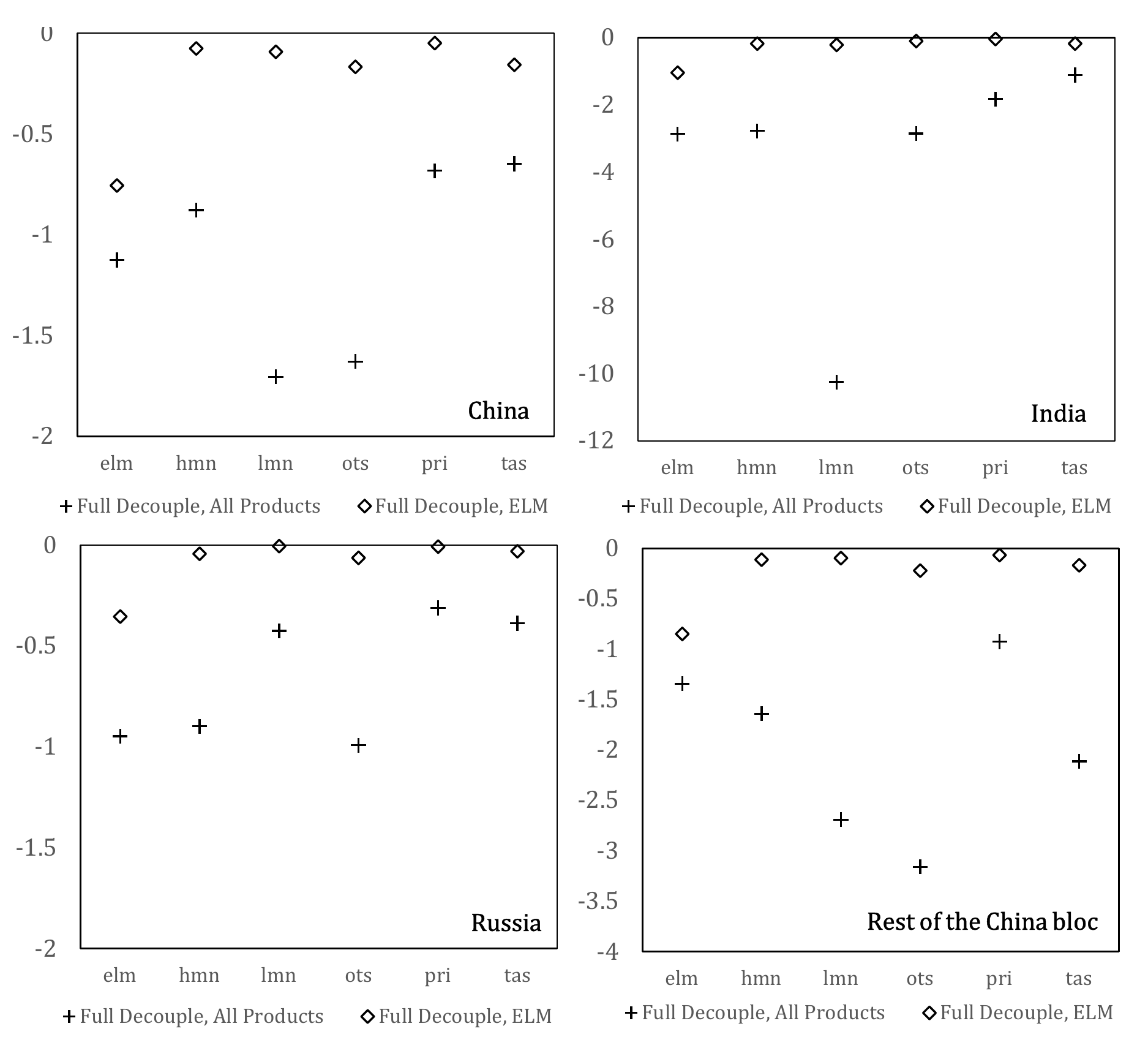}\caption{Cumulative Percentage
Change of the Fr\'echet Distribution scale parameter $\lambda_{d,t}^{i}$, by
scenario. {\protect\footnotesize {\ \textit{Full Decouple} increases iceberg
trade costs $\tau_{sd,t}^{i}$ by $160$ percentage points in either all sectors
or only in the Electronic Equipment (\texttt{elm}) sector. In both cases, we
set the parameter that controls the diffusion of ideas to $\beta=0.44$.
Country codes: \texttt{chn}, China; \texttt{ind}, India; \texttt{rus}, Russia;
\texttt{rwc}, Rest of Eastern bloc; \texttt{rwu}, Rest of Western bloc;
\texttt{lac}, Latin America; \texttt{e27}, European Union; \texttt{ode}, Other
Developed; \texttt{usa}, United States. Tables with the values for these
charts can be found in the Appendix.}}}%
\label{fig:elm_lambda}%
\end{figure}

Note that, due to the multi-sector structure of the model, an increase in
iceberg trade costs in one particular sector potentially has an indirect
effect in all sectors of the economy. The magnitude of such impact in a given
sector can be split between a direct effect (proportional to input use from
the \texttt{elm} sector as intermediates) and an indirect effect (proportional
to the use of the \texttt{elm} sector in the production of intermediates inputs).

Results in Figure \ref{fig:elm_lambda} show the productivity losses induced by
policy changes represented by the evolution of the Fr\'echet Distribution
scale parameter $\lambda_{d,t}^{i}$ for those regions in the Eastern bloc.
Contrasting the full decoupling in all sectors and one restricted to
electronic equipment shows that, across all regions, productivity losses are
substantially reduced and mostly restricted to the \texttt{elm} sector.

While there is some negative spillover effect to other sectors due to
input-output linkages, particularly to business services (\texttt{tas}), these
are very small for most regions. Regions such as Russia, which already had
limited exposure to Western intermediate sourcing in the main scenario, see
productivity losses go down to nearly zero across all sectors under the
scenario that limits decoupling to the \texttt{elm} sector. China's losses in
the \texttt{elm} sector are roughly similar to losses when decoupling happens
in all products, but other sectors are not substantially affected.

All other regions have non-negligible losses in the \texttt{elm} sector. The
largest changes happen for India and the Rest of the Eastern bloc. Those
regions have a lower productivity starting point and benefit proportionately
more from exposure to higher-quality intermediate inputs. For that reason,
full decoupling in all products leads to large differential losses in
productivity in those regions. The more restricted full decoupling in
\texttt{elm} scenario limits losses, since those are proportional to the use
of Western electronic equipment as inputs in the production of other sectors.

\begin{figure}[th]
\centering
\includegraphics[scale=0.7]{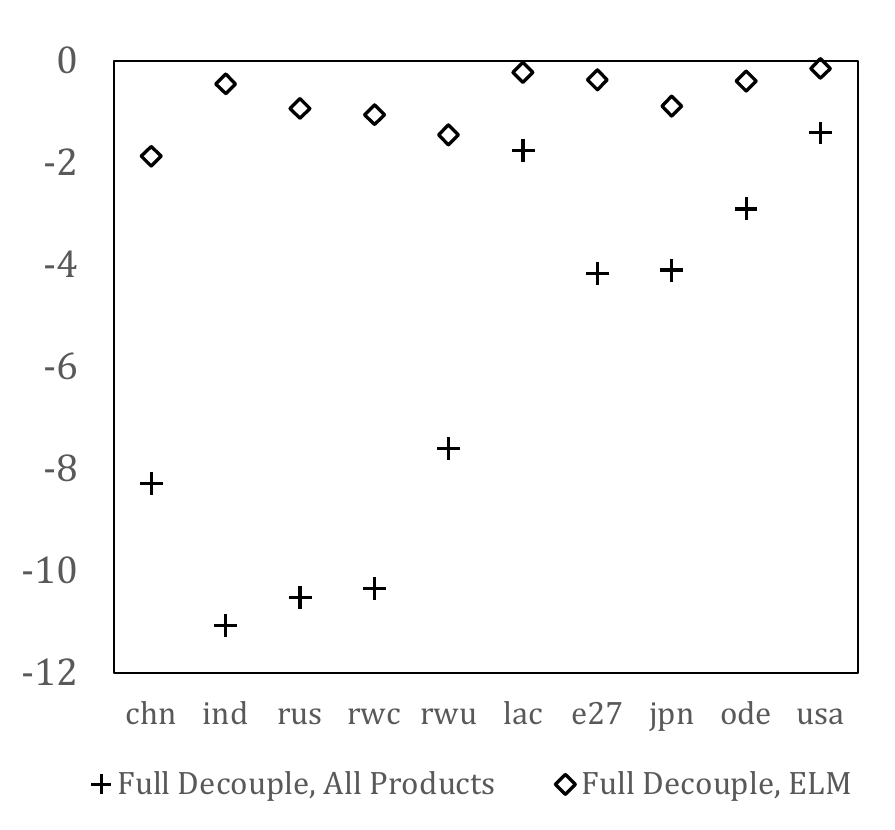}\caption{Cumulative Percentage
Change in Welfare (Real Income), by scenario.
{\protect\footnotesize {\ \textit{Full Decouple} increases iceberg trade costs
$\tau_{sd,t}^{i}$ by $160$ percentage points in either all sectors or only in
the Electronic Equipment (\texttt{elm}) sector. In both cases, we set the parameter that controls the diffusion of ideas to $\beta=0.44$.
Country codes: \texttt{chn}, China; \texttt{ind}, India; \texttt{rus}, Russia;
\texttt{rwc}, Rest of Eastern bloc; \texttt{rwu}, Rest of Western bloc;
\texttt{lac}, Latin America; \texttt{e27}, European Union; \texttt{ode}, Other
Developed; \texttt{usa}, United States. Tables with the values for these
charts can be found in the Appendix.}}}%
\label{fig:elm_gdp}%
\end{figure}

Changes in productivity map into changes in welfare, pictured in Figure
\ref{fig:elm_gdp}. While welfare losses are substantial, ranging from
$0.4-1.9\%$, they are very different in magnitude to the devastating results
of a full decoupling in all products, in which losses range between $8-12\%$.

These results underline two important observations. First, the costs of
sector-specific decoupling might be limited enough for this scenario to be
feasible. Second, input-output structures play an important role in magnifying
dynamic losses. Limiting decoupling to one specific sector tapers down
indirect magnification effects that happen through the input-output network.

\section{Conclusion\label{conclusion}}

We build a multi-sector multi-region general equilibrium model with dynamic
sector-specific knowledge diffusion in order to realistically investigate the
impact of large and persistent geopolitical conflicts on global trade
patterns, economic growth, and innovation. Canonical trade models typically
start from a fixed technology assumption, which thus misses a crucial source
of gains from trade through the diffusion of ideas.

In our theoretical contribution, we show that large trade costs can lead to
dynamic inefficiencies in sectoral knowledge diffusion. Furthermore, we show
that in a multi-sector framework deviations from optimal knowledge diffusion
happen both within and between sectors. Additionally, sectoral deviations
accumulate, such that total trade shares being close to their aggregate
optimal diffusion points is no longer sufficient to guarantee optimal
diffusion. A takeaway of our theoretical discussion is that, as the number of
sectors increases, so do the number of deviations from optimality and
diffusion losses tend to be higher with multiple sectors.

We then use this toolkit to simulate the trade, innovation, and welfare
effects of potential receding globalization characterized by economic decoupling between the East and West, yielding three main insights. First, the projected
welfare losses for the global economy of a decoupling scenario can be drastic,
being as large as 12\% in some regions, and are largest in the lower-income
regions as they would benefit less from technology spillovers from richer
areas. Second, the described size and pattern of welfare effects are specific
to the model with diffusion of ideas. Without diffusion of ideas the size and
variation across regions of the welfare losses would be substantially smaller.
Third, a multi-sector framework exacerbates diffusion inefficiencies induced
by trade costs relative to a single-sector one.

This has important implications for the role of the multilateral trading
system. First, the current system with global trade rules guaranteeing open
and free trade between all major players is especially important for the
lowest-income regions. Second, if geopolitical considerations would lead to a
split of the big players into two blocs, it would be important that an
institutional framework remains in place for smaller countries to keep open
trade relations with both blocs, in particular for the lowest income regions.

The toolkit we have built is versatile and can be employed for many other
research questions, in particular, focused on the analysis of policy. Future
research could be extended in various directions. We mention two. First, there
is ample empirical evidence for spillover effects from FDI, so the model could
be extended with FDI and sales by multinational affiliates. Second, the
framework with both technology spillovers and profits can be employed to
analyze the economic effects of subsidy policies in different regions, an
important policy topic in the multilateral trading system.

\printbibliography

@article{wyne_how_2020,
	title = {How to {Think} about {Potentially} {Decoupling} from {China}},
	volume = {43},
	issn = {0163-660X},
	url = {https://doi.org/10.1080/0163660X.2020.1735854},
	doi = {10.1080/0163660X.2020.1735854},
	number = {1},
	urldate = {2022-01-17},
	journal = {The Washington Quarterly},
	author = {Wyne, Ali},
	month = jan,
	year = {2020},
	pages = {41--64},
}

@article{levchenko_evolution_2016,
	title = {The evolution of comparative advantage: {Measurement} and welfare implications},
	volume = {78},
	issn = {0304-3932},
	shorttitle = {The evolution of comparative advantage},
	url = {https://www.sciencedirect.com/science/article/pii/S0304393216000088},
	doi = {10.1016/j.jmoneco.2016.01.005},
	language = {en},
	urldate = {2023-01-23},
	journal = {Journal of Monetary Economics},
	author = {Levchenko, Andrei A. and Zhang, Jing},
	month = apr,
	year = {2016},
	pages = {96--111},
}

@article{deng_specialization_2016,
	title = {Specialization {Dynamics}, {Convergence}, and {Idea} {Flows}},
	url = {https://ideas.repec.org//p/bai/series/series_wp_09-2016.html},
	language = {en},
	urldate = {2023-01-23},
	journal = {SERIES},
	author = {Deng, Liuchun},
	month = nov,
	year = {2016},
	note = {Number: 09-2016
Publisher: Dipartimento di Economia e Finanza - Università degli Studi di Bari "Aldo Moro"},
}

@article{dix-carneiro_trade_2017,
	title = {Trade {Liberalization} and {Regional} {Dynamics}},
	volume = {107},
	issn = {0002-8282},
	url = {https://www.jstor.org/stable/44871772},
	number = {10},
	urldate = {2022-04-18},
	journal = {The American Economic Review},
	author = {Dix-Carneiro, Rafael and Kovak, Brian K.},
	year = {2017},
	note = {Publisher: American Economic Association},
	pages = {2908--2946},
}

@misc{abrams_new_2022,
	title = {The {New} {Cold} {War}},
	url = {https://www.cfr.org/blog/new-cold-war-0},
	language = {en},
	urldate = {2022-03-15},
	journal = {Council on Foreign Relations},
	author = {Abrams, Elliot},
	month = mar,
	year = {2022},
}

@article{atalay_how_2017,
	title = {How {Important} {Are} {Sectoral} {Shocks}?},
	volume = {9},
	issn = {1945-7707},
	url = {https://www.jstor.org/stable/26528417},
	number = {4},
	urldate = {2022-03-23},
	journal = {American Economic Journal: Macroeconomics},
	author = {Atalay, Enghin},
	year = {2017},
	note = {Publisher: American Economic Association},
	pages = {254--280},
}

@misc{weber_russias_2022,
	title = {Russia’s {War} on {Ukraine}: the {EU}’s {Geopolitical} {Awakening}},
	shorttitle = {Russia’s {War} on {Ukraine}},
	url = {https://www.gmfus.org/news/russias-war-ukraine-eus-geopolitical-awakening},
	language = {en},
	urldate = {2022-03-17},
	journal = {GMFUS},
	author = {Weber, Gesine and Scheffer, Alexandra},
	month = mar,
	year = {2022},
}

@article{mearsheimer_why_2014,
	title = {Why the {Ukraine} {Crisis} {Is} the {West}'s {Fault}: {The} {Liberal} {Delusions} {That} {Provoked} {Putin}},
	volume = {93},
	issn = {0015-7120},
	shorttitle = {Why the {Ukraine} {Crisis} {Is} the {West}'s {Fault}},
	url = {https://www.jstor.org/stable/24483306},
	number = {5},
	urldate = {2022-03-15},
	journal = {Foreign Affairs},
	author = {Mearsheimer, John J.},
	year = {2014},
	note = {Publisher: Council on Foreign Relations},
	pages = {77--89},
}

@article{lukin_russiachina_2021,
	title = {The {Russia}–{China} entente and its future},
	volume = {58},
	issn = {1740-3898},
	url = {https://doi.org/10.1057/s41311-020-00251-7},
	doi = {10.1057/s41311-020-00251-7},
	language = {en},
	number = {3},
	urldate = {2022-03-15},
	journal = {International Politics},
	author = {Lukin, Artyom},
	month = jun,
	year = {2021},
	pages = {363--380},
}

@article{rivera-batiz_international_1991,
	title = {International trade with endogenous technological change},
	volume = {35},
	issn = {0014-2921},
	url = {https://www.sciencedirect.com/science/article/pii/001429219190048N},
	doi = {10.1016/0014-2921(91)90048-N},
	language = {en},
	number = {4},
	urldate = {2022-02-21},
	journal = {European Economic Review},
	author = {Rivera-Batiz, Luis A. and Romer, Paul M.},
	month = may,
	year = {1991},
	pages = {971--1001},
}

@article{hertel_how_2007,
	title = {How confident can we be of {CGE}-based assessments of {Free} {Trade} {Agreements}?},
	volume = {24},
	issn = {0264-9993},
	url = {https://www.sciencedirect.com/science/article/pii/S026499930700003X},
	doi = {10.1016/j.econmod.2006.12.002},
	language = {en},
	number = {4},
	urldate = {2022-01-17},
	journal = {Economic Modelling},
	author = {Hertel, Thomas and Hummels, David and Ivanic, Maros and Keeney, Roman},
	month = jul,
	year = {2007},
	pages = {611--635},
}

@article{dekle_unbalanced_2007,
	title = {Unbalanced {Trade}},
	volume = {97},
	issn = {0002-8282},
	url = {https://www.aeaweb.org/articles?id=10.1257/aer.97.2.351},
	doi = {10.1257/aer.97.2.351},
	language = {en},
	number = {2},
	urldate = {2022-01-17},
	journal = {American Economic Review},
	author = {Dekle, Robert and Eaton, Jonathan and Kortum, Samuel},
	month = may,
	year = {2007},
	pages = {351--355},
}

@article{navarro_peter_white_2018,
	title = {White {House} {National} {Trade} {Council} {Director} {Peter} {Navarro} on {Chinese} {Economic} {Aggression} {June} 28th, 2018},
	language = {en},
	author = {{Navarro, Peter}},
	year = {2018},
	pages = {5},
}

@book{scissors_partial_2020,
	title = {Partial decoupling from {China}: {A} brief guide},
	shorttitle = {Partial decoupling from {China}},
	publisher = {JSTOR},
	author = {Scissors, Derek},
	year = {2020},
}

@article{tang_decoupling_2021,
	title = {Decoupling in science and education: {A} collateral damage beyond deteriorating {US}–{China} relations},
	volume = {48},
	issn = {0302-3427, 1471-5430},
	shorttitle = {Decoupling in science and education},
	url = {https://academic.oup.com/spp/article/48/5/630/6322422},
	doi = {10.1093/scipol/scab035},
	language = {en},
	number = {5},
	urldate = {2022-01-17},
	journal = {Science and Public Policy},
	author = {Tang, Li and Cao, Cong and Wang, Zheng and Zhou, Zhuo},
	month = oct,
	year = {2021},
	pages = {630--634},
}

@article{walter_backlash_2021,
	title = {The {Backlash} {Against} {Globalization}},
	volume = {24},
	issn = {1094-2939, 1545-1577},
	url = {https://www.annualreviews.org/doi/10.1146/annurev-polisci-041719-102405},
	doi = {10.1146/annurev-polisci-041719-102405},
	language = {en},
	number = {1},
	urldate = {2022-01-16},
	journal = {Annual Review of Political Science},
	author = {Walter, Stefanie},
	month = may,
	year = {2021},
	pages = {421--442},
}

@article{zhao_is_2019,
	title = {Is a {New} {Cold} {War} {Inevitable}? {Chinese} {Perspectives} on {US}–{China} {Strategic} {Competition}},
	volume = {12},
	issn = {1750-8916, 1750-8924},
	shorttitle = {Is a {New} {Cold} {War} {Inevitable}?},
	url = {https://academic.oup.com/cjip/article/12/3/371/5544745},
	doi = {10.1093/cjip/poz010},
	language = {en},
	number = {3},
	urldate = {2022-01-16},
	journal = {The Chinese Journal of International Politics},
	author = {Zhao, Minghao},
	month = sep,
	year = {2019},
	pages = {371--394},
}

@article{melitz_impact_2003,
	title = {The {Impact} of {Trade} on {Intra}-{Industry} {Reallocations} and {Aggregate} {Industry} {Productivity}},
	volume = {71},
	issn = {0012-9682, 1468-0262},
	url = {http://doi.wiley.com/10.1111/1468-0262.00467},
	doi = {10.1111/1468-0262.00467},
	language = {en},
	number = {6},
	urldate = {2021-09-08},
	journal = {Econometrica},
	author = {Melitz, Marc J.},
	month = nov,
	year = {2003},
	pages = {1695--1725},
}

@article{armington_theory_1969,
	title = {A {Theory} of {Demand} for {Products} {Distinguished} by {Place} of {Production} ({Une} theorie de la demande de produits differencies d'apres leur origine) ({Una} teoria de la demanda de productos distinguiendolos segun el lugar de produccion)},
	volume = {16},
	issn = {00208027},
	url = {http://www.palgrave-journals.com/doifinder/10.2307/3866403},
	doi = {10.2307/3866403},
	number = {1},
	urldate = {2021-09-08},
	journal = {Staff Papers - International Monetary Fund},
	author = {Armington, Paul S.},
	month = mar,
	year = {1969},
	pages = {159},
}

@article{bernard_plants_2003,
	title = {Plants and {Productivity} in {International} {Trade}},
	volume = {93},
	issn = {0002-8282},
	url = {https://pubs.aeaweb.org/doi/10.1257/000282803769206296},
	doi = {10.1257/000282803769206296},
	language = {en},
	number = {4},
	urldate = {2021-09-05},
	journal = {American Economic Review},
	author = {Bernard, Andrew B and Eaton, Jonathan and Jensen, J. Bradford and Kortum, Samuel},
	month = aug,
	year = {2003},
	pages = {1268--1290},
}

@article{buera_idea_2018,
	title = {Idea {Flows} and {Economic} {Growth}},
	volume = {10},
	url = {https://doi.org/10.1146/annurev-economics-063016-104414},
	doi = {10.1146/annurev-economics-063016-104414},
	number = {1},
	urldate = {2021-04-29},
	journal = {Annual Review of Economics},
	author = {Buera, Francisco J. and Lucas, Robert E.},
	year = {2018},
	note = {\_eprint: https://doi.org/10.1146/annurev-economics-063016-104414},
	pages = {315--345},
}

@techreport{kaska_huawei_2019,
	address = {Tallinn},
	title = {Huawei, {5G} and {China} as a {Security} {Threat}},
	url = {https://www.ccdcoe.org/uploads/2019/03/CCDCOE-Huawei-2019-03-28-FINAL.pdf},
	language = {en},
	institution = {NATO Cooperative Cyber Defence Centre of Excellence (CCDCOE)},
	author = {Kaska, Kadri and Beckvard, Henrik and Minárik, Tomáš},
	year = {2019},
}

@article{buera_global_2020,
	title = {The {Global} {Diffusion} of {Ideas}},
	volume = {88},
	issn = {1468-0262},
	url = {https://onlinelibrary.wiley.com/doi/abs/10.3982/ECTA14044},
	doi = {https://doi.org/10.3982/ECTA14044},
	language = {en},
	number = {1},
	urldate = {2021-04-08},
	journal = {Econometrica},
	author = {Buera, Francisco J. and Oberfield, Ezra},
	year = {2020},
	pages = {83--114},
}

@article{nye_jr_power_2020,
	title = {Power and {Interdependence} with {China}},
	volume = {43},
	issn = {0163-660X},
	url = {https://doi.org/10.1080/0163660X.2020.1734303},
	doi = {10.1080/0163660X.2020.1734303},
	number = {1},
	urldate = {2021-04-08},
	journal = {The Washington Quarterly},
	author = {Nye Jr., Joseph},
	month = jan,
	year = {2020},
	pages = {7--21},
}

@article{colantone_trade_2018,
	title = {The {Trade} {Origins} of {Economic} {Nationalism}: {Import} {Competition} and {Voting} {Behavior} in {Western} {Europe}},
	volume = {62},
	issn = {1540-5907},
	shorttitle = {The {Trade} {Origins} of {Economic} {Nationalism}},
	url = {https://onlinelibrary.wiley.com/doi/abs/10.1111/ajps.12358},
	doi = {https://doi.org/10.1111/ajps.12358},
	language = {en},
	number = {4},
	urldate = {2021-04-08},
	journal = {American Journal of Political Science},
	author = {Colantone, Italo and Stanig, Piero},
	year = {2018},
	pages = {936--953},
}

@article{eaton_technology_2002,
	title = {Technology, {Geography}, and {Trade}},
	volume = {70},
	issn = {1468-0262},
	url = {https://onlinelibrary.wiley.com/doi/abs/10.1111/1468-0262.00352},
	doi = {https://doi.org/10.1111/1468-0262.00352},
	language = {en},
	number = {5},
	urldate = {2021-01-19},
	journal = {Econometrica},
	author = {Eaton, Jonathan and Kortum, Samuel},
	year = {2002},
	pages = {1741--1779},
}

@article{hage_choice_2011,
	title = {Choice or {Circumstance}? {Adjusting} {Measures} of {Foreign} {Policy} {Similarity} for {Chance} {Agreement}},
	volume = {19},
	issn = {1047-1987, 1476-4989},
	shorttitle = {Choice or {Circumstance}?},
	url = {https://www.cambridge.org/core/journals/political-analysis/article/abs/choice-or-circumstance-adjusting-measures-of-foreign-policy-similarity-for-chance-agreement/A2AA44B24EE65FC7401B7C44847B074D},
	doi = {10.1093/pan/mpr023},
	language = {en},
	number = {3},
	urldate = {2021-04-09},
	journal = {Political Analysis},
	author = {Häge, Frank M.},
	year = {2011},
	note = {Publisher: Cambridge University Press},
	pages = {287--305},
}

@book{keohane_power_2011,
	address = {Boston},
	edition = {4th edition},
	title = {Power \& {Interdependence}},
	isbn = {978-0-205-08291-9},
	language = {English},
	publisher = {Pearson},
	author = {Keohane, Robert and Nye Jr, Joseph},
	month = feb,
	year = {2011},
}

@book{waltz_theory_2010,
	address = {Long Grove, Ill},
	edition = {1st edition},
	title = {Theory of {International} {Politics}},
	isbn = {978-1-57766-670-7},
	language = {English},
	publisher = {Waveland Press},
	author = {Waltz, Kenneth N.},
	month = feb,
	year = {2010},
}

@article{wei_towards_2019,
	title = {Towards {Economic} {Decoupling}? {Mapping} {Chinese} {Discourse} on the {China}–{US} {Trade} {War}},
	volume = {12},
	issn = {1750-8916},
	shorttitle = {Towards {Economic} {Decoupling}?},
	url = {https://doi.org/10.1093/cjip/poz017},
	doi = {10.1093/cjip/poz017},
	number = {4},
	urldate = {2021-04-08},
	journal = {The Chinese Journal of International Politics},
	author = {Wei, Li},
	month = dec,
	year = {2019},
	pages = {519--556},
}

@article{bown_trumps_2019,
	title = {Trump's {Assault} on the {Global} {Trading} {System}: {And} {Why} {Decoupling} from {China} {Will} {Change} {Everything} {Essays}},
	volume = {98},
	shorttitle = {Trump's {Assault} on the {Global} {Trading} {System}},
	url = {https://heinonline.org/HOL/P?h=hein.journals/fora98&i=975},
	language = {eng},
	number = {5},
	urldate = {2021-04-08},
	journal = {Foreign Affairs},
	author = {Bown, Chad P. and Irwin, Douglas A.},
	year = {2019},
	pages = {125--137},
}

@article{nicita_cooperation_2018,
	title = {Cooperation in {WTO}’s {Tariff} {Waters}?},
	volume = {126},
	issn = {0022-3808, 1537-534X},
	url = {https://www.journals.uchicago.edu/doi/10.1086/697085},
	doi = {10.1086/697085},
	language = {en},
	number = {3},
	urldate = {2021-04-08},
	journal = {Journal of Political Economy},
	author = {Nicita, Alessandro and Olarreaga, Marcelo and Silva, Peri},
	month = jun,
	year = {2018},
	pages = {1302--1338},
}

@article{krugman_scale_1980,
	title = {Scale {Economies}, {Product} {Differentiation}, and the {Pattern} of {Trade}},
	volume = {70},
	issn = {0002-8282},
	url = {https://www.jstor.org/stable/1805774},
	number = {5},
	urldate = {2021-04-08},
	journal = {The American Economic Review},
	author = {Krugman, Paul},
	year = {1980},
	note = {Publisher: American Economic Association},
	pages = {950--959},
}

@techreport{alvarez_idea_2013,
	title = {Idea {Flows}, {Economic} {Growth}, and {Trade}},
	url = {https://www.nber.org/papers/w19667},
	language = {en},
	number = {w19667},
	urldate = {2021-04-08},
	institution = {National Bureau of Economic Research},
	author = {Alvarez, Fernando E. and Buera, Francisco J. and Lucas, Robert E. and Jr},
	month = nov,
	year = {2013},
	doi = {10.3386/w19667},
}

@article{autor_china_2013,
	title = {The {China} {Syndrome}: {Local} {Labor} {Market} {Effects} of {Import} {Competition} in the {United} {States}},
	volume = {103},
	issn = {0002-8282},
	shorttitle = {The {China} {Syndrome}},
	url = {https://www.aeaweb.org/articles?id=10.1257/aer.103.6.2121},
	doi = {10.1257/aer.103.6.2121},
	language = {en},
	number = {6},
	urldate = {2021-04-08},
	journal = {American Economic Review},
	author = {Autor, David H. and Dorn, David and Hanson, Gordon H.},
	month = oct,
	year = {2013},
	pages = {2121--2168},
}

@article{grossman_product_1989,
	title = {Product {Development} and {International} {Trade}},
	volume = {97},
	issn = {0022-3808},
	url = {https://www.jstor.org/stable/1833238},
	number = {6},
	urldate = {2021-03-09},
	journal = {Journal of Political Economy},
	author = {Grossman, Gene M. and Helpman, Elhanan},
	year = {1989},
	note = {Publisher: University of Chicago Press},
	pages = {1261--1283},
}

@article{arkolakis_new_2012,
	title = {New {Trade} {Models}, {Same} {Old} {Gains}?},
	volume = {102},
	issn = {0002-8282},
	url = {https://www.aeaweb.org/articles?id=10.1257/aer.102.1.94},
	doi = {10.1257/aer.102.1.94},
	language = {en},
	number = {1},
	urldate = {2021-01-23},
	journal = {American Economic Review},
	author = {Arkolakis, Costas and Costinot, Arnaud and Rodríguez-Clare, Andrés},
	month = feb,
	year = {2012},
	pages = {94--130},
}

@article{kortum_research_1997,
	title = {Research, {Patenting}, and {Technological} {Change}},
	volume = {65},
	issn = {0012-9682},
	url = {https://www.jstor.org/stable/2171741},
	doi = {10.2307/2171741},
	number = {6},
	urldate = {2021-01-22},
	journal = {Econometrica},
	author = {Kortum, Samuel S.},
	year = {1997},
	note = {Publisher: [Wiley, Econometric Society]},
	pages = {1389--1419},
}

@techreport{santacreu_knowledge_2017,
	title = {Knowledge {Diffusion}, {Trade} and {Innovation} across {Countries} and {Sectors}},
	url = {https://research.stlouisfed.org/wp/more/2017-029},
	language = {en},
	urldate = {2021-01-22},
	author = {Santacreu, Ana Maria and Li, Nan and Cai, Jie (April)},
	year = {2017},
	doi = {10.20955/wp.2017.029},
}

@article{jovanovic_growth_1989,
	title = {The {Growth} and {Diffusion} of {Knowledge}},
	volume = {56},
	issn = {0034-6527},
	url = {https://www.jstor.org/stable/2297501},
	doi = {10.2307/2297501},
	number = {4},
	urldate = {2021-01-22},
	journal = {The Review of Economic Studies},
	author = {Jovanovic, Boyan and Rob, Rafael},
	year = {1989},
	note = {Publisher: [Oxford University Press, Review of Economic Studies, Ltd.]},
	pages = {569--582},
}

@article{ossa_why_2015,
	title = {Why trade matters after all},
	volume = {97},
	issn = {0022-1996},
	url = {http://www.sciencedirect.com/science/article/pii/S0022199615001178},
	doi = {10.1016/j.jinteco.2015.07.002},
	language = {en},
	number = {2},
	urldate = {2021-01-21},
	journal = {Journal of International Economics},
	author = {Ossa, Ralph},
	month = nov,
	year = {2015},
	pages = {266--277},
}

@article{romer_endogenous_1990,
	title = {Endogenous {Technological} {Change}},
	volume = {98},
	url = {http://www.jstor.org/stable/2937632},
	language = {en},
	number = {5,},
	journal = {The Journal of Political Economy},
	author = {Romer, Paul M.},
	year = {1990},
	pages = {S71-- S102},
}

@article{eaton_international_1999,
	title = {International {Technology} {Diffusion}: {Theory} and {Measurement}},
	volume = {40},
	issn = {0020-6598},
	shorttitle = {International {Technology} {Diffusion}},
	url = {https://www.jstor.org/stable/2648766},
	number = {3},
	urldate = {2021-01-19},
	journal = {International Economic Review},
	author = {Eaton, Jonathan and Kortum, Samuel},
	year = {1999},
	note = {Publisher: [Economics Department of the University of Pennsylvania, Wiley, Institute of Social and Economic Research, Osaka University]},
	pages = {537--570},
}

\newpage

ONLINE APPENDIX

\appendix 

\setcounter{equation}{0}
\renewcommand{\theequation}{A-\arabic{equation}}

\newpage

\section{Additional tables calibration exercise}
\setcounter{table}{0}
\renewcommand{\thetable}{A\arabic{table}}

\begin{table}[htp!]
\caption{Growth Rate of Real GDP and Real GDP per Capita, respectively, Historically and in Simulations, using different values of $\beta$}%
\centering
\begin{tabular}
[c]{lcccc}\hline
$\beta$ & Mean & St.Dev. & $\max$ & $\min$\\\hline
GDP &  &  &  & \\
Historical & 3.60 & 2.66 & 8.90 & 0.67\\\hline
0.40 & 3.09 & 2.02 & 6.26 & 0.41\\
0.41 & 3.21 & 2.13 & 6.58 & 0.43\\
0.42 & 3.34 & 2.26 & 6.94 & 0.44\\
0.43 & 3.48 & 2.40 & 7.34 & 0.46\\
0.44 & 3.64 & 2.56 & 7.78 & 0.48\\
0.45 & 3.82 & 2.73 & 8.26 & 0.50\\
0.46 & 4.02 & 2.92 & 8.79 & 0.53\\
0.47 & 4.24 & 3.13 & 9.36 & 0.56\\
0.48 & 4.48 & 3.36 & 9.97 & 0.59\\
0.49 & 4.75 & 3.60 & 10.62 & 0.63\\
0.50 & 5.04 & 3.87 & 11.32 & 0.68\\\hline
& & & & \\\hline
GDP per capita &  &  &  & \\
Historical & 2.70 & 2.51 & 8.36 & 0.75\\\hline
0.40 & 2.20 & 1.65 & 4.91 & 0.47\\
0.41 & 2.32 & 1.76 & 5.23 & 0.48\\
0.42 & 2.44 & 1.89 & 5.59 & 0.50\\
0.43 & 2.59 & 2.03 & 5.98 & 0.51\\
0.44 & 2.75 & 2.19 & 6.41 & 0.53\\
0.45 & 2.92 & 2.36 & 6.89 & 0.54\\
0.46 & 3.12 & 2.55 & 7.41 & 0.57\\
0.47 & 3.34 & 2.76 & 7.97 & 0.59\\
0.48 & 3.58 & 2.98 & 8.57 & 0.62\\
0.49 & 3.84 & 3.22 & 9.22 & 0.65\\
0.50 & 4.13 & 3.48 & 9.91 & 0.69\\\hline
\label{tab: beta_version}
\end{tabular}
\end{table}

\newpage

\begin{table}[htp!]
\caption{The squared difference between the sum of the historical and
simulated mean and standard deviation of GDP, GDP per capita and their sum}%
\label{beta_mm}%
\centering
\begin{tabular}
[c]{rrrr}\hline%
\multicolumn{1}{l}{$\beta$} & \multicolumn{1}{l}{GDP} & \multicolumn{1}{l}{GDP pc} & \multicolumn{1}{l}{Sum}\\\hline
0.40 & 0.67 & 1.00 & 1.67\\
0.41 & 0.43 & 0.72 & 1.15\\
0.42 & 0.23 & 0.46 & 0.69\\
0.43 & 0.08 & 0.25 & 0.33\\
0.44 & 0.01 & 0.11 & 0.12\\
0.45 & 0.06 & 0.07 & 0.13\\
0.46 & 0.25 & 0.17 & 0.42\\
0.47 & 0.64 & 0.46 & 1.10\\
0.48 & 1.28 & 0.98 & 2.26\\
0.49 & 2.22 & 1.79 & 4.01\\
0.50 & 3.54 & 2.96 & 6.50\\\hline
\end{tabular}
\end{table}

\newpage

\begin{table}[htp!]
\caption{Growth Rate of Real GDP and Real GDP per Capita, respectively, Historically and in Simulations, using different values of $\beta$ between $0$ and $0.6$}%
\label{beta_0_06}
\centering
\begin{tabular}
[c]{lcccc}\hline
$\beta$ & Mean & St.Dev. & $\max$ & $\min$\\\hline
GDP &  &  &  & \\
Historical & 3.60 & 2.66 & 8.90 & 0.67\\\hline
0 & 2.17 & 1.10 & 3.79 & 0.32\\
0.5 & 2.19 & 1.12 & 3.84 & 0.32\\
0.10 & 2.20 & 1.13 & 3.87 & 0.32\\
0.15 & 2.23 & 1.16 & 3.93 & 0.32\\
0.20 & 2.27 & 1.20 & 4.02 & 0.33\\
0.25 & 2.34 & 1.26 & 4.19 & 0.33\\
0.30 & 2.46 & 1.39 & 4.49 & 0.35\\
0.35 & 2.68 & 1.61 & 5.08 & 0.37\\
0.40 & 3.09 & 2.02 & 6.26 & 0.41\\
0.45 & 3.82 & 2.73 & 8.26 & 0.50\\
0.50 & 5.04 & 3.87 & 11.32 & 0.68\\
0.55 & 6.89 & 5.42 & 15.39 & 1.01\\
0.60 & 9.50 & 7.32 & 20.53 & 1.66\\\hline
& & & & \\\hline
GDP per capita &  &  &  & \\
Historical & 2.70 & 2.51 & 8.36 & 0.75\\\hline
0 & 1.29 & 0.72 & 2.27 & 0.40\\
0.5 & 1.31 & 0.75 & 2.33 & 0.40\\
0.10 & 1.32 & 0.76 & 2.36 & 0.40\\
0.15 & 1.35 & 0.78 & 2.41 & 0.40\\
0.20 & 1.39 & 0.82 & 2.53 & 0.40\\
0.25 & 1.46 & 0.89 & 2.74 & 0.41\\
0.30 & 1.58 & 1.01 & 3.10 & 0.42\\
0.35 & 1.80 & 1.24 & 3.75 & 0.44\\
0.40 & 2.20 & 1.65 & 4.91 & 0.47\\
0.45 & 2.92 & 2.36 & 6.89 & 0.54\\
0.50 & 4.13 & 3.48 & 9.91 & 0.69\\
0.55 & 5.96 & 5.02 & 13.92 & 0.98\\
0.60 & 8.54 & 6.88 & 18.97 & 1.55\\\hline
\end{tabular}
\end{table}

\newpage

\begin{table}[htp!]
\caption{Growth Rate of Real GDP and Real GDP per Capita, respectively, Historically and in Simulations, using different values of $\beta$ between $0$ and $0.6$ with an Autonomous Technology Growth Rate of $\alpha=2.36$}%
\label{beta_0_06_alpha}
\centering
\begin{tabular}
[c]{lcccc}\hline
$\beta$ & Mean & St.Dev. & $\max$ & $\min$\\\hline
GDP &  &  &  & \\
Historical & 3.60 & 2.66 & 8.90 & 0.67\\\hline
0 & 2.17 & 1.10 & 3.79 & 0.32\\
0.5 & 2.19 & 1.12 & 3.84 & 0.32\\
0.10 & 2.21 & 1.14 & 3.88 & 0.32\\
0.15 & 2.23 & 1.16 & 3.94 & 0.32\\
0.20 & 2.28 & 1.21 & 4.04 & 0.33\\
0.25 & 2.35 & 1.28 & 4.23 & 0.33\\
0.30 & 2.49 & 1.41 & 4.56 & 0.35\\
0.35 & 2.73 & 1.65 & 5.22 & 0.37\\
0.40 & 3.17 & 2.09 & 6.48 & 0.42\\
0.45 & 3.95 & 2.85 & 8.60 & 0.52\\
0.50 & 5.23 & 4.03 & 11.76 & 0.71\\
0.55 & 7.16 & 5.62 & 15.92 & 1.08\\
0.60 & 9.84 & 7.52 & 21.17 & 1.77\\\hline
& & & & \\\hline
GDP per capita &  &  &  & \\
Historical & 2.70 & 2.51 & 8.36 & 0.75\\\hline
0 & 1.29 & 0.72 & 2.27 & 0.40\\
0.5 & 1.31 & 0.75 & 2.33 & 0.40\\
0.10 & 1.33 & 0.77 & 2.36 & 0.40\\
0.15 & 1.35 & 0.79 & 2.43 & 0.40\\
0.20 & 1.40 & 0.83 & 2.56 & 0.40\\
0.25 & 1.47 & 0.91 & 2.79 & 0.41\\
0.30 & 1.61 & 1.04 & 3.18 & 0.42\\
0.35 & 1.85 & 1.28 & 3.89 & 0.44\\
0.40 & 2.28 & 1.72 & 5.13 & 0.48\\
0.45 & 3.05 & 2.48 & 7.22 & 0.56\\
0.50 & 4.32 & 3.64 & 10.35 & 0.72\\
0.55 & 6.22 & 5.21 & 14.45 & 1.04\\
0.60 & 8.87 & 7.08 & 19.56 & 1.66\\\hline
\end{tabular}
\end{table}

\newpage

\section{Construction of Initial Sectoral Productivity Parameters $\lambda_{s,0}^i$}
\label{appendix: lambda}

In our model, the location parameter of Fr\'echet distribution of a given industry-country $\lambda_{d,t}^i$ evolves endogenously according to a law of motion, as described by equation (\ref{eq: final_law_of_motion}). To calibrate the model, we need initial values $(\lambda_{d,0}^i)_{d \in \mathcal{D}, i \in \mathcal{I}}$. We proxy for the initial values using labor productivity in different sectors for each aggregate region and industry in our sample in the base-year 2014.

We do so by combining two different databases: the World Input Output Database's Social Economic Accounts (WIOD-SEA \textemdash \url{http://wiod.org/database/seas16}) and the World Bank's Global Productivity Database (WB-GPD \textemdash \url{https://www.worldbank.org/en/research/publication/global-productivity}). WIOD-SEA a reports value added in local currency and employed population for 42 countries and 56 industries. WB-GPD reports value added in local currency and employed population for 103 countries and 9 industries. For countries whose data are available in both databases, we use the data from WIOD-SEA, which is more granular. 

The first step is to create a cross-walk between WIOD-SEA industries and the more aggregate sectors in our model, namely: \texttt{elm}, Electronic Equipment; \texttt{hmn}, Heavy manufacturing; \texttt{lmn}, Light manufacturing; \texttt{ots}, Other Services; \texttt{pri}, Primary Sector; \texttt{tas}, Business services. We then convert value added in local currency to PPP-USD and market-rate USD using a panel of PPP and market exchange rates from the World Bank's World Development Indicators.

Afterwards, we did a similar cross-walk for the country-sector pairs in the WB-GPD database. The detailed cross-walk can be found at the end of this Appendix. However, WB-GPD only reports one aggregate manufacturing sector, while our model disaggregates manufacturing into three subsectors. In order to make them compatible, we take the following steps: (a) we classify countries as advanced and emerging markets in both the WIOD-SEA and the WB-GPD databases; (b) we calculate the average share of value added and employed workers in total manufacturing for each of the manufacturing subsectors (\texttt{elm}, \texttt{hmn}, and \texttt{lmn}) for emerging markets and advanced economies, respectively, in the WIOD-SEA database; and (c) we use those shares and reported value added and employed workers from the WB-GPD database in order to input, for each country, a disaggregation of total manufacturing into \texttt{elm}, \texttt{hmn}, and \texttt{lmn}. We then convert value added in local currency to PPP-USD and market-rate USD using a panel of PPP and market exchange rates from the World Bank's World Development Indicators.

Finally, we collapse PPP-USD value added, market-rate USD value added, and number of workers for the regions of our model (\texttt{chn}, China; \texttt{ind}, India; \texttt{rus}, Russia; \texttt{rwc}, Rest of Eastern bloc; \texttt{rwu}, Rest of Western bloc;  \texttt{lac}, Latin America; \texttt{e27}, European Union; \texttt{ode}, Other Developed; \texttt{usa}, United States), and calculate, for each region-industry pair, labor productivity as:

\begin{equation*}
    \lambda_{d,0}^i = \frac{{PPP\$VA}_{d,0}^i}{L_{d,0}^i}
\end{equation*}

using PPP-USD value added per worker.

\newpage

\begin{table}
\caption{Cross Walk Between WIOD-SEA and Model}
\vspace{3pt}
\begin{center}
\begin{tabular}{l|l}
\hline
\textbf{WIOD-SEA Sector}                   & \textbf{Model Sector}     \\
\hline
A01-03, B                                                              & pri          \\
C10-19, C31-32                                                         & lmn          \\
C20-25,    C28-30                                                      & hmn          \\
C26-27                                                                 & elm          \\
C33, D35, E36, F, G45-47, H50-53, I, L68,   N, O84, P85, Q, R, S, T, U & ots          \\
J58, J61, K64-66, M71-73                                               & tas         \\
\hline
\end{tabular}
\end{center}
\end{table}

\begin{table}
\caption{Cross Walk Between WB-GPD and Model}
\vspace{3pt}
\begin{center}
\begin{tabular}{l|l}
\hline
\textbf{WB-GPD Sector}                   & \textbf{Model Sector}     \\
\hline
1.Agriculture                   & pri               \\
2.Mining                        & pri               \\
3.Manufacturing                 & (see methodology) \\
4.Utilities                     & ots               \\
5.Construction                  & ots               \\
6.Trade services                & tas               \\
7.Transport services            & ots               \\
8.Finance amd business services & tas               \\
9.Other services                & ots              \\
\hline
\end{tabular}
\end{center}
\end{table}

\begin{itemize}
    \item \textbf{Countries in WIOD-SEA}: Austria, Belgium, Brazil, Bulgaria, Canada, China, Croatia, Cyprus, Czech Republic, Denmark, Estonia, Finland, France, Germany, Greece, Hungary, India, Indonesia, Ireland, Italy, Japan, Korea, Rep., Latvia, Lithuania, Luxembourg, Malta, Mexico, Netherlands, Norway, Poland, Portugal, Romania, Russian Federation, Slovak Republic, Slovenia, Spain, Sweden, Switzerland, Turkey, United Kingdom, United States.
    \item \textbf{Countries in WB-GPD}: Angola, Argentina, Australia, Austria, Azerbaijan, Belgium, Burkina Faso, Bangladesh, Bulgaria, Belize, Bolivia, Brazil, Botswana, Canada, Switzerland, Chile, China, Cameroon, Colombia, Costa Rica, Cyprus, Czech Republic, Germany, Denmark, Dominican Republic, Algeria, Ecuador, Egypt, Spain, Estonia, Ethiopia, Finland, Fiji, France, United Kingdom, Georgia, Ghana, Greece, Guatemala, China, Hong Kong SAR, Honduras, Croatia, Hungary, Indonesia, India, Ireland, Iran, Iceland, Italy, Jamaica, Jordan, Japan, Kenya, Republic of Korea, Lao People's Dem Rep, Saint Lucia, Sri Lanka, Lesotho, Lithuania, Luxembourg, Latvia, Morocco, Mexico, Montenegro, Mongolia, Mozambique, Mauritius, Malawi, Malaysia, Namibia, Nigeria, Netherlands, Norway, Nepal, New Zealand, Pakistan, Philippines, Poland, Portugal, Paraguay, Qatar, Romania, Russian Federation, Rwanda, Senegal, Singapore, Sierra Leone, Serbia, Slovakia, Slovenia, Sweden, Eswatini, Thailand, Turkey, Chinese Taipei, United Republic of Tanzania, Uganda, Uruguay, United States, St. Vincent and the Grenadines, Viet Nam, South Africa, Zambia.
\end{itemize}

\newpage

\section{Mathematical Derivation of Dynamic Innovation}

    \subsection{Evolution of the Productivity Frontier} In this section, we largely follow the steps of the mathematical appendix to \textcite{buera_global_2020} to the particularities of our model. In the proofs, for simplicity, we initially abstract away from sectoral specific elasticities, using $\theta$ rather than $\theta_i$, but at the end generalize the results to accommodate them. For any period, domestic technological frontier evolves according to:

    \begin{equation*}
        F_{d,t+\Delta}^{i}(z) = 
        \underbrace{F_{d,t}^{i}(z)}_{Pr\{\text{productivity}<z\text{ at } t\}} \times \underbrace{\Big( 1 - \int_t^{t+\Delta} \int \alpha_\tau z^{-\theta} (z')^{\beta\theta} dG_{d,\tau}^{i}(z') d\tau \Big)}_{Pr\{\text{no better draws in } (t,t+\Delta)\} }
    \end{equation*}
    
    Rearranging and using the definition of the derivative:
    
    \begin{equation*}
        \frac{d}{dt} \ln F_{s,t}^{i}(z) = \lim_{\Delta \to 0} \frac{F_{s,t+\Delta}^{i}(z) - F_{s,t}^{i}(z)}{F_{s,t}^{i}(z)} = -   \int \alpha_t z^{-\theta} (z')^{\beta\theta} dG_{d,t}^{i}(z')
    \end{equation*}
    
    Define $\lambda_{s,t}^{i} = \int_{-\infty}^t  \alpha_\tau  \int (z')^{\beta\theta} dG_{s,\tau}^{i}(z') d\tau$ and integrate both sides wrt to time:
    
    \begin{eqnarray*}
        \int_0^t \frac{d}{d\tau} \ln F_{s,\tau}^{i}(z) d\tau &=& - z^{-\theta}  \int_0^t \int \alpha_\tau (z')^{\beta\theta} dG_{d,\tau}^{i}(z') d\tau \\
        \ln \Big( \frac{F_{s,\tau}^{i}(z)}{F_{s,0}^{i}(z)} \Big) &=& - z^{-\theta}  (\lambda_{s,t}^{i} - \lambda_{s,0}^{i})     \\
        F_{s,t}^{i}(z) &=& F_{s,0}^{i}(z) \exp \{ - z^{-\theta}  (\lambda_{s,t}^{i} - \lambda_{s,0}^{i}) \}
    \end{eqnarray*}       
    
    Assuming that the initial distribution is Fr\'echet $F_{s,0}^{i}(z) = \exp \{ -\lambda_{s,0}^{i} z^{-\theta} \}$ guarantees that the distribution will be Fr\'echet at any point in time:
    
    \begin{equation}
        F_{s,t}^{i}(z) = \exp \{ -  \lambda_{s,t}^{i} z^{-\theta} \} 
    \end{equation}
    
    \subsection{Law of Motion of Productivity}
    
    As seen above, we have defined:
    
    \begin{equation*}
        \lambda_{s,t}^{i} = \int_{-\infty}^t  \alpha_\tau  \int (z')^{\beta\theta} dG_{s,\tau}^{i}(z') d\tau
    \end{equation*}
    
    Differentiating this definition with respect to time and applying Leibnitz's Lemma yields:
    
    \begin{equation*}
        \dot{\lambda}_{s,t}^{i} = \alpha_t  \int (z')^{\beta\theta} dG_{s,t}^{i}(z')
    \end{equation*}
    
    We use these results and work with a discrete approximation of the law of motion for productivity:
    
    \begin{equation}
        \Delta\lambda_{s,t}^{i} = \alpha_t  \int (z')^{\beta\theta} dG_{s,t}^{i}(z')
    \end{equation}
    
    The source distribution $G_{d,t}^{i}(z') \equiv \sum_{j \in \mathcal{I}} \eta_{d,t-1}^{i,j} \sum_{s \in \mathcal{D}}  H_{sd,t-1}^{i,j}(z')$, where $\eta_{d,t}^{i,j}$ is the expenditure share of sector $j$ in the cost of intermediates when producing good $i$ in region $d$; and  $H_{sd,t-1}^{i,j}(z')$ is the fraction of commodities for which the lowest cost supplier in period $t-1$ is a firm located in $s \in \mathcal{D}$ with productivity weakly less than $z'$. 

    We focus our attention on the integral $\int z^{\beta \theta} dH_{sd,t}^{i,j}(z)$. Let $F_{s,t}^{i}(z_2,z_2) = \exp \{ - \lambda_{s,t}^{i} z_2^{-\theta}  \}$ and $F_{s,t}^{i}(z_1,z_2) =(1+ \lambda_{s,t}^{i}[z_2^{-\theta}-z_1^{-\theta}] )  \exp \{ - \lambda_{s,t}^{i} z_2^{-\theta}  \}$ are, respectively, the probability that a productivity draw is below $z_2$, and that the maximum productivity is $z_1$ and the second highest productivity is $z_2$\footnote{To see the latter, note that:
    \begin{eqnarray*}
        Prob(z_1 \le Z_1, z_2 \le Z_2 ) &=& F_{s,t}^{i}(Z_2) + \int_0^{Z_2} \int_{Z_2}^{Z_1} f_{s,t}^{i}(y) f_{s,t}^{i}(y') dy' dy \\
        &=& F_{s,t}^{i}(Z_2) + F_{s,t}^{i}(Z_2) ( F_{s,t}^{i}(Z_1) - F_{s,t}^{i}(Z_2) ) \\
        &=& (1+ \lambda_{s,t}^{i}[Z_2^{-\theta}-Z_1^{-\theta}] )  \exp \{ - \lambda_{s,t}^{i} Z_2^{-\theta}  \}
        \end{eqnarray*}
        }. Let for each $n$, $ A^i_{n,t} \equiv  \nicefrac{\tilde{x}_{nd,t}^{i}}{\tilde{x}_{sd,t}^{i}}$, such that $s$ will have a lower cost than $d$ iff $A^i_{n,t} z_{n,t}^{i}(\omega) < z_{s,t}^{i}(\omega)$. Region $s$ with highest productivity producers $z_1, z_2$ will supply the commodity $i \in \mathcal{I}$ in region $d$ with the following probability:
    \begin{eqnarray*}
        \mathcal{F}^{i,j}_{sd,t-1}(z_1,z_2) &=& \int_0^{z_2} \Pi_{n \neq s} F^{j}_{n,t-1} \Big(A^i_{n,t} y, A^i_{n,t} y  \Big) d F^{j}_{s,t-1} (y,y) \\
        &+& \int_{z_2}^{z_1} \Pi_{n \neq s} F^{j}_{n,t-1} \Big(A^i_{n,t} z_2, A^i_{n,t} z_2  \Big) \frac{d}{dz_1} F_{s,t-1}^{j}(z_1,z_2)
    \end{eqnarray*}
    
    The first term in the right hand side denotes the probability that the lowest cost producer at destination $d$ is from $s$ and has productivity lower than $z_2$, while the second term accounts for the probability that the lowest cost producer at destination $d$ is from $s$ and has productivity in the range $[z_2,z_1)$. We will evaluate each integral separately. First, take the first term:
    
    \begin{eqnarray*}
    & & \int_0^{z_2} \Pi_{n \neq s} F^{j}_{n,t-1} \Big(A_{n,t-1}^{i} y, A_{n,t-1} y  \Big) d F^{j}_{s,t-1} (y,y) \\
    & & =  \int_0^{z_2}  \exp \Bigg\{ - \sum_{n \neq s} \lambda_{n,t-1}^{j} \Big(A_{n,t-1}^{i} y\Big)^{-\theta} \Bigg\} \theta \lambda_{s,t}^{m,j} y^{-\theta-1}  \exp \{ - \lambda_{s,t-1}^{j} y^{-\theta}  \} dy \\
    & & =  \lambda_{s,t-1}^{j} \int_0^{z_2}  \theta  y^{-\theta-1}  \exp \Bigg\{ - \sum_{n} \lambda_{n,t-1}^{j} \Big(A_{n,t-1}^{i} \Big)^{-\theta} y^{-\theta} \Bigg\}  dy \\
    & & =  \lambda_{s,t-1}^{j} \Bigg[ \frac{1}{\sum_{n} \lambda_{n,t-1}^{j} \Big(A_{n,t-1}^{i} \Big)^{-\theta}} \exp \Bigg\{ - \sum_{n} \lambda_{n,t-1}^{j} \Big(A_{n,t-1}^{i} \Big)^{-\theta} y^{-\theta} \Bigg\} \Bigg]_{y=0}^{y=z_2} \\
    & & =  \pi_{sd,t-1}^{i,j} \exp \Big\{ - \frac{\lambda_{s,t-1}^{j}}{\pi_{sd,t-1}^{i,j}} z_2^{-\theta} \Big\} 
    \end{eqnarray*}
    
    Now consider the second term.
    
    \begin{eqnarray*}
    & & \int_{z_2}^{z_1} \Pi_{n \neq s} F^{j}_{n,t-1} \Big(A^i_{n,t} z_2, A^i_{n,t} z_2  \Big) \frac{d}{dz_1} F_{s,t-1}^{j}(z_1,z_2) \\
    & & =  \int_{z_2}^{z_1}\exp \Bigg\{ - \sum_{n \neq s} \lambda_{n,t-1}^{i} \Big(A_{n,t-1}^{i} z_2\Big)^{-\theta} \Bigg\}  \theta \lambda_{s,t}^{m,j} z_1^{-\theta-1}   \exp \{ - \lambda_{s,t-1}^{j} z_2^{-\theta}  \} dz_1 \\
    & & =  \exp \Bigg\{ - \sum_{n} \lambda_{n,t-1}^{j} \Big(A_{n,t-1}^{i} z_2\Big)^{-\theta} \Bigg\}  \lambda_{s,t-1}^{j} \int_{z_2}^{z_1}  \theta  z_1^{-\theta-1}   dz_1\\
    & & =  \exp \Big\{ - \frac{\lambda_{s,t-1}^{j}}{\pi_{sd,t-1}^{i,j}} z_2^{-\theta} \Big\}   \lambda_{s,t-1}^{j} (z_2^{-\theta} - z_1^{-\theta} ) \\
    \end{eqnarray*}
    
    Therefore:
    
    \begin{eqnarray}
        \mathcal{F}^{i,j}_{sd,t-1}(z_1,z_2) = \exp \Big\{ - \frac{\lambda_{s,t-1}^{j}}{\pi_{sd,t-1}^{i,j}} z_2^{-\theta} \Big\} \Big( \pi_{sd,t-1}^{i,j} + \lambda_{s,t-1}^{j} (z_2^{-\theta} - z_1^{-\theta} ) \Big)
    \end{eqnarray}
    
    Note that:
    
    \begin{equation}\label{eq: intermediate_joint_density}
        \int z^{\beta \theta} dH_{sd,t}^{i,j}(z) = \int_0^{\infty} \int_{z_2}^\infty z_{1}^{\beta \theta} \frac{\partial^2 \mathcal{F}^{i,j}_{sd,t-1}(z_1,z_2)}{\partial z_1 \partial z_2} dz_1 dz_2
    \end{equation}
    
    and that we can calculate the joint density explicitly:
    
    \begin{eqnarray*}
        \frac{\partial^2 \mathcal{F}^{i}_{sd,t-1}(z_1,z_2)}{\partial z_1 \partial z_2} &=& \frac{\partial}{\partial z_2} \exp \Big\{ - \frac{\lambda_{s,t-1}^{j}}{\pi_{sd,t-1}^{i,j}} z_2^{-\theta} \Big\}  \theta \lambda_{s,t-1}^{j} z_1^{-\theta-1} \\
        &=& \frac{1}{\pi_{sd,t-1}^{i,j}} \exp \Big\{ - \frac{\lambda_{s,t-1}^{j}}{\pi_{sd,t-1}^{i,j}} z_2^{-\theta} \Big\}  (\theta \lambda_{s,t-1}^{j} z_1^{-\theta-1}) (\theta \lambda_{n,t-1}^{j} z_2^{-\theta-1})
    \end{eqnarray*}
    
    Plugging this into (\ref{eq: intermediate_joint_density}):
    
    \begin{eqnarray*}
       & & \int_0^{\infty} \int_{z_2}^\infty z_{1}^{\beta \theta} \frac{1}{\pi_{sd,t-1}^{i,j}} \exp \Big\{ - \frac{\lambda_{s,t-1}^{j}}{\pi_{sd,t-1}^{i,j}} z_2^{-\theta} \Big\}  (\theta \lambda_{s,t-1}^{j} z_1^{-\theta-1}) (\theta \lambda_{n,t-1}^{j} z_2^{-\theta-1}) dz_1 dz_2 \\
       & & = \int_0^{\infty} \frac{1}{\pi_{sd,t-1}^{i,j}} \exp \Big\{ - \frac{\lambda_{s,t-1}^{j}}{\pi_{sd,t-1}^{i,j}} z_2^{-\theta} \Big\}  (\theta \lambda_{s,t-1}^{j} z_2^{-\theta-1}) \lambda_{s,t-1}^{j} \int_{z_2}^\infty(\theta  z_1^{-\theta(1-\beta)-1})  dz_1 dz_2 \\
       & & = \int_0^{\infty} \frac{1}{\pi_{sd,t-1}^{i,j}} \exp \Big\{ - \frac{\lambda_{s,t-1}^{j}}{\pi_{sd,t-1}^{i,j}} z_2^{-\theta} \Big\}  (\theta \lambda_{s,t-1}^{j} z_2^{-\theta-1}) \lambda_{s,t-1}^{j} \frac{1}{1-\beta} z_2^{-\theta(1-\beta)} dz_2 
    \end{eqnarray*}
    
    Using a change of variables, let $\gamma \equiv \frac{\lambda_{s,t-1}^{i}}{\pi_{sd,t-1}^{i}} z_2^{-\theta}$, which implies that $d\gamma = - \theta \frac{\lambda_{s,t-1}^{i}}{\pi_{sd,t-1}^{i}} z_2^{-\theta-1}  dz$
    
    Replacing above:
    
    \begin{eqnarray*}
        & & (\lambda_{s,t-1}^{j})^{\beta}  (\pi_{sd,t-1}^{i,j})^{1-\beta} \frac{1}{1-\beta} \int_0^{\infty}  \exp \Big\{ - \gamma \Big\}  \eta^{(1-\beta)} d\gamma \\
        & & = (\lambda_{s,t-1}^{j})^{\beta}  (\pi_{sd,t-1}^{i,j})^{1-\beta}  \frac{1}{1-\beta} \Gamma(2-\beta) \\
        & & = (\lambda_{s,t-1}^{j})^{\beta}  (\pi_{sd,t-1}^{i,j})^{1-\beta} \Gamma(1-\beta) \qquad (\because \Gamma(y+1) = y \Gamma(y))
    \end{eqnarray*}
    
    Therefore, replacing into the law of motion for the location parameter of the Fr\'echet distribution:

    \begin{eqnarray*}
        \Delta \lambda_{d,t}^{i} &=& \alpha_{t} \int z^{\beta \theta} dG_{d,t}^{i}(z) \\
         &=& \alpha_{t} \sum_{j \in \mathcal{I}} \eta_{d,t-1}^{i,j} \sum_{s \in \mathcal{D}} \int z^{\beta \theta} dH_{sd,t-1}^{i,j}(z) \\ 
         &=& \alpha_{t} \Gamma(1-\beta) \sum_{j \in \mathcal{I}}  \eta_{d,t-1}^{i,j} \sum_{s \in \mathcal{D}} (\lambda_{s,t-1}^{j})^{\beta}  (\pi_{sd,t-1}^{i,j})^{1-\beta} 
    \end{eqnarray*}
    
    \noindent which is the same expression as in equation (\ref{eq: final_law_of_motion}).

    \newpage

    \section{Optimal Diffusion Levels}\label{appendix: optimaldiffusion}
    
    \subsection{Two-by-Two Economy}
    
    If a Benevolent Planner were to chose domestic trade shares to maximize idea diffusion to a given sector at the home economy, she would solve the following concave programming problem:
    
    \begin{equation}
        \max_{\{\pi_{h}^{i,i}, \pi_{h}^{i,-i}\}}   \footnotesize{\eta^{i}  [(\pi_{h}^{i,i})^{1-\beta}(\lambda_h^i)^{\beta} + (1-\pi_{h}^{i,i})^{1-\beta}(\lambda_f^i)^{\beta}] +  (1-\eta_{d}^{i}) [(\pi_{h}^{i,-i})^{1-\beta}(\lambda_h^{-i})^{\beta} + (1-\pi_{h}^{i,-i})^{1-\beta}(\lambda_f^{-i})^{\beta}]}
    \end{equation}
    
    For $\pi_{h}^{i,i}$, the first order condition satisfies:
    
    \begin{eqnarray*}
        \eta^i (1-\beta) [(\pi_{h}^{i,i})^{-\beta}(\lambda_h^i)^{\beta} - (1-\pi_{h}^{i,i})^{-\beta}(\lambda_f^i)^{\beta}] &=& 0 \\
    (\pi_{h}^{i,i})^{-\beta}(\lambda_h^i)^{\beta} &=& (1-\pi_{h}^{i,i})^{-\beta}(\lambda_f^i)^{\beta} \\
    (\pi_{h}^{i,i})^{\text{Diffusion Optimum}} &=& \frac{\lambda_h^i}{\lambda_f^i + \lambda_h^i}
    \end{eqnarray*}
    
    This result is the building block of the ratios that we express in Section 3. If we want to calculate the within sector ratio of total domestic trade expenditure, we can write:
    
    \begin{equation*}
        \Bigg( \frac{\eta^i \pi_h^{i,i}}{\eta^i (1-\pi_h^{i,i})} \Bigg)^{\text{Diffusion Optimum}} = \frac{\lambda_h^i}{\lambda_f^i + \lambda_h^i} \times \Big( \frac{\lambda_f^i}{\lambda_f^i + \lambda_h^i} \Big)^{-1} = \frac{\lambda_h^i}{\lambda_f^i}
    \end{equation*}
    
    Similarly, if we want to write a cross-sector ratio of total domestic trade expenditure shares, we can write:
    
    \begin{eqnarray*}
        \Bigg( \frac{\eta^i \pi_h^{i,i}}{(1-\eta^i) \pi_{h}^{i,-i}} \Bigg)^{\text{Diffusion Optimum}} = \underbrace{\frac{\eta^i}{1-\eta^i}}_{\text{cost share}} \times \underbrace{\frac{\lambda_h^i}{\lambda_h^{-i}}}_{\text{own-productivity}} \times \underbrace{ \Bigg( \frac{\lambda_h^i + \lambda_f^i}{\lambda_h^{-i} + \lambda_f^{-i}} \Bigg)^{-1} }_{\text{industry-wise productivity}}
    \end{eqnarray*}
    
    \noindent which is the same as equation (\ref{eq: planner}).
    
    \subsection{Multi-Sector, Multi-Region Economy}
    
    For each commodity $i$ in macrosector $j$, the Benevolent Planner maximizes:
    
    \begin{eqnarray}
    &    & \max_{\{\pi_{sd,t-1}^{i,j}\}_{j,i \in \mathcal{I}, s \in \mathcal{D}}} \sum_{j \in \mathcal{I}} \eta_{d,t-1}^{i,j} \sum_{s \in \mathcal{D}} (\pi_{sd,t-1}^{i,j})^{1-\beta}(\lambda_{s,t-1}^{j})^{\beta}  \\
        &s.t.& \forall (i,j) \in \mathcal{I} \times \mathcal{I} \qquad  \sum_{s \in \mathcal{D}} \pi_{sd,t-1}^{i,j} = 1 \nonumber
    \end{eqnarray}
    
    \noindent Let $\varphi$ be the Lagrange multiplier. Then, for each $(s,i,j)$ first order conditions satisfy:
    \begin{eqnarray*}
        (1-\beta) \eta_{d,t-1}^{i,j} (\pi_{sd,t-1}^{i,j})^{\beta}(\lambda_{s,t-1}^{j})^{\beta} &=& \varphi \\
        (\pi_{sd,t-1}^{i,j})^{\text{Diffusion Optimum}} &=& \varphi^{-\frac{1}{\beta}} [(1-\beta) \eta_{d,t-1}^{i,j}]^{\frac{1}{\beta}} \lambda_{s,t-1}^{j}
    \end{eqnarray*}
    
    \noindent using the constraint:
    
    \begin{equation*}
        \sum_{s \in \mathcal{D}} (\varphi^{-\frac{1}{\beta}} [(1-\beta) \eta_{d,t-1}^{i,j}]^{\frac{1}{\beta}} \lambda_{s,t-1}^{j}) = 1 \iff \varphi^{-\frac{1}{\beta}} = [(1-\beta) \eta_{d,t-1}^{i,j}]^{-\frac{1}{\beta}} ( \sum_{s \in \mathcal{D}}  \lambda_{s,t-1}^{j})^{-1}
    \end{equation*}

    \noindent Therefore:
    
    \begin{equation}
        (\pi_{sd,t-1}^{i,j})^{\text{Diffusion Optimum}} = \frac{\lambda_{s,t-1}^{j}}{\sum_{k \in \mathcal{D}}  \lambda_{k,t-1}^{m,j}}
    \end{equation}
    
If we want to calculate the within sector ratio of total domestic trade expenditure, we can write:
    
    \begin{equation*}
        \Bigg( \frac{\eta_{d,t-1}^{i,j} \pi_{sd,t-1}^{i,j} }{ \eta_{d,t-1}^{i,j} \pi_{nd,t-1}^{i,j} } \Bigg)^{\text{Diffusion Optimum}} = \frac{\lambda_{s,t-1}^{j}}{\sum_{k \in \mathcal{D}}  \lambda_{k,t-1}^{m,j}} \times \Big( \frac{\lambda_{n,t-1}^{j}}{\sum_{k \in \mathcal{D}}  \lambda_{k,t-1}^{m,j}} \Big)^{-1} = \frac{\lambda_{s,t-1}^{j}}{\lambda_{n,t-1}^{j}}
    \end{equation*}
    
    Similarly, if we want to write a cross-sector ratio of total domestic trade expenditure shares, we can write:
    
    \begin{eqnarray*}
        \Bigg( \frac{\eta_{d,t-1}^{i,j} \pi_{sd,t-1}^{i,j} }{ \eta_{d,t-1}^{i,p} \pi_{nd,t-1}^{i,p} } \Bigg)^{\text{Diffusion Optimum}} &=& \frac{\eta_{d,t-1}^{i,j}}{\eta_{d,t-1}^{i,p}} \times \frac{\lambda_{s,t-1}^{j}}{\lambda_{n,t-1}^{p}} \times \Big(  \frac{\sum_{k \in \mathcal{D}}  \lambda_{k,t-1}^{j}}{\sum_{k \in \mathcal{D}}  \lambda_{k,t-1}^{p}}\Big)^{-1} \\
    \end{eqnarray*}
    
    \noindent which is analogous to equation \ref{eq: planner}. The actual trade allocation under the multi-country, multi-sector framework satisfies:

    \begin{eqnarray*}
        \Bigg( \frac{\eta_{d,t-1}^{i,j} \pi_{sd,t-1}^{i,j} }{ \eta_{d,t-1}^{i,p} \pi_{nd,t-1}^{i,p} } \Bigg)^{\text{Actual Trade}} &=& \frac{\eta_{d,t-1}^{i,j} }{\eta_{d,t-1}^{i,p}} \times \frac{\lambda_{s,t-1}^{j} (\tilde{x}_{sd,t-1}^{j})^{-\theta} }{\lambda_{n,t-1}^{p} (\tilde{x}_{nd,t-1}^{p})^{-\theta}} \times \Big(  \frac{ \sum_{k \in \mathcal{D}} \lambda_{k,t}^{j} ( \tilde{x}_{kd,t}^{j})^{-\theta} }{ \sum_{k \in \mathcal{D}} \lambda_{k,t}^{p} ( \tilde{x}_{kd,t}^{p})^{-\theta} }\Big)^{-1} 
    \end{eqnarray*}    
    
\noindent which is analogous to equation \ref{eq: ft}. The ratio the actual trade allocation for the planner's allocation satisfies:

\begin{equation*}
    \aleph = \frac{(\tilde{x}_{sd,t-1}^{j})^{-\theta} }{ (\tilde{x}_{nd,t-1}^{p})^{-\theta}} \times \Big(  \frac{\sum_{k \in \mathcal{D}}  \lambda_{k,t-1}^{j}}{ \sum_{k \in \mathcal{D}} \lambda_{k,t}^{j} ( \tilde{x}_{kd,t}^{j})^{-\theta} } \Big) \times  \Big(  \frac{\sum_{k \in \mathcal{D}}  \lambda_{k,t-1}^{p}}{ \sum_{k \in \mathcal{D}} \lambda_{k,t}^{p} ( \tilde{x}_{kd,t}^{p})^{-\theta} }\Big)^{-1}
\end{equation*}

\newpage

\section{Other Mathematical Derivations}

    \subsection{Trade shares} In this model, since there are infinitely many varieties in the unit interval, the expenditure share of destination region $d \in \mathcal{D}$ on goods coming from source country $s \in \mathcal{D}$ converge to their expected values. Let $\pi_{sd,t}^{i}$ denote the share of expenditures of consumers in region $d \in \mathcal{D}$ on commodity $i \in \mathcal{I}^{m}$ coming from region $s \in \mathcal{D}$ and, let for each $n$, $(A^i_{n,t})^{-1} \equiv \nicefrac{\tilde{x}_{sd,t}^{i}}{\tilde{x}_{nd,t}^{i}}$. This share will satisfy:
    
    \begin{eqnarray}
           \pi_{sd,t}^{i} &=& Pr \Big( \frac{\tilde{x}_{sd,t}^{i}}{z_{s,t}^{i}(\omega)} < \min_{(n \neq s)} \Big\{ \frac{\tilde{x}_{nd,t}^{i}}{z_{n,t}^{i}(\omega)}  \Big\}  \Big) \nonumber \\
           &=& \int_{0}^{\infty} Pr(z_{s,t}^{i}(\omega) = z) Pr ( z_{n,t}^{i}(\omega) < z A^{i}_{n} ) dz  \nonumber \\
           &=& \int_{0}^{\infty} f_{s,t}^{i}(z)  \Pi_{(n\neq s)} F_{n,t}(z A^{i}_{n} ) dz  \nonumber \\
            &=&  \int_{0}^{\infty} \theta \lambda_{s,t}^{i} z^{-(1+\theta)} e^{ -( \sum_{n \in \mathcal{D}} \lambda_{n,t}^{i} (A^{i}_{n})^{-\theta})z^{-\theta} }  dz   \nonumber \\
            &=& \frac{\lambda_{s,t}^{i} (\tilde{x}_{sd,t}^{i})^{-\theta}}{\sum_{n \in \mathcal{D}} \lambda_{n,t}^{i} (\tilde{x}_{md,t}^{i})^{-\theta}}  \nonumber \\
            &=& \frac{\lambda_{s,t}^{i} (\tilde{x}_{sd,t}^{i})^{-\theta}}{\Phi_{d,t}^{i} }
    \end{eqnarray}
    
    Similarly, since countries use the same aggregate final goods as intermediate inputs, cost shares in intermediates for each supplying sector $j$ and region $s$ used in the production of good $i$ in region $d$ satisfies:
    
    \begin{equation}
        \pi_{sd,t}^{i,j} = \frac{\lambda_{s,t}^{j} (\tilde{x}_{sd,t}^{j})^{-\theta}}{\Phi_{d,t}^{j} }
    \end{equation}
    
    \noindent which are the same as expressed in (\ref{eq: trade_share_com}).

    \subsection{Price levels}
    
    Recall from equations ($\ref{pd}$) that the prices of commodities and intermediate goods can be expressed, respectively, as:
    \begin{eqnarray*}
        p_{d,t}^{i} &=& \Big[ \int_{[0,1]} p_{d,t}^{i}(\omega)^{1-\sigma}  d\omega \Big]^{\frac{1}{1-\sigma}}
    \end{eqnarray*}
    
    Let $\Omega_{sd,t}^{i}$ and $\Omega_{sd,t}^{i,j}$ denote the subsets of $\Omega = [0,1]$ for which the region $s \in \mathcal{D}$ is a supplier in destination region $d \in \mathcal{D}$. We can then rewrite price levels above as:

    \begin{eqnarray*}
        p_{d,t}^{i} &=& \Big[ \sum_{s \in \mathcal{D}} \int_{\Omega_{sd,t}^{i}} p_{d,t}^{i}(\omega)^{1-\sigma}  d\omega \Big]^{\frac{1}{1-\sigma}}     \end{eqnarray*}
    
    Similarly, we restate $\mathcal{F}^{i}_{sd,t}(z_1,z_2)$ and the analogous measure $\mathcal{F}^{i,j}_{sd,t}(z_1,z_2)$:
    
    \begin{eqnarray}
        \mathcal{F}^{i}_{sd,t}(z_1,z_2) &=& \exp \Big\{ - \frac{\lambda_{s,t}^{i}}{\pi_{sd,t}^{i}} z_2^{-\theta} \Big\} \Big( \pi_{sd,t}^{i} + \lambda_{s,t}^{i} (z_2^{-\theta} - z_1^{-\theta} ) \Big)     \end{eqnarray}
    
    \noindent which denote the fraction of varieties that $d$ purchases from $s$ with productivity up to $z_1$ and whose second best producer is not more efficient than than $z_2$.     Recall that, from the Bertrand competition assumption, we can write, for each variety $\omega$:
    
    \begin{equation*}
        p_{d,t}^{i}(\omega) = \min \Bigg\{ \frac{\sigma}{\sigma-1} \frac{\tilde{x}_{sd,t}^{i}}{z^{i}_{1s,t}(\omega)} ,  \frac{\tilde{x}_{sd,t}^{i}}{z^{i}_{2s,t}(\omega)}  \Bigg\}
    \end{equation*}
    
    So we can rewrite the equation $\int_{\Omega_{sd,t}^{i}} p_{d,t}^{i}(\omega)^{1-\sigma} d\omega$ in the following fashion:
    
    \small
    \begin{eqnarray*}
        & & \int_{\Omega_{sd,t}^{i}} p_{d,t}^{i}(\omega)^{1-\sigma} d\omega \\
        &=& \int_0^\infty  \int_{z_2}^\infty (p_{d,t}^{i})^{1-\sigma} \frac{\partial^2 \mathcal{F}^{i}_{sd,t}(z_1,z_2)}{\partial z_1 \partial z_2} dz_1 dz_2 \\
        &=& \int_0^\infty  \int_{z_2}^\infty \min \Bigg\{ \frac{\sigma}{\sigma-1} \frac{\tilde{x}_{sd,t}^{i}}{z_1} ,  \frac{\tilde{x}_{sd,t}^{i}}{z_2}  \Bigg\}^{1-\sigma} \frac{1}{\pi_{sd,t}^{i}} \exp \Big\{ - \frac{\lambda_{s,t}^{i}}{\pi_{sd,t}^{i}} z_2^{-\theta} \Big\}  (\theta \lambda_{s,t}^{i} z_1^{-\theta-1}) (\theta \lambda_{s,t}^{i} z_2^{-\theta-1}) dz_1 dz_2
    \end{eqnarray*}
    \normalsize
    
    With a change of variables, denote $\eta_1 \equiv \frac{\lambda_{s,t}^{i}}{\pi_{sd,t}^{i}} z_1^{-\theta}$ and $\eta_2 \equiv \frac{\lambda_{s,t}^{i}}{\pi_{sd,t}^{i}} z_2^{-\theta}$ and $d \eta_1 = - \frac{\theta \lambda_{s,t}^{i} z_1^{-\theta-1}}{\pi_{sd,t-1}^{i}} dz_1$,  $d \eta_2 = - \frac{\theta \lambda_{s,t}^{i} z_2^{-\theta-1}}{\pi_{sd,t-1}^{i}} dz_2$, which allows us to rewrite the equation above as:
    
    \small
    \begin{eqnarray*}
        & & \int_{\Omega_{sd,t}^{i}} p_{d,t}^{i}(\omega)^{1-\sigma} d\omega \\
        &=& \pi_{sd,t}^{i} \int_0^\infty  \int_{0}^{\eta_2} \min \Bigg\{ \frac{\sigma}{\sigma-1} \frac{\tilde{x}_{sd,t}^{i}}{z_1} ,  \frac{\tilde{x}_{sd,t}^{i}}{z_2}  \Bigg\}^{1-\sigma} \exp \Big\{ - \eta_2 \Big\}  d\eta_1 d\eta_2 \\
        &=& \pi_{sd,t}^{i} \Big(\frac{\lambda_{s,t}^{i}}{\pi_{sd,t}^{i}}\Big)^{-\frac{1-\sigma}{\theta}} (\tilde{x}_{sd,t}^{i})^{1-\sigma}  \int_0^\infty  \int_{0}^{\eta_2} \min \Bigg\{ \Big(\frac{\sigma}{\sigma-1}\Big)^{\theta} \eta_1 ,  \eta_2 \Bigg\}^{\frac{1-\sigma}{\theta}}  \exp \Big\{ - \eta_2 \Big\} d\eta_1 d\eta_2 \\
        &=& \pi_{sd,t}^{i} \Big(\frac{\lambda_{s,t}^{i}}{\pi_{sd,t}^{i}}\Big)^{-\frac{1-\sigma}{\theta}} (\tilde{x}_{sd,t}^{i})^{1-\sigma} \Bigg[  \int_0^\infty  \int_{\Big(\frac{\sigma}{\sigma-1}\Big)^{-\theta} \eta_2}^{\eta_2} \eta_2^{\frac{1-\sigma}{\theta}}  \exp \Big\{ - \eta_2 \Big\} d\eta_1 d\eta_2 \\
        &+& \int_0^\infty  \int_{0}^{\Big(\frac{\sigma}{\sigma-1}\Big)^{-\theta} \eta_2} \Big( \frac{\sigma}{\sigma-1}\Big)^{1-\sigma} \eta_1^{\frac{1-\sigma}{\theta}}  \exp \Big\{ - \eta_2 \Big\} d\eta_1 d\eta_2 \Bigg] \\
        &=& \pi_{sd,t}^{i} \Big(\frac{\lambda_{s,t}^{i}}{\pi_{sd,t}^{i}}\Big)^{-\frac{1-\sigma}{\theta}} (\tilde{x}_{sd,t}^{i})^{1-\sigma} \Big[ 1- \Big(\frac{\sigma}{\sigma-1}\Big)^{-\theta} + \frac{\theta}{1-\sigma+\theta} \Big(\frac{\sigma}{\sigma-1}\Big)^{-\theta}  \Big]  \cdot \int_0^\infty  \eta_2^{\frac{1-\sigma}{\theta}+1}  \exp \Big\{ - \eta_2 \Big\} d\eta_2 \\
        &=& \pi_{sd,t}^{i} \Big(\frac{\lambda_{s,t}^{i}}{\pi_{sd,t}^{i}}\Big)^{-\frac{1-\sigma}{\theta}} (\tilde{x}_{sd,t}^{i})^{1-\sigma} \Big[ 1- \Big(\frac{\sigma}{\sigma-1}\Big)^{-\theta} + \frac{\theta}{1-\sigma+\theta} \Big(\frac{\sigma}{\sigma-1}\Big)^{-\theta}  \Big]  \Gamma \Big(\frac{1-\sigma}{\theta}+2 \Big) \\
        &=& \pi_{sd,t}^{i} \Big(\frac{\lambda_{s,t}^{i}}{\pi_{sd,t}^{i}}\Big)^{-\frac{1-\sigma}{\theta}} (\tilde{x}_{sd,t}^{i})^{1-\sigma} \Big[ 1- \Big(\frac{\sigma}{\sigma-1}\Big)^{-\theta} + \frac{\theta}{1-\sigma+\theta} \Big(\frac{\sigma}{\sigma-1}\Big)^{-\theta}  \Big] \frac{1-\sigma+\theta}{\theta}  \Gamma \Big( \frac{1-\sigma+\theta}{\theta} \Big) \\
        &=& \Big[ 1 - \frac{\sigma-1}{\theta} + \frac{\sigma-1}{\theta} \Big(\frac{\sigma}{\sigma-1}\Big)^{-\theta} \Big] \cdot \Gamma \Big( \frac{1-\sigma+\theta}{\theta} \Big) \cdot \pi_{sd,t}^{i} \Big(\frac{\lambda_{s,t}^{i} (\tilde{x}_{sd,t}^{i})^{-\theta} }{\pi_{sd,t}^{i}}\Big)^{-\frac{1-\sigma}{\theta}} \\
        &=& \Big[ 1 - \frac{\sigma-1}{\theta} + \frac{\sigma-1}{\theta} \Big(\frac{\sigma}{\sigma-1}\Big)^{-\theta} \Big] \cdot \Gamma \Big( \frac{1-\sigma+\theta}{\theta} \Big) \cdot \pi_{sd,t}^{i} \Big(\frac{\lambda_{s,t}^{i} (\tilde{x}_{sd,t}^{i})^{-\theta} }{\pi_{sd,t}^{i}}\Big)^{-\frac{1-\sigma}{\theta}}  \\
        &=& \Big[ 1 - \frac{\sigma-1}{\theta} + \frac{\sigma-1}{\theta} \Big(\frac{\sigma}{\sigma-1}\Big)^{-\theta} \Big] \cdot \Gamma \Big( \frac{1-\sigma+\theta}{\theta} \Big) \cdot \pi_{sd,t}^{i} \Big(\sum_{n \in \mathcal{D}} \lambda_{n,t}^{i} (\tilde{x}_{nd,t}^{i})^{-\theta} \Big)^{-\frac{1-\sigma}{\theta}} 
    \end{eqnarray*}
    \normalsize
    
    Therefore: 
    
    \begin{eqnarray}
        p_{d,t}^{i} &=& \Big[ \sum_{s \in \mathcal{D}} \int_{\Omega_{sd,t}^{i}} p_{d,t}^{i}(\omega)^{1-\sigma}  d\omega \Big]^{\frac{1}{1-\sigma}} \nonumber \\
        p_{d,t}^{i} &=&\Big[ 1 - \frac{\sigma-1}{\theta} + \frac{\sigma-1}{\theta} \Big(\frac{\sigma}{\sigma-1}\Big)^{-\theta} \Big]^{\frac{1}{1-\sigma}} \cdot \Gamma \Big( \frac{1-\sigma+\theta}{\theta} \Big)^{\frac{1}{1-\sigma}} \cdot \Big(\sum_{n \in \mathcal{D}} \lambda_{n,t}^{i} (\tilde{x}_{nd,t}^{i})^{-\theta} \Big)^{-\frac{1}{\theta}} \cdot \Big[ \sum_{s \in \mathcal{D}} \pi_{sd,t}^{i} \Big]^{\frac{1}{1-\sigma}}  \nonumber \\
        p_{d,t}^{i} &=& \Big[ 1 - \frac{\sigma-1}{\theta} + \frac{\sigma-1}{\theta} \Big(\frac{\sigma}{\sigma-1}\Big)^{-\theta} \Big]^{\frac{1}{1-\sigma}} \cdot \Gamma \Big( \frac{1-\sigma+\theta}{\theta} \Big)^{\frac{1}{1-\sigma}} \cdot \Big(\sum_{n \in \mathcal{D}} \lambda_{n,t}^{i} (\tilde{x}_{nd,t}^{i})^{-\theta} \Big)^{-\frac{1}{\theta}}
    \end{eqnarray}
    
    Which is the same as (\ref{eq: price_com_solve}) after allowing the  elasticities to be sector-specific.
    
    \subsection{Marginal costs and profits}
    
    From equation (\ref{qd}) we can derive standard CES demand functions as:
    
    \begin{eqnarray}
        q_{d,t}^{i}(\omega) &=& \Big( \frac{p_{d,t}^{i}(\omega)}{p_{d,t}^{i}} \Big)^{-\sigma} \frac{e_{d,t}^{i}}{p_{d,t}^{i}} \\
        c_{d,t}^{i,j}(\omega) &=& \Big( \frac{p_{d,t}^{m,j}(\omega)}{p_{d,t}^{m,j}} \Big)^{-\sigma} \frac{e_{d,t}^{i,j}}{pc_{d,t}^{m,j}}
    \end{eqnarray}
    
    \noindent where $p_{d,t}^{i}$ satisfies equations (\ref{eq: price_com_solve}); $e_{d,t}^{i}$ denotes expenditure on commodity $i$ of macro-sector $m$ in country $d$; and $e_{d,t}^{i,j}$ denotes expenditure on intermediate input $j$ used in the production of commodity $i$ of macro-sector $m$ in country $d$.
    
    As in previous subsections of the Appendix, we will derive the expression for the marginal cost and mark-up for the production of variety $q_{d,t}^{i}(\omega)$ and state a corresponding expression for $c_{d,t}^{i,j}(\omega)$. The marginal cost of producing variety $\omega$ sourced in country $s$ and consumed in country $s$ is:
    
    \begin{equation*}
        \frac{\tilde{x}_{d,t}^{i}}{z_1(\omega)} q_{d,t}^{i}(\omega)
    \end{equation*}
    
    \noindent and total cost of varieties sourced in country $s$ and consumed in country $s$ can be expressed as:
    
    \begin{equation*}
        \int_{\Omega_{sd,t}^{i}} \frac{\tilde{x}_{d,t}^{i}}{z_1(\omega)} q_{d,t}^{i}(\omega)  d\omega = \int_{\Omega_{sd,t}^{i}} \frac{\tilde{x}_{d,t}^{i}}{z_1(\omega)} \Big( \frac{p_{d,t}^{i}(\omega)}{p_{d,t}^{i}} \Big)^{-\sigma} \frac{e_{d,t}^{i}}{p_{d,t}^{i}}  d\omega
    \end{equation*}
    
    As in the previous section of the Appendix, we let $\Omega_{sd,t}^{i}$ and $\Omega_{sd,t}^{i,j}$ denote the subsets of $\Omega = [0,1]$ for which the  region $s \in \mathcal{D}$ is a supplier in destination region $d \in \mathcal{D}$. We can then rewrite the integral above as:
    
    \small
    \begin{eqnarray*}
        & & \int_{\Omega_{sd,t}^{i}} \frac{\tilde{x}_{d,t}^{i}}{z_1(\omega)} \Big( \frac{p_{d,t}^{i}(\omega)}{p_{d,t}^{i}} \Big)^{-\sigma} \frac{e_{d,t}^{i}}{p_{d,t}^{i}}  d\omega \\
        &=& \tilde{x}_{d,t}^{i}  \frac{e_{d,t}^{i}}{(p_{d,t}^{i})^{1-\sigma}} \int_0^\infty  \int_{z_2}^\infty (z_1)^{-1} (p_{d,t}^{i})^{-\sigma} \frac{\partial^2 \mathcal{F}^{i}_{sd,t}(z_1,z_2)}{\partial z_1 \partial z_2} dz_1 dz_2 \\
        &=&  \tilde{x}_{d,t}^{i}   \frac{e_{d,t}^{i}}{(p_{d,t}^{i})^{1-\sigma}} \int_0^\infty  \int_{z_2}^\infty \frac{1}{z_1} \min \Bigg\{ \frac{\sigma}{\sigma-1} \frac{\tilde{x}_{sd,t}^{i}}{z_1} ,  \frac{\tilde{x}_{sd,t}^{i}}{z_2}  \Bigg\}^{-\sigma} \frac{1}{\pi_{sd,t}^{i}} \exp \Big\{ - \frac{\lambda_{s,t}^{i}}{\pi_{sd,t}^{i}} z_2^{-\theta} \Big\}  (\theta \lambda_{s,t}^{i} z_1^{-\theta-1}) (\theta \lambda_{s,t}^{i} z_2^{-\theta-1}) dz_1 dz_2
    \end{eqnarray*}
    \normalsize

    Once again, use a change of variables, denote $\eta_1 \equiv \frac{\lambda_{s,t}^{i}}{\pi_{sd,t}^{i}} z_1^{-\theta}$ and $\eta_2 \equiv \frac{\lambda_{s,t}^{i}}{\pi_{sd,t}^{i}} z_2^{-\theta}$ and $d \eta_1 = - \frac{\theta \lambda_{s,t}^{i} z_1^{-\theta-1}}{\pi_{sd,t-1}^{i}} dz_1$,  $d \eta_2 = - \frac{\theta \lambda_{s,t}^{i} z_2^{-\theta-1}}{\pi_{sd,t-1}^{i}} dz_2$, which allows U.S. to rewrite the equation above as:

    \small
    \begin{eqnarray*}
        & & \int_{\Omega_{sd,t}^{i}} \frac{\tilde{x}_{d,t}^{i}}{z_1(\omega)} \Big( \frac{p_{d,t}^{i}(\omega)}{p_{d,t}^{i}} \Big)^{-\sigma} \frac{e_{d,t}^{i}}{p_{d,t}^{i}}  d\omega \\
        &=&  \pi_{sd,t}^{i} \frac{e_{d,t}^{i}}{(p_{d,t}^{i})^{1-\sigma}}  \Big(\frac{\lambda_{s,t}^{i}}{\pi_{sd,t}^{i}}\Big)^{-\frac{1-\sigma}{\theta}} (\tilde{x}_{sd,t}^{i})^{1-\sigma} \int_0^\infty  \int_{0}^{\eta_2 } \eta_1^{\frac{1}{\theta}} \min \Bigg\{ \Big( \frac{\sigma}{\sigma-1} \Big)^{\theta} \eta_1 ,  \eta_2 \Bigg\}^{-\frac{\sigma}{\theta}} d\eta_1 d \eta_2 \\
        &=&  \pi_{sd,t}^{i} \frac{e_{d,t}^{i}}{(p_{d,t}^{i})^{1-\sigma}}  \Big(\frac{\lambda_{s,t}^{i}}{\pi_{sd,t}^{i}}\Big)^{-\frac{1-\sigma}{\theta}} (\tilde{x}_{sd,t}^{i})^{1-\sigma} \Bigg[  \int_0^\infty  \int_{\Big(\frac{\sigma}{\sigma-1}\Big)^{-\theta} \eta_2}^{\eta_2} \eta_1^{\frac{1}{\theta}}  \eta_2^{\frac{-\sigma}{\theta}}  \exp \Big\{ - \eta_2 \Big\} d\eta_1 d\eta_2 \\
        &+& \int_0^\infty  \int_{0}^{\Big(\frac{\sigma}{\sigma-1}\Big)^{-\theta} \eta_2} \Big( \frac{\sigma}{\sigma-1}\Big)^{-\sigma} \eta_1^{\frac{1-\sigma}{\theta}}  \exp \Big\{ - \eta_2 \Big\} d\eta_1 d\eta_2 \Bigg] \\
        &=&  \pi_{sd,t}^{i} \frac{e_{d,t}^{i}}{(p_{d,t}^{i})^{1-\sigma}}  \Big(\frac{\lambda_{s,t}^{i}}{\pi_{sd,t}^{i}}\Big)^{-\frac{1-\sigma}{\theta}} (\tilde{x}_{sd,t}^{i})^{1-\sigma} \Bigg[  \int_0^\infty  \frac{\theta}{1+\theta} \Big[ 1 -  \Big(\frac{\sigma}{\sigma-1}\Big)^{-1-\theta} \Big] \eta_2^{\frac{1-\sigma+\theta}{\theta}}  \exp \Big\{ - \eta_2 \Big\} d\eta_2 \\
        &+& \int_0^\infty  \frac{\theta}{1-\sigma+\theta} \Big(\frac{\sigma}{\sigma-1}\Big)^{-1-\theta} \eta_1^{\frac{1-\sigma+\theta}{\theta}}  \exp \Big\{ - \eta_2 \Big\} d\eta_2 \Bigg] \\
        &=&  \pi_{sd,t}^{i} \frac{e_{d,t}^{i}}{(p_{d,t}^{i})^{1-\sigma}}  \Big(\frac{\lambda_{s,t}^{i}}{\pi_{sd,t}^{i}}\Big)^{-\frac{1-\sigma}{\theta}} (\tilde{x}_{sd,t}^{i})^{1-\sigma} \Bigg[ \frac{\theta}{1+\theta} \Big[ 1 -  \Big(\frac{\sigma}{\sigma-1}\Big)^{-1-\theta} \Big] + \frac{\theta}{1-\sigma+\theta} \Big(\frac{\sigma}{\sigma-1}\Big)^{-1-\theta} \Bigg] \Gamma \Big( \frac{1-\sigma}{\theta} + 2 \Big) \\
        &=&  \pi_{sd,t}^{i} \frac{e_{d,t}^{i}}{(p_{d,t}^{i})^{1-\sigma}}  \Big(\frac{\lambda_{s,t}^{i}}{\pi_{sd,t}^{i}}\Big)^{-\frac{1-\sigma}{\theta}} (\tilde{x}_{sd,t}^{i})^1{-\sigma} \Bigg[ \frac{\theta}{1+\theta} \Big[ 1 -  \Big(\frac{\sigma}{\sigma-1}\Big)^{-1-\theta} \Big] + \frac{\theta}{1-\sigma+\theta} \Big(\frac{\sigma}{\sigma-1}\Big)^{-1-\theta} \Bigg] \frac{1-\sigma+\theta}{\theta} \Gamma \Big( \frac{1-\sigma+\theta}{\theta} \Big) \\
        &=& \Big[ 1 - \frac{\sigma}{1+\theta} + \frac{\sigma}{1+\theta} \Big( \frac{\sigma}{\sigma-1} \Big)^{-1-\theta} \Big] \Gamma \Big( \frac{1-\sigma+\theta}{\theta} \Big)   \pi_{sd,t}^{i} \frac{e_{d,t}^{i}}{(p_{d,t}^{i})^{1-\sigma}}  \Big(\frac{\lambda_{s,t}^{i}}{\pi_{sd,t}^{i}}\Big)^{-\frac{1-\sigma}{\theta}} (\tilde{x}_{sd,t}^{i})^{1-\sigma}  \\
        &=& \Big[ 1 - \frac{\sigma}{1+\theta} + \frac{\sigma}{1+\theta} \Big( \frac{\sigma}{\sigma-1} \Big)^{-1-\theta} \Big] \Gamma \Big( \frac{1-\sigma+\theta}{\theta} \Big)   \pi_{sd,t}^{i} \frac{e_{d,t}^{i}}{(p_{d,t}^{i})^{1-\sigma}} \Big(\frac{(\tilde{x}_{sd,t}^{i})^{-\theta}\lambda_{s,t}^{i}}{\pi_{sd,t}^{i}}\Big)^{\frac{\sigma}{\theta}} \Big(\frac{(\tilde{x}_{sd,t}^{i})^{-\theta}\lambda_{s,t}^{i}}{\pi_{sd,t}^{i}}\Big)^{-\frac{1}{\theta}}  \\
        &=& \Big[ 1 - \frac{\sigma}{1+\theta} + \frac{\sigma}{1+\theta} \Big( \frac{\sigma}{\sigma-1} \Big)^{-1-\theta} \Big] \Gamma \Big( \frac{1-\sigma+\theta}{\theta} \Big)   \pi_{sd,t}^{i} \frac{e_{d,t}^{i}}{(p_{d,t}^{i})^{1-\sigma}} \Big( \sum_{n \in \mathcal{D}}(\tilde{x}_{nd,t}^{i})^{-\theta}\lambda_{n,t}^{i} \Big)^{-\frac{1-\sigma}{\theta}}
    \end{eqnarray*}
    \normalsize
    
    Using the expression for $(p_{d,t}^{i})^{1-\sigma}$:
    
    \small
    \begin{eqnarray*}
        &=& \frac{ \Big[ 1 - \frac{\sigma}{1+\theta} + \frac{\sigma}{1+\theta} \Big( \frac{\sigma}{\sigma-1} \Big)^{-1-\theta} \Big] \Gamma \Big( \frac{1-\sigma+\theta}{\theta} \Big)   \pi_{sd,t}^{i} e_{d,t}^{i} \Big( \sum_{n \in \mathcal{D}}(\tilde{x}_{nd,t}^{i})^{-\theta}\lambda_{n,t}^{i} \Big)^{-\frac{1-\sigma}{\theta}} }{ \Big[ 1 - \frac{\sigma-1}{\theta} + \frac{\sigma-1}{\theta} \Big(\frac{\sigma}{\sigma-1}\Big)^{-\theta} \Big] \Gamma \Big( \frac{1-\sigma+\theta}{\theta} \Big) \Big(\sum_{n \in \mathcal{D}} \lambda_{n,t}^{i} (\tilde{x}_{nd,t}^{i})^{-\theta} \Big)^{-\frac{1-\sigma}{\theta}}} \\ 
        &=& \frac{ \Big[ 1 - \frac{\sigma}{1+\theta} + \frac{\sigma}{1+\theta} \Big( \frac{\sigma}{\sigma-1}  \Big)^{-1-\theta} \Big]}{ \Big[ 1 - \frac{\sigma-1}{\theta} + \frac{\sigma-1}{\theta} \Big(\frac{\sigma}{\sigma-1}\Big)^{-\theta} \Big] } \pi_{sd,t}^{i} e_{d,t}^{i} \\
        &=& \frac{\theta}{1+\theta} \frac{ 1 + \theta - \sigma + \sigma^{-\theta}(\sigma-1)^{1+\theta} }{ 1 + \theta - \sigma + \sigma^{-\theta}(\sigma-1)^{1+\theta} } \pi_{sd,t}^{i} e_{d,t}^{i} 
        \end{eqnarray*}
    \normalsize
    
    Therefore, total cost equals:
    
    \begin{equation}
        C_{s,t}^{i} = \sum_{d \in \mathcal{D}} \int_{\Omega_{sd,t}^{i}} \frac{\tilde{x}_{d,t}^{i}}{z_1(\omega)} \Big( \frac{p_{d,t}^{i}(\omega)}{p_{d,t}^{i}} \Big)^{-\sigma} \frac{e_{d,t}^{i}}{p_{d,t}^{i}}  d\omega = \frac{\theta}{1+\theta} \sum_{d \in \mathcal{D}} \pi_{sd,t}^{i} e_{d,t}^{i} 
    \end{equation}
    
    Profits can be expressed compactly as total revenue minus total cost:
    
    \begin{equation}
        \Pi_{s,t}^{i} = \sum_{d \in \mathcal{D}} \pi_{sd,t}^{i} e_{d,t}^{i} - \frac{\theta}{1+\theta} \sum_{d \in \mathcal{D}} \pi_{sd,t}^{i} e_{d,t}^{i}  = \frac{1}{1+\theta} \sum_{d \in \mathcal{D}} \pi_{sd,t}^{i} e_{d,t}^{i} 
    \end{equation}
    
    Analogously, total costs and profits of intermediary producers are, respectively:
    
    \begin{eqnarray}
        c_{s,t}^{i} = \frac{\theta}{1+\theta} \sum_{d \in \mathcal{D}} \pi_{sd,t}^{i,j} e_{d,t}^{i,j}, \qquad 
        \Pi_{s,t}^{i} = \frac{1}{1+\theta} \sum_{d \in \mathcal{D}} \pi_{sd,t}^{i,j} e_{d,t}^{i,j} 
    \end{eqnarray}
    
    These are allow analogous to the expression in the paper after allowing the  elasticities to be sector-specific.

\end{document}